\documentclass[%
 reprint,
 amsmath,
 amssymb,
 aps,
 pre,
]{revtex4-2}

\usepackage{hyperref}
\usepackage{graphicx}
\usepackage{dcolumn}
\usepackage{bm}
\usepackage{mathrsfs,amsfonts,amsmath,amssymb,amsthm,bbm,chemarr}
\usepackage{stmaryrd} 
\usepackage[lofdepth,lotdepth]{subfig}
\usepackage{mathtools}

\newcommand{\smallO}[1]{\ensuremath{\mathop{}\mathopen{}{\scriptstyle\mathcal{O}}\mathopen{}\left(#1\right)}}

\newcommand{\reels}{\mathbb{R}} 
\newcommand{\relatifs}{\mathbb{Z}}
\newcommand{\esp}{\mathbb{E}} 
\newcommand{\proba}{\mathbb{P}}
\newcommand{\espe}[1]{\mathbb{E}\left[#1\right]}

\newtheorem{proposition}{Proposition}

\newtheorem{conjecture}{Conjecture}
\newtheorem{assumption}{Assumption}

\newtheorem{definition}{Definition}
\newtheorem{lem}{Lemma}
\newtheorem{remarque}{Remark}
\newtheorem{nota}{Notation}
\newtheorem{consequence}{Consequence}
\newtheorem{mainresult}{Main Result}
\renewcommand{\theassumption}{A\arabic{assumption}}

\newcommand{\supp}{{\mathrm{supp}}}

\DeclareRobustCommand\circled[1]{\tikz[baseline=(char.base)]{
		\node[shape=circle,draw,inner sep=2pt] (char) {#1};}}

\usepackage{todonotes}  \newboolean{DisplaySOS}
\setboolean{DisplaySOS}{true} 

\graphicspath{{images/}}

\begin{document}

\preprint{APS/123-QED}

\title{Mean-field analysis of a neural network with stochastic STDP}

\author{Pascal Helson}
\email{pascal.helson@u-bordeaux.fr}
\affiliation{Univ. Bordeaux, CNRS, IMN, UMR 5293, F-33000 Bordeaux, France}

\author{Etienne Tanré}%
\affiliation{Université Côte d'Azur, Inria, CNRS, LJAD, France}%
\author{Romain Veltz}%
\affiliation{Université Côte d’Azur, Inria, France}

\begin{abstract}
Synaptic plasticity is the biological foundation of learning and memory and a key inspiration for efficient, spike-based neuromorphic computing. In both contexts, spike-timing-dependent plasticity (STDP) is a crucial adaptive rule. However, while analytical tools exist for spiking networks with fixed connections, they fail when STDP is introduced, creating a long-standing theoretical impasse.
In this work, we demonstrate how mean-field theory can successfully bridge this gap. We derive the first McKean-Vlasov mean-field limit for a network of stochastic spiking units with discrete-state STDP synapses, providing a rigorous, low-dimensional description of a system previously considered analytically intractable due to strong heterogeneity and adaptive couplings.
While motivated by neuroscience -- introducing plasticity into a stochastic Wilson-Cowan model -- our framework establishes a general theoretical tool for studying the collective dynamics of large systems of interacting spiking (as opposed to rate-based) units with adaptation. This result provides a new methodological avenue in statistical physics for analysing a wide class of complex systems with plastic interactions.

\end{abstract}

\maketitle

During the last few decades, synaptic plasticity has been widely studied and its importance in the study of neural networks is far from being completely understood. Many fascinating
questions remain open regarding network structure formation \cite{ocker2015self,maistrenko_multistability_2007},
stability \cite{mongillo2017intrinsic,ocker2019training} and function \cite{brzosko2019neuromodulation}. 
From a modelling perspective, one of the most studied plasticity rules is the so-called spike-timing-dependent plasticity (STDP), in which synaptic changes depend on
the relative timing of pre- and post-synaptic spikes. However, biological neural network models with STDP pose challenges for both theoretical and numerical analyses. The numerical bottleneck arises from 
the $N^2$ scaling of synapses' number when a large number $N$ of neurons is considered. Furthermore, the intricate coupling between the dynamics of the spiking
neurons and the synapses, along with plasticity-induced heterogeneity poses many difficulties.

One approach to address
these challenges is to assume that plasticity is very slow compared to neural activity, and thus to leverage slow-fast theory. In this framework, synaptic updates triggered by spike pairs can be modelled either deterministically
or stochastically. In the deterministic case, each pair leads to a tiny
(infinitesimal) increase in weights \cite{kempter1999hebbian} that usually depends on the time between spikes according to biological experiments of 
STDP \cite{bi1998synaptic}; see \cite{robert2021stochastic,robert2021stochastic_scaling} 
for a mathematical formulation and 
scaling analysis of such rules on one synapse dynamics. Consequently, we observe a significant (macroscopic) change of the weights only after several spike pairs. 
In the stochastic case, by contrast, every pair of spikes can lead to a macroscopic weight change with a probability depending on the spike time differences~\cite{o2005graded}.  
In this setting, the slow property of plasticity arises from the small probability of update associated to each spike pair~\cite{appleby2005synaptic,appleby2006stable,helson2017new,helson2021plasticity}. 
While these models show similarities when averaging across realisations, their mathematical structures differ and thus require distinct analytical tools (deterministic versus probabilistic).

For both stochastic and deterministic cases, results remain incomplete when considering neural networks with STDP. Indeed, they involve intractable quantities; 
respectively the stationary distribution of the neural network \cite{helson2017new,helson2021plasticity} and complex spike statistics in heterogeneous 
networks \cite{ocker2015self}.
In addition, both models still explicitly involve $N$ neurons, 
meaning that no model reduction has been achieved for such STDP-based interactions yet. A natural approach to do so is to derive a mean-field (MF) description. More generally, MF theory aims at replacing a high-dimensional interacting particle system by an effective limit equation describing the behaviour of a typical unit in the large population limit. Such reductions have become one of the main mathematical tools for analysing and predicting the collective behaviour of large neural systems \cite{la2021mean}.
For completeness, a MF limit was derived in a spiking network assuming slow plasticity depending on the global activity (in contrast to pairwise activity dependence in STDP) and using probabilistic tools \cite{galtier2012multiscale} or deterministic tools (PDE analysis on the Fokker-Planck equations) \cite{perthame2017distributed}.

Another limitation of the slow-fast framework is its reliance on the controversial assumption that plasticity dynamics is slower than neural dynamics. Indeed, it has been shown that synaptic plasticity 
can occur at very different timescales. For example, the shrinkage and enlargement of dendritic spines are respectively on the order of ten minutes and one minute \cite{kasai2021spine}. In addition, receptor resources (responsible for the strength of a synaptic connection) can get in/out the cell on the order of 
milliseconds \cite{choquet2013dynamic}. On the pre-synaptic side, bouton formation occurs on the order of a few minutes \cite{vasin2019two}. 
Finally, it was recently shown that the Drosophila mushroom body is segmented in different clusters where plasticity has different timescales regulated 
by dopamine; see \cite{aso2016dopaminergic} for the original paper and \cite{fulton2024common} for a review of recent advances on this topic. 
All of this points toward multiple plasticity timescales within the brain and not only slow ones compared to neural dynamics \cite{rodrigues_stochastic_2023}.

As a consequence, the slow-fast assumption is not universally applicable in the brain~\cite{zenke2017temporal,pechersky2017stochastic,lansner2023fast}, and it is unclear how to
analyse such complex systems -- e.g. plastic spiking neural networks -- without relying on the latter. 
While STDP has been implemented numerically
in many network models~\cite{masquelier2008spike,masquelier2009competitive,gilson2010emergence, gilson2011stability,ocker2015self}, rigorous mathematical
analyses of these systems remain scarce. In particular,
no tractable reduced model capable of capturing
both plasticity and heterogeneity has yet been established.
One natural candidate is MF theory, which
has successfully been applied to several interacting neural
network models without plasticity~\cite{montbrio2015macroscopic,de2015hydrodynamic, delarue2015particle,fournier2016toy,delattre2016statistical,chevallier2017mean,farkhooi2017complete}.

However, when dealing with spiking networks with plastic interactions, traditional MF theory tools become ineffective
due to the heterogeneity and dynamic nature of these interactions. As a result, a belief has emerged suggesting that MF theory is unsuitable for analysing networks
with STDP~\cite{huang2023new,lorenzi2023multi}. Despite this, a MF limit can still be derived by making certain simplifications, such as considering short-term plasticity \cite{di_volo_synchronous_2013,gast_mean-field_2021} or assuming a symmetric STDP curve, and
employing theories like the Ott-Antonsen ansatz \cite{duchet2023mean} to derive MF approximations. Additionally, it has been suggested that MF theory does not account for spike correlations as it
dilutes the impact of individual spike pairs by a factor of $\frac{1}{N}$ and mainly depends on the mean of synaptic weights \cite{robert2021spontaneous}. Thereby, MF theory does not account for the synaptic weights' heterogeneity \cite{farkhooi2017complete}. This explains in part the lack of results in this area. Another reason is that MF limits have been derived using ansatz (like the famous Ott-Antonsen ansatz \cite{montbrio2015macroscopic,duchet2023mean}) or trying to find closed equations on the first order statistics of the variable of interest; the latter methods break down when considering heterogeneous and plastic networks.

Here, we use another approach based on deriving the limit equation for the empirical distribution of the variable of interest \cite{sznitman_topic}; we call it the McKean-Vlasov mean-field (MKV-MF) limit. Hence, our analysis marks the first
exploration of MKV-MF dynamics in a spiking neural network composed of interacting neurons with STDP. In contrast to previous MF approaches, synaptic weights remain heterogeneous and the explicit tracking of spike correlations is unnecessary in our framework. More specifically, we study in this article a probabilistic Wilson-Cowan neural network model in which neurons elicit spikes and interact through a stochastic STDP rule. The approach we developed is broadly applicable and may be adapted to other neural network models, for example featuring more biophysically realistic membrane potential dynamics \cite{veltz2025analysis}.

The proposed framework is not limited to the study of STDP; it can also be used to derive MF models of other systems composed of interacting units such as spiking lasers \cite{dolcemascolo2020effective}, Ising models with plastic interactions \cite{pechersky2017stochastic} or epidemic processes on complex networks \cite{pastor2015epidemic}. In addition, it opens the door to new mathematical questions such as establishing the uniqueness of the solution to an equation similar to the limit system we derived
and confirming the convergence to a deterministic limit distribution solution of the latter \cite{sznitman_topic}. Moreover, studying the limit system would provide insight into the initial model. In particular, this study opens new avenues for preventing weight divergence without relying on soft or hard bounds. In addition, our approach significantly reduces simulation costs, crucial for models involving synaptic plasticity in neural networks. A pertinent example is the recent development of adaptive deep brain stimulation (DBS). Therefore, modelling the effects of DBS on the synaptic weights of stimulated neurons using STDP could aid in identifying stimulation patterns that lead to long-term changes \cite{duchet2023mean}. The potential impact is significant as DBS is used to alleviate symptoms in many brain diseases like movement disorders, depression, obsessive-compulsive disorders (OCD), Alzheimer's disease and epilepsy \cite{lozano2019deep}.

The paper is organised as follows. In section \ref{sec:Presentation of the model}, we first present the microscopic model (initial neural network description) before presenting the first hints on the macroscopic model. In particular, we introduce the \textit{new} variables describing the neural network from which we perform a MF analysis in section \ref{sec:Study of the empirical measure and its limit}. Hence, the latter is devoted to the derivation of the limit equation from the analysis of the empirical distribution of the \textit{new} neural network description. 
Thereby, a McKean-Vlasov equation \cite{mckean1966class} is derived on a typical neuron which is composed of: the neuron state, the time from its last
spike and the distribution of the triplet composed of the neuron state, the time since its last spike and its pre-synaptic weights. We finally perform numerical simulations on these last equations in section \ref{sec:Simulations} and compare the resulting activity to the finite-size system (ground truth).

\section{Presentation of the model}\label{sec:Presentation of the model}
We first present the model before deriving its new description based on the typical neuron definition.
\subsection{Microscopic model}\label{sec:micro_model}
We study a network of $N$ binary neurons, $V_t^{i,N} \in \{0,1\}$, modelled with jump processes and all-to-all connected via the matrix of discrete synaptic weights \cite{o2005graded,abarbanel_synaptic_2005}, 
$W_t^N \in \relatifs^{N^2}$. 
Neuron $i$ is said to \textit{spike} when its membrane potential, $V_t^{i,N}$, 
{jumps from $0$ to $1$}. This occurs at rate $\alpha(I_t^{i,N})$ where
\begin{equation}\label{eq:20210215}
	I_t^{i,N} \coloneqq
	\frac{1}{N} \sum_j W_{t}^{ij,N} V_{t}^{j,N}
\end{equation}
is the incoming synaptic current into neuron $i$ at time $t$. The function $\alpha$ is positive and bounded (usually a sigmoid function).

We denote by $S_t^{i,N}\geq0$ the time spent since the last spike of the neuron $i$. We are interested in the piecewise deterministic Markov process (PDMP) \cite{davis1984piecewise} such that for all $i$ the dynamics between jumps is
\begin{itemize}
    \item $dS_t^{i,N}=dt$.
\end{itemize}
The jumps are
\begin{itemize}
\item First, the jumps of $(V_t^{i,N},S_t^{i,N})$
\begin{align*}
(0,S_t^{i,N}) &\overset{\alpha(I_{t}^{i,N})}{\longrightarrow}  (1,0),\\
(1,S_t^{i,N}) &\overset{\beta > 0}{\longrightarrow}  (0,S_t^{i,N}),
\end{align*}
    \item \textbf{at a spike time \(\tau_{i_0}\) of neuron \(i_0\)}, we have 
    \(V_{{\tau_{i_0}}^-}^{i_0,N} = 0\), \(V_{{\tau_{i_0}}}^{i_0,N} = 1\) and $S_{{\tau_{i_0}}}^{i_0,N}=0$. 
    The outgoing weights $W_{{\tau_{i_0}}^-}^{ji_0,N}$ and incoming weights $ W_{{\tau_{i_0}}^-}^{i_0j,N}$ evolve via independent jumps for each $j$, with outgoing weights decreasing while incoming weights increasing according to probabilities that depend on the time of the last spike of neuron $j$. It thus associates the weight dynamics to the pairing of the last spike of neurons $i_0$ and $j$ as follows:
    \begin{align*}
    &\proba\left(W_{{\tau_{i_0}}}^{ji_0,N}  = W_{{\tau_{i_0}}^-}^{ji_0,N} - 1\right) = p^-(S_{{\tau_{i_0}}^-}^j,W_{{\tau_{i_0}}^-}^{ji_0,N}),
    \\
    &\proba\left(W_{{\tau_{i_0}}}^{ji_0,N} = W_{{\tau_{i_0}}^-}^{ji_0,N}\right) \quad \ \ = 1- p^-(S_{{\tau_{i_0}}^-}^j,W_{{\tau_{i_0}}^-}^{ji_0,N}),
    \\
    &\proba\left(W_{{\tau_{i_0}}}^{i_0j,N} = W_{{\tau_{i_0}}^-}^{i_0j,N} + 1\right)  = p^+(S_{{\tau_{i_0}}^-}^j,W_{{\tau_{i_0}}^-}^{i_0j,N}),
    \\
    &\proba\left(W_{{\tau_{i_0}}}^{i_0j,N} = W_{{\tau_{i_0}}^-}^{i_0j,N}\right) \quad \ \ = 1- p^+(S_{{\tau_{i_0}}^-}^j,W_{{\tau_{i_0}}^-}^{i_0j,N}).
    \end{align*}
\end{itemize}
where $p^{\pm}$ are functions (representing jump probabilities) from $\reels^+ \times \relatifs$ to $[0,1]$. 
Hence, the synaptic weights' dynamics depend -- stochastically -- on the elapsed time since the last spikes of the neurons and in particular, it takes into account only the last pair of spikes which makes the rule a pairwise stochastic STDP. Its deterministic counterpart would make the weights jump of an amount proportional to $p^\pm(s,w)$ for every new pair of (last) spikes.
Such dynamics is illustrated in Figure \ref{fig:model_schematic}.

Thus, the firing rate is given by \(\alpha(I_t^{i,N})\) and the neurons return to their resting potential $0$ at constant rate $\beta > 0$. The weight $\frac{1}{N}W_t^{ij,N}$ represents the effect of a neuron $j$ on the neuron $i$ at time $t$.
Note that we do not assume that the diagonal elements \( W_t^{ii,N} \) are null. Instead, we assume that they follow dynamics similar to that of \( W_t^{ij,N} \) for \( i \neq j \). 
This means that whenever neuron \( i \) spikes, \( W_t^{ii,N} \) can either stay the same, increase, or decrease. Regardless of these changes, these weights do not affect the dynamics because they appear in the firing rates of neuron \( i \) only as \( W^{ii}_t V^i_t \). Since, when potentially firing, neuron \( i \) is at rest, \( V^i_t = 0 \), it makes the latter product null.

\begin{figure*}
    \includegraphics[width=0.8\textwidth]{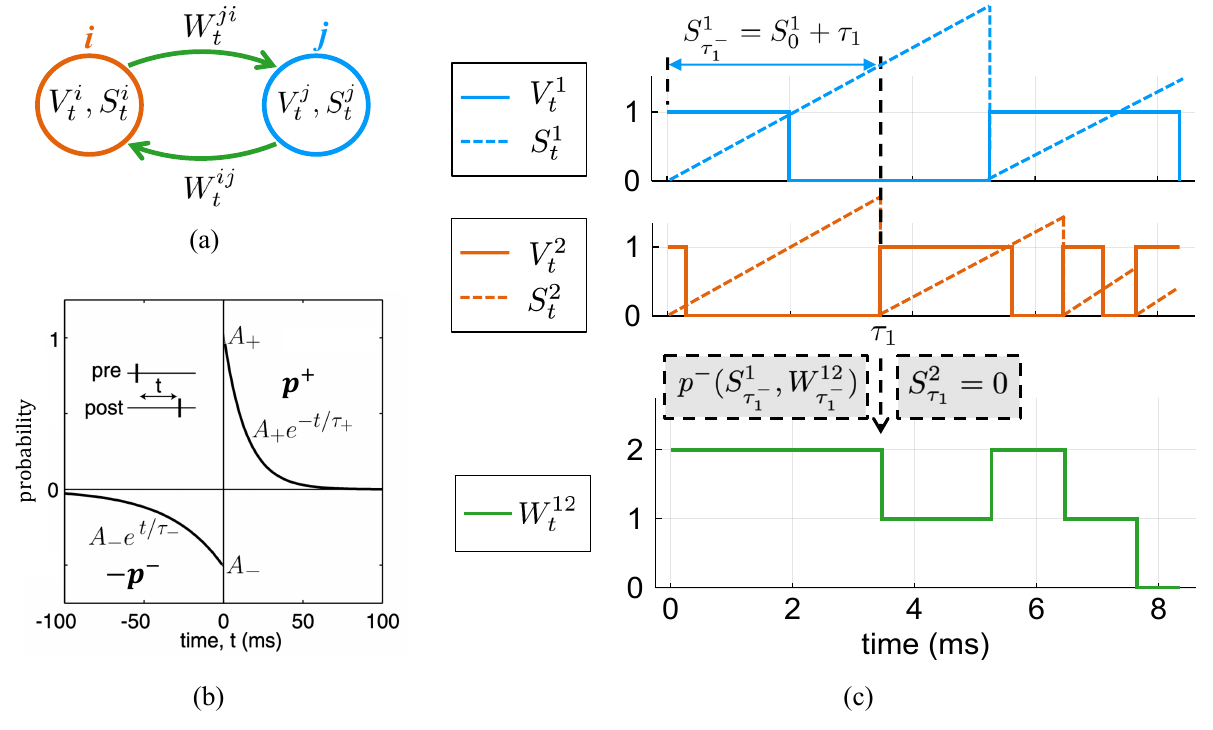}
    \caption{\textbf{Key model components and their dynamics.}
        (1a): Neurons $i$ and $j$ are described by their potential, $V_t^{i}$ and $V_t^{j}$, time since the last spike $S_t^{i}$ and $S_t^{j}$ and connection weights, $W_t^{ij}(=W_t^{i \gets j})$ and $W_t^{ji}(=W_t^{j \gets i})$.
        (1b): Example of the STDP functions $p^\pm$ governing the plasticity mechanism. Adapted from \cite{izhikevich2003relating}.
        (1c): Membrane potentials $V$ are binary. Every time a neuron spikes ($V$-jump from $0$ to $1$), the synaptic weights may change depending on the STDP functions $p^\pm$, which depend on the time since the last spike $S$ and the weight considered $W$. For example, when neuron $2$ spikes at time $\tau_1$, the weight $W^{12}$ may decrease with probability $p^-(S_{\tau_1^-}^{1}, W_{\tau_1^-}^{12})$, while the time since the last spike of neuron $2$ is reset: $S_{\tau_1}^{2}=0$. Note that in our example the $S$ start at $0$, so $S_{\tau_1^-}^{1}=\tau_1$, and the synaptic weight changes at every spike which is usually not the case as such jumps are stochastic.
        }
        \label{fig:model_schematic}
\end{figure*}

The neuron model (stochastic Wilson-Cowan model) has been widely studied \cite{bressloff2010metastable,benayoun2010avalanches}. It reproduces many biological features of a network such as oscillations, bistability, metastability, and memory formation \cite{litwin2014formation}. 
The model is a special case of the one studied in \cite{helson2017new,helson2021plasticity}.
Thus, we obtain a model rich enough to reproduce biological phenomena, but still simple enough to be mathematically tractable and easy to be simulated with thousands of neurons.

We end the description of the microscopic model by giving assumptions
on the initial conditions.
\begin{assumption}\label{hyp:cdt-init}
For any number of neurons \(N\), the initial conditions
$\Big(\big(V_0^{i,N},S_0^{i,N},(W_0^{ij,N})_{1\leq j \leq N}\big)\Big)_{1\leq i \leq N}$ are assumed to be independent and identically distributed. 
The distribution $\rho_0$ of the second component $S_0^{i,N}$ does not depend on $N$, is absolutely continuous with respect to the Lebesgue measure 
and has a bounded density.
However, the third component $W$ belongs to $\mathbb Z^{N^2}$ and thus  its distribution depends on $N$.
\end{assumption}
Under this assumption, the neural network possesses two important properties that
are extensively used throughout this paper.
First, the times from the last spike are almost 
surely distinct, see Lemma~\ref{lem:supp-s-distinct} in the Supplemental Material (SM).
Second, 
at any time \(t\geq 0\),
their distributions are also absolutely continuous with respect to the Lebesgue measure (admit a probability density), see Lemma~\ref{lem:densite-S-t} in the SM.

\subsection{Toward the macroscopic model}
The aim of this section is to explain the challenge of deriving a MF limit in the context of plastic interactions, and how we address the latter. We begin with an overview of our method, then introduce the new variables (auxiliary variables) required to derive the MF limit for plastic synapses, along with their dynamics. This sets the stage for why we can ultimately close the equations on these new auxiliary variables.
\subsubsection{Method overview}
If the weights were fixed, one could consider the behaviour of the average membrane potential, 
\[
    m_t^N = \frac{1}{N} \sum_i V_t^{i,N},
\] 
when $N$ tends to infinity. This specific case has been studied in detail in \cite{farkhooi2017complete}. In particular, they describe the condition on the weights under which this limit is not deterministic anymore as well as the finite-size correction to perform in such a stochastic limit framework. Our model is more complex and in particular, we cannot write a closed equation on $m_t^N$. Therefore, we will use the formalism employed in the context of propagation of chaos, see \cite{sznitman_topic} for more details.
\\
Let us say we can describe our system from $N$ variables $X_t^{i,N}$. For example, $X_t^{i,N}\coloneqq V_t^{i,N}$ would work in the previous example of fixed weights. Then, we are interested in the dynamics of the empirical distribution
\[
    \mu_t^N \coloneqq \frac{1}{N} \sum\limits_i \delta_{X_t^{i,N}}.
\]
In particular, we want to find the possible deterministic limit
$(\mu_t^*)_{t \geq 0}$ of $(\mu_t^N)_{t \geq 0}$ when $N$ tends to infinity. 
In such context, it usually turns out that the \textit{typical neuron} $(X_t^*)_{t \geq 0}$ 
is a stochastic process with distribution $(\mu_t^*)_{t \geq 0}$ satisfying a McKean-Vlasov \cite{mckean1966class} stochastic differential equation (SDE): the dynamics of $X_t^*$ depends on its own distribution $\mu_t^*$.
In our case, the latter will have the form
\[
\frac{dX^*_t}{dt}
= b(X^*_t,\mu^*_t)
+ \sum_i h(X^*_{t_i^-}, z_i),
\]
where $(t_i,z_i)$ are points of a marked Poisson process \cite{last2018lectures}
and $b, h,$ are functions respectively describing the drift and the jumps of the process $X_t^*$. The latter functions depend on the model parameters.
\\
The plasticity of interactions makes the state space of variables of order $N^2$ -- with the presence of the $N^2$ weights $(W^{ij})_{i,j}$ -- thus making non-trivial the definition of the auxiliary variables $X_t^{i,N}$ needed for the formalism described above. In what follows, we show that the limit process $X_t^*$ lives on the peculiar space
\[
\{0,1\} \times \mathbb R^+ \times \mathcal P(\{0,1\} \times \mathbb R^+ \times \mathbb Z)
\]
where $\mathcal P(\{0,1\} \times \mathbb R^+ \times \mathbb Z)$ is the space of probability distributions on $\{0,1\} \times \mathbb R^+ \times \mathbb Z$.

\subsubsection{Definitions of the auxiliary variables}
In the following, with a slight abuse of notation, we also call
\textit{neuron} the process describing the neuron. The neuron $i$, $X_t^{i,N}$, includes the state of the neuron $V_t^{i,N}$, its time since the last spike $S_t^{i,N}$ and the empirical distribution \(\xi_t^{i,N}\) of the triplets \((V_t^{j,N},S_t^{j,N},W_t^{ij,N})_{1\leq j \leq N}\), that is the joint distribution of the incoming weights and the neural state $(V,S)$ associated to them.
\begin{definition}\label{def:mesure-xi}
	We define the following random variable on the space of probability distributions on 
\begin{equation}\label{eq:Em}
	E_m \coloneqq  \{0,1\} \times \reels^+ \times \relatifs:
\end{equation}
\[
\forall i \in \{1,\cdots,N\}, \qquad
\xi_t^{i,N} \coloneqq   \frac{1}{N} \sum\limits_j \delta_{ ( V_t^{j,N}, S_t^{j,N}, W_t^{ij,N} ) },
\]
where $\delta$ is the Dirac delta distribution.

In particular, denoting by $\mathcal P_N(G)$ the set of atomic probability distributions on $G$ with $N$ distinct atoms 
with uniform weights $\frac{1}{N}$, we have that for all outcomes $\omega \in \Omega$,
$\xi_t^{i,N}(\omega) \in \mathcal P_N(E_m)$.

Using the function $I:\mathcal P(E_m) \rightarrow \reels$:
\begin{equation}\label{def:I_syn-current}
    \forall \xi \in \mathcal P(E_m), \quad
I(\xi) \coloneqq  \esp^\xi(WV) = \int_{E_m}  w v \ \xi(dv, ds, dw)
,
\end{equation}
we note that the knowledge of \(\xi_t^{i,N}\) is enough to obtain the 
incoming synaptic current \( I_t^{i,N} = I(\xi^{i,N}_t) \).

We can now re-write the model with the new definition of a \textit{neuron}.
The neurons are now described by the following triplets, for all $i \in \{1,\cdots,N\},$
\begin{align*}
	X_t^{i,N} \coloneqq  (V_t^{i,N}, S_t^{i,N}, \xi_t^{i,N})
\end{align*}
and the system state is then given by
\begin{align*}
    X^N_t \coloneqq  (X^{1,N}_t,\cdots, X^{N,N}_t).
\end{align*}
\end{definition}
Now that we have defined the auxiliary variables that describe the neural network,
we can define the empirical distribution associated to the \textit{new} neural network description, $(\mu_t^N)_{t \geq 0}$. Importantly, it is defined on a space that does not depend on $N$.
\begin{definition}\label{def:mesure-mu}
    The empirical distribution
    \begin{equation}\label{def:mu-xi}
    \mu_t^N \coloneqq  \frac{1}{N} \sum\limits_i \delta_{X_t^{i,N}} = \frac{1}{N} \sum\limits_i 
    \delta_{(V_t^{i,N}, S_t^{i,N}, \xi_t^{i,N})}
    \end{equation}
    is a random probability distribution over the space of probability distribution on
    \begin{equation}\label{eq:E}
    	E \coloneqq  \{0,1\} \times \reels^+ \times \mathcal{P}(E_m).
    \end{equation}
    Hence, for all $\omega \in \Omega$,
    \(\mu_t^N(\omega) \in \mathcal P_N(E) \subset \mathcal P(E) \).
\end{definition}
In the following, we keep in mind the randomness of $\mu_t^N$
and $(\xi_t^{i,N})_{1 \leq i \leq N}$ and we alleviate the notations by referring to the
outcomes $\omega \in \Omega$ only when it helps understanding. From the definition of
$(\mu_t^N)_{t \geq 0}$ and the ones of
$(\xi_t^{1,N})_{t \geq 0}, \cdots, (\xi_t^{N,N})_{t \geq 0}$, we note that they have
the same joint distribution on $V$ and $S$.
\begin{remarque}\label{rem:egalite-mu-xi}
    For all $i \in \{1,\cdots,N\}$ and $t \geq 0$,
    we have by Definitions~\ref{def:mesure-xi} and~\ref{def:mesure-mu} that
    for all $(v,A) \in \{0,1\} \times \mathcal B (\reels^+)$,
    \[
    \mu_t^N (v, A, \mathcal{P}(E_m) ) = \xi_t^{i,N} ( v, A, \relatifs ).
    \]
    In what follows, we denote this property by:
    $\mu_t^N ( \cdot, \cdot, \mathcal{P}(E_m) ) = \xi_t^{i,N} ( \cdot, \cdot, \relatifs )$.
\end{remarque}

\subsubsection{Dynamics of the system}
Under Assumption~\ref{hyp:cdt-init}, the initial \textit{neurons} $X_0^{1,N},\cdots,X_0^{N,N}$ 
are $N$ independent realisations of a unique distribution. This assumption is at the heart of our ability to describe the dynamics of $X_t^N$.
\begin{remarque}\label{rem:20210215}
    Even though the knowledge of \((\xi_t^{i,N})_{1\leq i\leq N}\) is sufficient to evaluate the spiking rates of the neurons \((X_t^{i,N})_{1\leq i\leq N}\) , they are empirical distributions and thus do not contain the \textit{labels}  of the pre-synaptic neurons. 
	In particular, as soon as one neuron, say neuron \(i_0\), has a jump of \(V^{i_0,N}_t\), we have to propagate this jump (and its potential consequences, such as a reset of \(S^{i_0,N}_t\) and changes in the synaptic weights) to every \((\xi_t^{i,N})\). For any \(i\), we know that one atom of \(\xi_t^{i,N}\) corresponds to the neuron \(i_0\) but we have to find it. Thanks to Lemma~\ref{lem:supp-s-distinct}, we identify this atom with the value \(S^{i_0,N}_t\).
\end{remarque}
To do so, we denote by $\mathcal{T}_t$ the set of times $\tau\in [0,t]$ such that  $V_{\tau}^{i,N} \neq V_{\tau^-}^{i,N}$ 
for some \(i\in\{1,\cdots,n\}\). 
First, in between two consecutive jumps at times $\tau$ and $\tilde{\tau}$, $\tau<\tilde{\tau}$, 
the time since the last spike of all neurons increases linearly: for all $i \in \{1,\cdots, N\}$ and $t\in [\tau,\tilde{\tau}[$,
\begin{align*}
    S_t^{i,N} &= S_{\tau}^{i,N}+t-\tau,
    \\
    \xi_t^{i,N} &= \frac{1}{N} \sum_{ (\tilde{V}, \tilde{S}, \tilde{W}) \in \supp(\xi_{\tau}^{i,N}) } \delta_{ (\tilde{V}, \tilde{S}+t-\tau, \tilde{W})},
\end{align*}
where \(supp(\xi_t^{i,N})\) is the support of the empirical distribution \(\xi_t^{i,N}\), namely the triplets \(V,S,W\) where it has some mass (\(\frac{1}{N}\) in this specific case). Now considering the jump part, we denote by $\mathcal{T}_t^{i} \subset \mathcal{T}_t$ the spike times of neuron $i$ and 
$\tilde{\mathcal{T}}_t^{i} \subset \mathcal{T}_t$ the return to resting
potential of neuron $i$ ($V^i: 1 \to 0$). Among these jumping times, $\mathcal{T}_t^{ij,\pm} \subset \mathcal{T}_t^i$ stands for the spike times of
neuron $i$ such that $W^{ij}$ potentiates/increases ($\mathcal{T}_t^{ij,+}$) and $W^{ji}$ depresses/decreases ($\mathcal{T}_t^{ij,-}$).
Hence, for all $i \in \{1,\cdots,N\},$ and $\tau \in \mathcal{T}_t$,
\begin{align}\label{eq:dyn_xi}
	\begin{split}
		V_{\tau}^{i,N} & = 
		\left\{
		\begin{array}{ll}
			1- V_{\tau^-}^{i,N} \text{ if } \tau \in \mathcal{T}_t^i \cup \tilde{\mathcal{T}}_t^i,
			\\
			\\
			V_{\tau^-}^{i,N} \text{ otherwise},
		\end{array}
		\right.
		\\[1em]  
		S_{\tau}^{i,N} & = 
		\left\{
		\begin{array}{ll}
			0 \text{ if } \tau \in \mathcal{T}_t^i,
			\\
			\\
			S_{\tau^-}^{i,N} \text{ otherwise},
		\end{array}
		\right.
		\\[1em]
		\Delta \xi_{\tau}^{i,N} & = \xi_{\tau}^{i,N} - \xi_{\tau^-}^{i,N}
		\\
		& = 
		\frac{1}{N}\sum_{ (\tilde{V}, \tilde{S}, \tilde{W}) \in \supp(\xi_{\tau^-}^{i,N}) }
		\Biggl\{ 
		\\
		& 
		\sum_{j} 
		\mathbbm 1_{\{ \tau \in \mathcal{T}_t^{j}, \tilde S = S_{\tau^-}^{j,N} \}}
		\big[
		\delta_{\big(
			1,0, \tilde W 
			- \mathbbm 1_{\{ \tau \in \mathcal{T}_t^{ji,-}\}}
			\big)} 
		-
		\delta_{\big(
			0,\tilde S, \tilde W 
			\big)} 
		\big]
		\\
		& 
		+ \sum_{j}
		\mathbbm 1_{\{ \tau \in \mathcal{T}_t^{ij,+}, \tilde S = S_{\tau^-}^{j,N} \}}
		\big[\delta_{\big(
			\tilde V,\tilde S, \tilde W 
			+ 1
			\big)}
		- \delta_{\big(
			\tilde V,\tilde S, \tilde W 
			\big)}
		\big]
		\\
		& 
		+ \sum_j \mathbbm 1_{\{\tau \in \tilde{\mathcal{T}}_t^j, \tilde S = S_{\tau^-}^{j,N} \}}
		\big[
		\delta_{\big(0,\tilde S, \tilde W\big)}
		-
		\delta_{\big(1,\tilde S, \tilde W\big)}    
		\big]
		\Biggr\},
	\end{split}
\end{align}
where the characteristic function $\mathbbm 1_{\{x \in A\}}$ is equal to $1$ if $x \in A$ and $0$ otherwise. 
The first line contains the jumps related to the spike of neuron $j \neq i$; reset of $S^j$ to $0$ and potential depression of $W^{ji}$.
The second line contains the potentiations of the weights $(W^{ij})_{j \neq i}$, this happens when the neuron $i$ spikes. 
The final line is the return of the neurons to their resting potential, from $V^j=1$ to $V^j=0$. Note that we no longer see the functions $p^\pm$ which are in fact hidden in the jumps $\tau^{ij,\pm}$. We describe them rigorously when giving the dynamics of \(\mu^N_t\). The latter follows directly from that of \(X^N_t\) as the only difference is that there is no label anymore (as for the \(\xi\)). They can be recovered again thanks to the time since the last spike $\tilde{S}$ present in each atom $\tilde{X} = (\tilde{V},\tilde{S},\tilde{\xi})$ of \(\mu^N_t\).

\section{Main results}\label{sec:Study of the empirical measure and its limit}
Our first important result is the Markov property of the process $(\mu_t^N)_{t \geq 0}$ defined on $\mathcal P_N(E)$, see Proposition~\ref{prop:mu-Markov}
in the SM.

The main result of this paper is to provide the possible deterministic limit
$(\mu_t^*)_{t \geq 0}$ of $(\mu_t^N)_{t \geq 0}$ when $N$ tends to infinity. In particular, we consider a process with distribution $(\mu_t^*)_{t \geq 0}$ as the \textit{typical neuron} and denote it by $(X_t^*)_{t \geq 0} = (V_t^*,S_t^*,\xi_t^*)_{t \geq 0}$.

\begin{widetext}
    \begin{mainresult}\label{th:equation limite systeme de particules avec plasticite}
    Under the assumptions~\ref{hyp:cdt-init},~\ref{ass:lim-mu-star-1},~\ref{ass:lim-mu-star-2} and~\ref{ass:compabilite EDP} detailed in the SM,
    the limit objects \(\mu^*_t\) and \((V_t^*,S_t^*,\xi_t^*)_{t \geq 0}\) satisfy the following McKean-Vlasov SDE:
    \begin{enumerate}
        \item \(\mu^*_t\) is the distribution of \((V_t^*,S_t^*,\xi_t^*)\): \(\mu^*_t = \mathcal{L}(V_t^*,S_t^*,\xi_t^*)\),
		\item $V_t^*: 0 \xrightleftharpoons[\beta]{\alpha(I(\xi_t^*))} 1$. We denote by $\mathcal{T}_{t}^*$ the spikes times, 
        that are the times  on $[0,t]$ when \(V\) jumps from \(0\) to \(1\),
		\item \( S_t^* = S_0^* + t - \sum_{\tau \in \mathcal{T}_t^*} S_{\tau^-}^* \),
        \item $\xi_t^*$ admits a density in $s$ such that $\xi_t^*(v,ds,w)= \xi_t^*(v,s,w)ds$ (abuse of notation) and between the spikes
        \begin{align*}
    		&\left\{
    		\begin{array}{ll}
    		\partial_t{\xi_t^*}(0,s,w) = -\partial_s {\xi_t^*}(0,s,w)
    		+ \beta{\xi_{t}^*}(1,s,w)
    		- \int_{ \mathcal{P}(E_m)}
    		\alpha\big(I(\xi')\big) 
    		\frac{{\xi_t^*}(0,s,w)}{{\xi_t^*}(0,s,\relatifs)}
    		\mu_t^*(0,s,d\xi')
    		\\
    		\\
    		{\xi_t^*}(0,0,w) \hspace{0.4em}= 0,
    		\end{array}
    		\right.
    		\\
    		\\
    		&\left\{
    		\begin{array}{ll}
    		\hspace{-0.5em}
    		\partial_t{\xi_t^*}(1,s,w) \hspace{-0.2em}= -\partial_s {\xi_t^*}(1,s,w)
    		- \beta {\xi_t^*}(1,s,w)
    		\\
    		\\
    		{\xi_t^*}(1,0,w)  \hspace{0.2em} = \int_{\reels^+ \times \mathcal{P}(E_m)}
    		\alpha\big(I(\xi')\big)
    		\frac{{p^-(S_t^*,w+1)\xi_t^*}(0,s',w+1) 
    			+ (1-p^-(S_t^*,w)){\xi_t^*}(0,s',w)}{{\xi_t^*}(0,s',\relatifs)}
    		\mu_t^*(0,s',d\xi')ds'.
    		\end{array}
    		\right.
    		\end{align*}
        \item At a spike time $\tau \in \mathcal{T}_\infty^*$, $\xi_{\tau^-}^*$ jumps to $\nu^+\{\xi_{\tau^-}^*\}$ where for all probability distribution functions $\xi$ on $E_m$,
		\[
		\nu^+\{\xi\} (v,s,w) \coloneqq  p^+(s,w-1) \xi(v,s,w-1) + (1-p^+(s,w)) \xi(v,s,w).
		\]
	\end{enumerate}
\end{mainresult}
\end{widetext}

We now give some details on this limit system. The second point details the membrane potential (spiking) activity of the neuron. Note that the spiking jump depends on the distribution $\xi_t^*$ through the function $I$ defined in \eqref{def:I_syn-current}. The third point is the linear increase of the time since the last spike between the spikes of the neuron.
The fourth point gives the partial differential equations (PDE) between the jumps (spikes).
The PDE on $\xi_t^*$ is made up of three parts. First, the linear increase of $s$ which is the term in $-\partial_s$. Second, the mass transport from \({\xi_t^*}(1,\cdot,\cdot)\) to ${\xi_t^*}(0,\cdot,\cdot)$ at rate $\beta$. Finally, the mass transport from ${\xi_t^*}(0,\cdot,\cdot)$ to ${\xi_t^*}(1,\cdot,\cdot)$ at a rate that depends on $\mu_t^*$ with the reset of $s$ in $0$
giving the boundary condition in $s=0$; see the ${\xi_t^*}(1,0,\cdot)$ term. This corresponds to a spike of a pre-synaptic neuron, hence leading to the depression of the associated weight (outgoing weight for this pre-synaptic neuron). Then, the fifth and last point. When the \textit{typical neuron} spikes, $V_t^*:0 \to 1$, the weights' potentiation occurs; meaning a transfer of mass from ${\xi_t^*}(\cdot,\cdot,w-1)$ to ${\xi_t^*}(\cdot,\cdot,w)$.

Note that all the changes of $(\xi_t^*)_{t \geq 0}$ are transfers of mass on the space on which it is defined; hence showing that its total mass is conserved and thus $(\xi_t^*)_{t \geq 0}$ remains a probability distribution (of mass $1$) over time. In addition, this system of equations involves the distribution of its solution making $(X_t^*)_{t \geq 0}$ follow a McKean–Vlasov SDE. 
In particular, the process $(X_t^*)_{t \geq 0}$ solves the SDE
\begin{align*}
X_t^*
= X_0^* + \int_0^t
b(X_u^*,\mu_u^*) du
+ \int_0^t\int_{\reels^+} h(X_{u^-}^*,z) \zeta^*(dz,du)
\end{align*}
where $\zeta^*$ is a Poisson random measure on $\reels^+ \times \reels^+$
with intensity $dzdu$, the probability distribution \(\mu_u^*\) is the distribution of the process \(X^*_u\), 
and the functions $b$ and $h$ are
defined according to Main Result \ref{th:equation limite systeme de particules avec plasticite}: 
\begin{align*}
	&h(X_{u^-}^*,z)  = \\
    &\Big(1,-S_{u^-}^*,\nu^+(\xi_{u^-}^*)- \xi_{u^-}^*\Big)\mathbbm 1_{ \big\{ z \leq \alpha\big(I(\xi_{u^-}^*)\big) \mathbbm 1_{ \{V_{u^-}^* = 0\} }\big\} }
    \\
     &\qquad
    + (-1,0,0) \mathbbm 1_{ \big\{ z \leq \beta \mathbbm 1_{ \{V_{u^-}^* = 1\} }\big\} }
    \\
	&b(X_u^*,\mu_u^*)  = \\
    &\Big(0,1,-\partial_s {\xi_u^*}  + \beta \delta_0 \otimes {\xi_u^*}(1,\cdot,\cdot) - \beta \delta_1 \otimes {\xi_u^*}(1,\cdot,\cdot)
	\\
    & \qquad
    +\delta_1 \otimes \nu^{-,1} (S_u^*,{\xi_u^*}, \mu_{u}^* ) - 
	\delta_0 \otimes \nu^0 ( {\xi_u^*}, \mu_{u}^* )
	\Big)
\end{align*}
where $\nu^+$, $\nu^{0}$ and $\nu^{-,1}$ are defined respectively in \eqref{eq:def-nuplus}, \eqref{eq:defnutilde0} and \eqref{eq:def nu_moins_1}.
Note that the process $(X_t^*)_{t \geq 0}$ lives on the peculiar space $E=\{0,1\} \times \reels^+ \times \mathcal P(E_m)$ which
contains the space of distributions on $E_m$.
We give the computations leading to such a limit candidate in the SM where we first study the system without plasticity \ref{supp:The case without plasticity} and then with plasticity \ref{supp:The case with plasticity}.

\begin{remarque}\label{rem:echangeable}
Classically, the proof of the convergence of the
particle system to the limit dynamics is done when the particles are exchangeable.
By Assumption~\ref{hyp:cdt-init}, the distributions $\xi_0^{1,N}, \cdots, \xi_0^{N,N}$ have the same distribution.
This assumption added to the fact that the variables $V_t^{i,N}$, $S_t^{i,N}$ and $W_t^{ij,N}$ have for all $i,j,$ the same dynamics,
implies that for all $i$ and $t \geq 0$, the $\xi_t^{i,N}$ are equal in distribution.
	We conclude this remark with the crucial following point:
	using Assumption~\ref{hyp:cdt-init}, we deduce that the distribution of $X_t^N$ is 
	exchangeable, 
    which means that 
	for any permutation $\sigma$ of $\{1,\cdots,N\}$, $X_t^{\sigma,N} =
	(X_t^{\sigma(1),N},\cdots,X_t^{\sigma(N),N})$ has the same distribution as $X_t^N$.
\end{remarque}

\section{Numerical comparison with the neural network}
\label{sec:Simulations}

We simulate the stochastic STDP model. In particular, we
illustrate our results by comparing the simulation of the finite-size
neural network versus the simulation of the limit system.

We use a sigmoid for the function $\alpha$ and the classical STDP curves for $p^+$ and $p^-$,
\begin{align*}
    &\alpha(x) = \frac{\alpha_M-\alpha_m}{1+e^{\sigma(\theta - x)}} + \alpha_m
    \\
	&\text{and} \quad
	p^{+/-}(s,w) = A_{+/-} e^{-\frac{s}{\tau_{+/-}}}\mathbbm{1}_{ \llbracket w_{min}, w_{max}\rrbracket}(w),
\end{align*}
with the following parameters:
\begin{align*}
& N = 5000,\ \alpha_m = 0.05\ ms^{-1},\ \alpha_M = \beta = 1\ ms^{-1},
\\
&
\tau_+ = 1.5\ ms,\tau_- = 2\ ms,\ A_+ = 0.8,\ A_- = 0.6,
\\
&
dt = 0.05\ ms,\ \sigma = 1.5\ mV^{-1},\ \theta = 0\ mV,
\\
&
w_{max} = - w_{min} = 10.
\end{align*}

Note that the functions $p^\pm$ are explicitly defined to depend on the weights, which allows us to choose the above form that halts plasticity updates once the bounds are reached. This approach ensures consistency between the theoretical framework and the simulations with or without hard bounds. The use of bounded weights is primarily for computational convenience.

For the limit system, we simulate $N$ neurons $X_t^{1,*},\cdots,X_t^{N,*}$ having the same dynamics as $X_t^*$ (see Main Result \ref{th:equation limite systeme de particules avec plasticite} and Remark \ref{rem:echangeable}), except that instead of $\mu_t^*$
we use $\mu_t^{*,N} = \frac 1N \sum_i \delta_{X_t^{i,*}}$. Hence, we do not
need to compute $\mu_t^*$. For instance, we use $\xi_t^{i,*}$ to
compute the $I_t^{i,*} = I(\xi_t^{i,*})$. For more details on the code, see
Sec. \ref{ann:ch-4} in the SM.

The results obtained with these simulations are compared with those obtained from a finite-size neural network:
$(V_t^{i,N},S_t^{i,N},W_t^{ij,N})_{1 \leq i,j \leq N, t\geq 0}$.
The mathematical objects illustrated in Figure \ref{esperance-W-I} are the empirical mean and distributions of different variables. 
\begin{definition}\label{def:empirical_dist_mean}
    Consider a finite ensemble of size $M$ that we denote by $Z = {Z^1, \cdots, Z^M}$, its empirical mean is
    \[
        \overline{Z} = \frac{1}{M} \sum_{i=1}^M Z^i,
    \]
    and its empirical distribution is
    \[
        \hat{P}_Z^M = \frac{1}{M} \sum_{i=1}^M\delta_{Z^i}.
    \]
\end{definition}
We compare the temporal evolution of the empirical mean of the potential (Figure~\ref{sub:esperance-V-final-500ms}), the time since the last spike (Figure~\ref{sub:esperance-S-final-500ms}) and synaptic weight (Figure~\ref{sub:esperance-W-final-500ms}). In particular, 
\begin{align*}
    &Z \in \{V_t^N,V_t^*,\xi_t^*(\cdot,\reels^+,\relatifs)\} \text{ in Figure~\ref{sub:esperance-V-final-500ms}},
    \\
    &Z \in \{S_t^N,S_t^*,\xi_t^*(\{0,1\},\cdot,\relatifs)\} \text{ in Figure~\ref{sub:esperance-S-final-500ms},}
    \\
    &Z \in \Big\{W_t^{N},\xi_t^{*}(\{0,1\},\reels^+,\cdot)\Big\} \text{ in Figure~\ref{sub:esperance-W-final-500ms}.}
\end{align*} 
Moreover, we compare, at the end of the simulation $t=500ms$,
the distributions of the intensities of the incoming currents
onto the neurons (Figure~\ref{sub:distribution-I-final-500ms}),
the distributions of the time since the last
spikes for neurons in state $V=0$ (Figure~\ref{sub:distribution-S-final-V0-500ms}) and $V=1$ (Figure~\ref{sub:distribution-S-final-V1-500ms}). 
The latter are estimated in three ways: using the empirical distribution from the original system and the MF system with $N$ typical neurons (see Definition \ref{def:empirical_dist_mean}), and directly from the distribution $\xi_t^{i,*}$ for a given neuron $i=1$. The distance between the curves for long times in Figure~\ref{sub:esperance-S-final-500ms} is due to the approximation done in $\xi_t^{i,*}$ as we have to bound it in $S$ (see Sec. \ref{ann:ch-4} in the SM for more details); bound which does not exist in the original model. The slight shift in the synaptic current follows from this error of approximation which leads to an accumulation in the upper bound as it can be seen in Figure~\ref{sub:distribution-S-final-V0-500ms}.

Overall, the simulations' results show a good agreement between the initial network and the limit one. 
In the time-series plots, we see that on average, the MF variables follow the ground truth closely. We also observe that the whole distribution of these variables is captured. Interestingly, we see that the synaptic currents have a wide distribution (far from the initial condition which is a Dirac distribution), showing that each neuron receives different inputs even though the latter is an expectation, see equation~\eqref{def:I_syn-current}.
\begin{figure*}
    \subfloat[]{
    \includegraphics[width=0.45\textwidth]{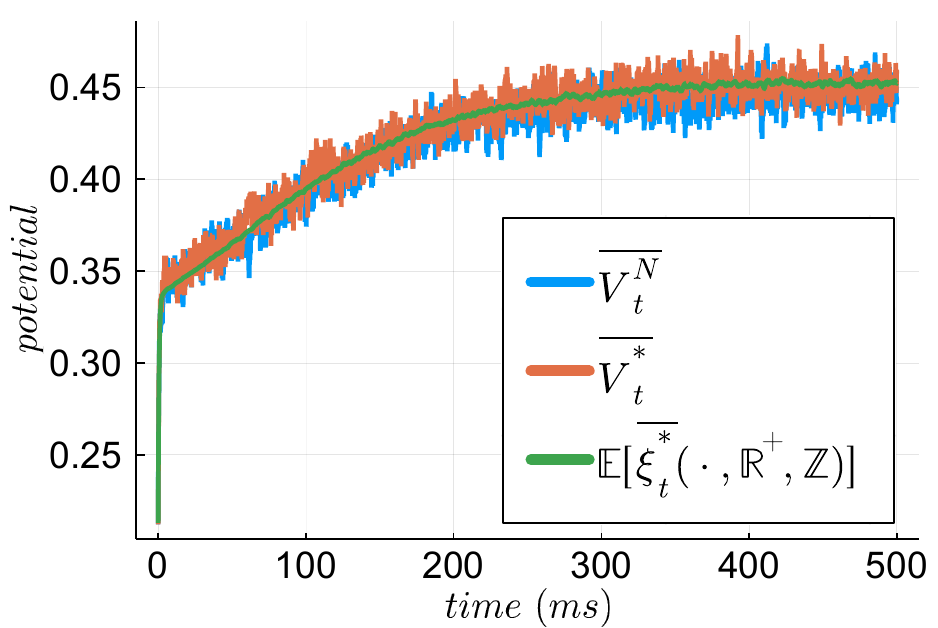}
    \label{sub:esperance-V-final-500ms}}
    \subfloat[]{
    \includegraphics[width=0.45\textwidth]{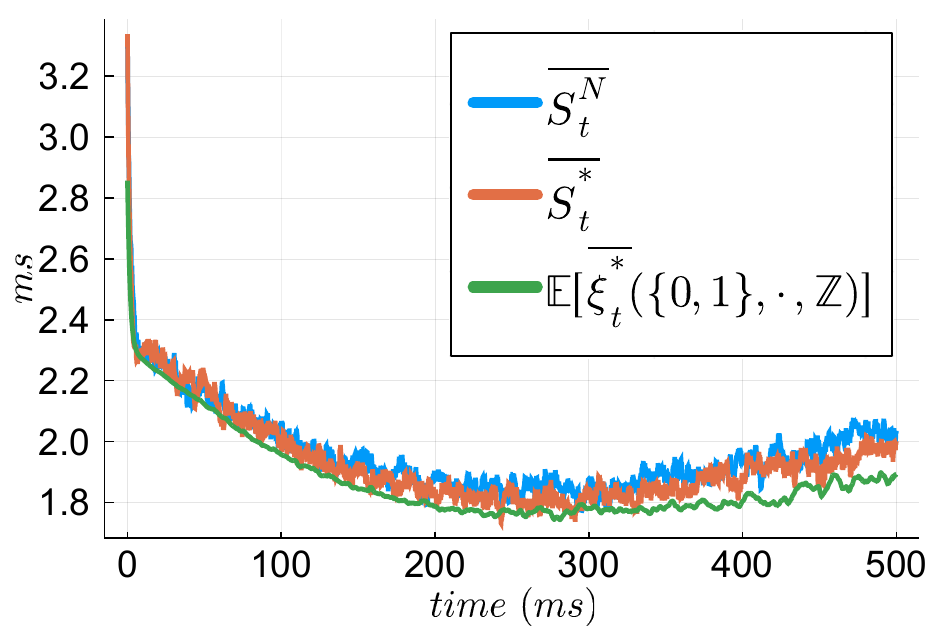}
    \label{sub:esperance-S-final-500ms}}
    \\
    \subfloat[]{
    \includegraphics[width=0.45\textwidth]{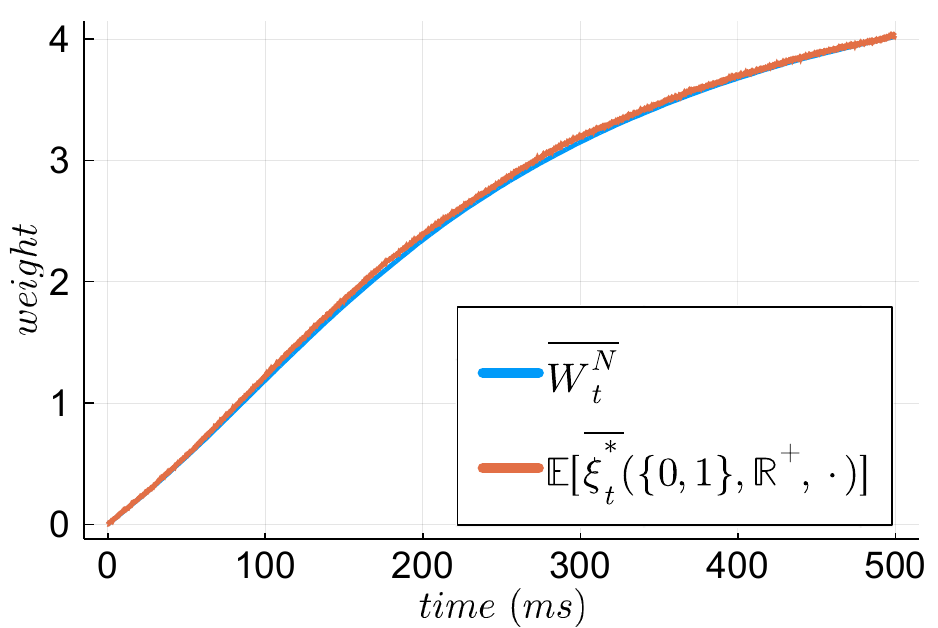}
    \label{sub:esperance-W-final-500ms}}
    \subfloat[]{
    \includegraphics[width=0.45\textwidth]{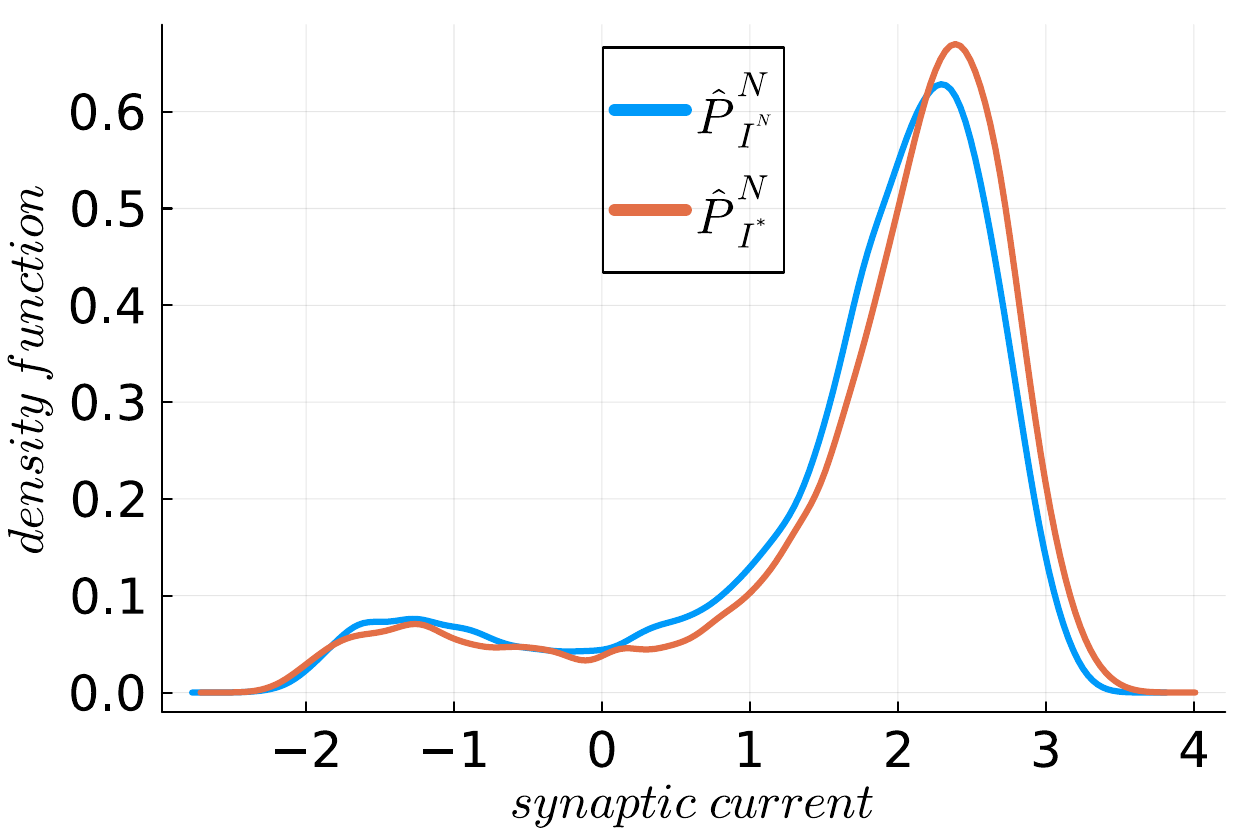}
    \label{sub:distribution-I-final-500ms}}
    \\
    \subfloat[]{
        \includegraphics[width=0.45\textwidth]{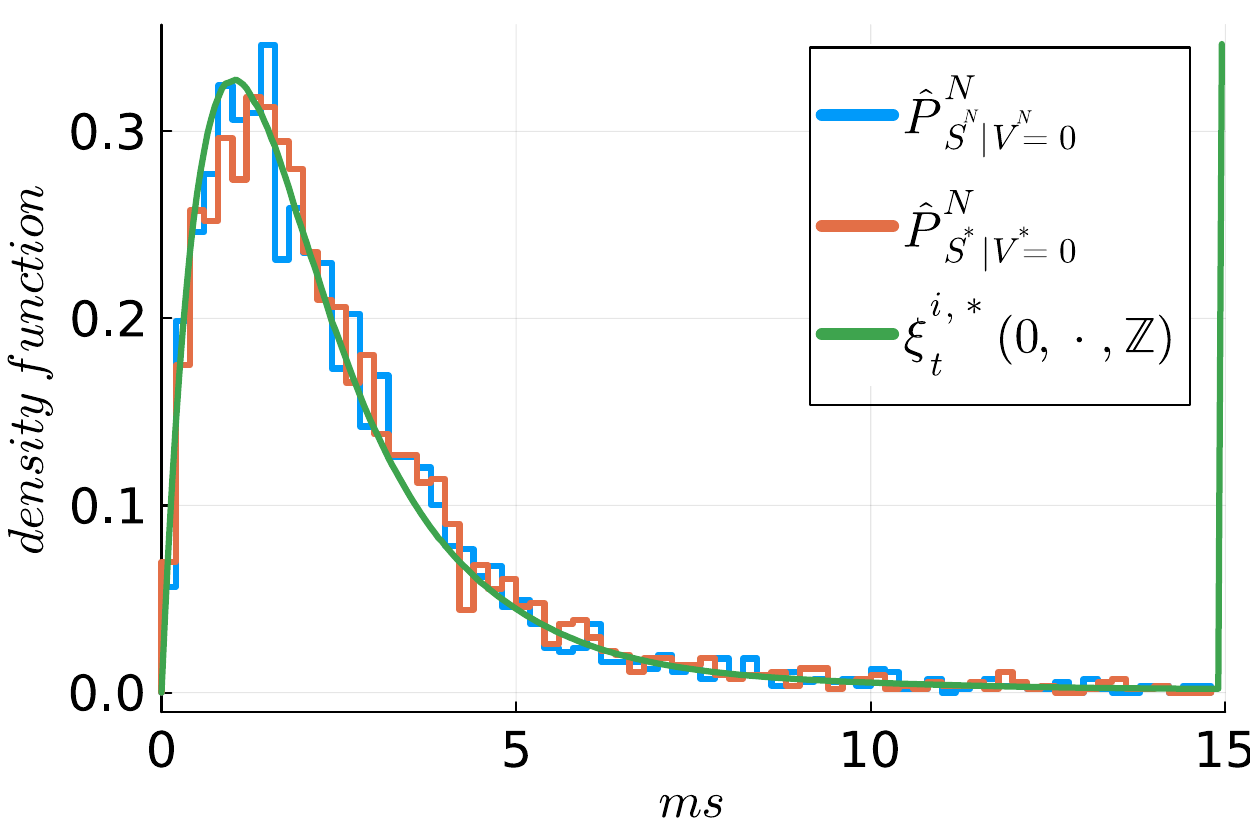}
        \label{sub:distribution-S-final-V0-500ms}}
    \subfloat[]{\includegraphics[width=0.45\textwidth]{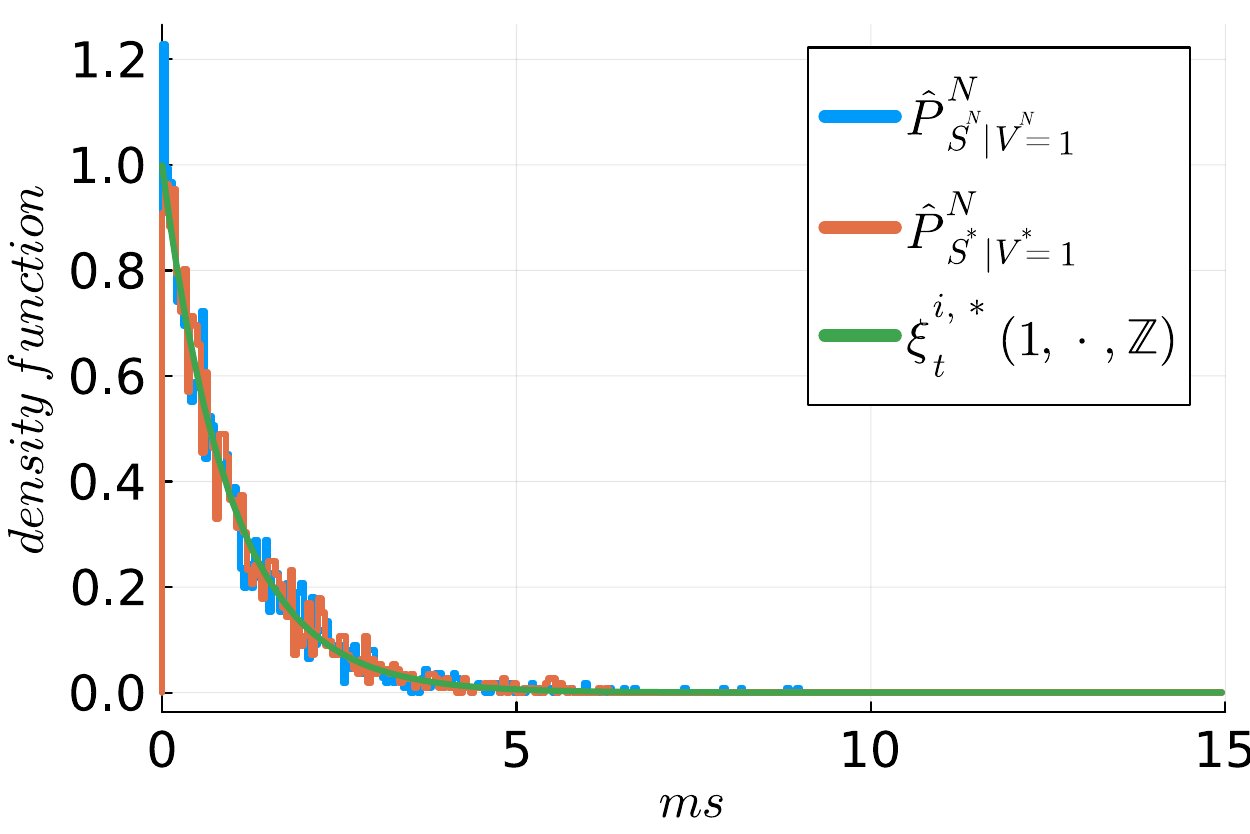}
    \label{sub:distribution-S-final-V1-500ms}}
    \caption{\label{esperance-W-I}
        \textbf{Comparisons between the MKV-MF model and the initial model.}
        (\ref{sub:esperance-V-final-500ms}-c): Empirical mean (see Definition \ref{def:empirical_dist_mean} and below) of the potentials \eqref{sub:esperance-V-final-500ms}, time since the last spike \eqref{sub:esperance-S-final-500ms}, and the weights \eqref{sub:esperance-W-final-500ms} across time. \eqref{sub:distribution-I-final-500ms}: Empirical distributions (see Definition \ref{def:empirical_dist_mean}) at time $t=500ms$ of the synaptic currents onto the neurons. (\ref{sub:distribution-S-final-V0-500ms}-f):  Distributions (see Definition \ref{def:empirical_dist_mean} and below) at time $t=500ms$ of the time since the last spike for neurons in state $V=0$ \eqref{sub:distribution-S-final-V0-500ms} and $V=1$ \eqref{sub:distribution-S-final-V1-500ms}.
        }
\end{figure*}

\section{Discussion}
We derived a McKean-Vlasov mean-field (MKV-MF)
equation from a neural network model with stochastic STDP. To our knowledge, this is the first MF model of a network with STDP that does not rely on the symmetric STDP curve \cite{duchet2023mean} or the slow-fast framework where the plasticity is slow compared to the neuronal dynamics \cite{perthame2017distributed}; assumptions in conflict with experimental observations \cite{zenke2017temporal, lansner2023fast, choquet2013dynamic, kasai2021spine, fulton2024common}. When considering spiking neural networks without plasticity, MFs have been derived with, for example, the Ott-Antonsen ansatz \cite{montbrio2015macroscopic} when weights are homogeneous. In the case of heterogeneous fixed weights, methods based on the estimation of the cross-correlation have been used \cite{ocker2017statistics}.
However, MKV-MFs have already been shown to be needed in such networks with both binary neurons and synaptic weights \cite{farkhooi2017complete} and more recently in integrate-and-fire neurons using the theory of dense graph limits (graphons) \cite{jabin2024non}.

This MF opens the door to the development of new mathematical tools necessary to prove convergence to the limit system and its analysis. The next steps would then be to show that the limit system has a unique solution and to prove the tightness \cite{sznitman_topic}
of the empirical distributions to conclude on their convergence 
to a unique deterministic limit. Finally, studying the limit system would provide insight into the finite-sized model and a particularly interesting study would be its long time behaviour \cite{cormier2020long}.
We are currently working on a simplified version of this model for these analytical derivations.

We believe that the model cannot be additionally reduced as we tried different model descriptions without success. A crucial element is the knowledge of the last spike time of each neuron.
We could have used the outgoing (instead of incoming) weights in $\xi$ but this leads to much more complex equations; see for example the synaptic current definition \eqref{def:I_syn-current}. To emphasize this point, in the simulations, when starting with a Dirac distribution of synaptic currents, we end up with a wide distribution of them. This happens even though these currents are defined as expectations; see \eqref{def:I_syn-current}. It shows that the limit model has some strong heterogeneity that our MKV-MF is able to capture.

In addition, the MF is quite flexible as we have only weak assumptions on the model parameters. For example, we could easily extend the results to account for Dale's law by making sure the initial weights follow it and then preventing them from changing sign using the weight dependence on the STDP functions $p^\pm$. Looking forward, our framework is flexible enough to incorporate more complex elements like the triplet-STDP rule \cite{pfister2006triplets} and more realistic neuron models \cite{clopath2010connectivity}. This would involve adding a second S-like variable to account for the time since the second last spike, and introducing a model like the leaky integrate-and-fire instead of a binary one as well as a voltage dependence of $p^{\pm}$. Such flexibility stands in contrast to the MF models recently derived for rate-based networks with plasticity \cite{clark2024theory, du_synaptic_2025}, where these extensions are not readily possible.

Finally, the simulation complexity of the limit system is here on the order of $N$ rather than $N^2$ originally. Hence, we have drastically reduced the simulation costs; despite the typical neuron being infinite dimensional.
Additionally, it closely follows the finite-size ground truth system both on average and distribution-wise as we showed in Figure \ref{esperance-W-I}.

\section{Conclusion}
Our approach paves the way for new mathematical tools to analyse complex neural systems more efficiently. 
It effectively addresses the combination of two of the most intricate phenomena in neuroscience: spikes and synaptic plasticity. This combination has been inadequately analysed in recurrent network settings until now. Using advanced mathematical tools, we reveal the MKV-MF limit of a biological spiking neural network with STDP.

\section{Acknowledgements}
PH acknowledges funding from Digital Future and the Ecole Doctorale en Sciences Fondamentales et Appliquées (ED.SFA, ED 364). PH expresses deep gratitude for the insightful discussions with Gonzalo Uribarri regarding the manuscript. Additionally, PH is particularly thankful for the mathematical derivations discussed with Milica Tomašević and Quentin Cormier, made possible through a BOUM project funded by the SMAI.

\section{Data Availability}
The code that supports the findings of this article is openly available at \url{https://github.com/paschels/MF_PRE}.


\section*{Supplemental Material}

\setcounter{theoreme}{0}
\setcounter{proposition}{0}
\setcounter{corollaire}{0}
\setcounter{conjecture}{0}
\setcounter{assumption}{0}
\setcounter{example}{0}
\setcounter{definition}{0}
\setcounter{lem}{0}
\setcounter{remarque}{0}
\setcounter{nota}{0}
\setcounter{consequence}{0}
\setcounter{mainresult}{0}

\setcounter{equation}{0}
\setcounter{figure}{0}
\setcounter{table}{0}
\setcounter{section}{0}

\renewcommand{\thetheoreme}{S\arabic{theoreme}}
\renewcommand{\theproposition}{S\arabic{proposition}}
\renewcommand{\thecorollaire}{S\arabic{corollaire}}
\renewcommand{\theconjecture}{S\arabic{conjecture}}
\renewcommand{\theassumption}{S\arabic{assumption}}
\renewcommand{\theexample}{S\arabic{example}}
\renewcommand{\thedefinition}{S\arabic{definition}}
\renewcommand{\thelem}{S\arabic{lem}}
\renewcommand{\theremarque}{S\arabic{remarque}}
\renewcommand{\thenota}{S\arabic{nota}}
\renewcommand{\theconsequence}{S\arabic{consequence}}
\renewcommand{\themainresult}{S\arabic{mainresult}}

\renewcommand{\theequation}{S\arabic{equation}}
\renewcommand{\thefigure}{S\arabic{figure}}
\renewcommand{\thesection}{S\arabic{section}}

\section{Proofs}
\subsection{Density and Markovianity of the empirical distribution}
\subsubsection{Properties of the time since the last spike \texorpdfstring{$S$}{TEXT}}
Under Assumption~\ref{hyp:cdt-init}, the neural network possesses two important properties that
are extensively used throughout this paper.
First, the times from the last spike are almost 
surely distinct.
\begin{lem}\label{lem:supp-s-distinct}
	Under Assumption~\ref{hyp:cdt-init}. Consider the process
	$X_t^N$ solution of the microscopic model described in Section \ref{sec:micro_model}.
	Then, for any \(t\geq 0\) and for all \(i\neq j \in \{1,\cdots,N\}\),
	we have almost surely $S_t^{i,N}\neq S_t^{j,N}$.
\end{lem}
\begin{proof}
	By assumption, for all $i$, the distribution of $S_0^{i,N}$ admits a density and hence,
	the $(S_0^{i,N})_{1\leq i \leq N}$ are almost surely distinct.
	Between the jumps of $(V_t^{i,N})_{1\leq i \leq N}$ from $0$ to $1$, 
	we have \(d(S_t^{i,N} - S_t^{j,N}) = 0\).
	Finally, the probability that two different neurons spike at the exact same time is zero.
\end{proof}
Second, at any time \(t\geq 0\), the distribution of \(S^{i,N}_t\) is absolutely continuous with respect to the Lebesgue measure \(\lambda\) and has a density \(\rho_t\) almost surely finite.
\begin{lem}\label{lem:densite-S-t}
	We denote by $\lVert \rho_0 \rVert_{\infty}$ the upper bound of $\rho_0$ and
	we assume that Assumption~\ref{hyp:cdt-init} holds.
	Then, for any $t \geq 0$,
	the distribution $\rho_t$ of $S_t^{i,N}$ also admits a density. In particular, for any $A
	\in \mathcal B(\reels^+)$,
	\[
	\rho_t(A) \leq ( \|\alpha\|_\infty + \lVert \rho_0 \rVert_{\infty} ) \lambda(A).
	\]
\end{lem}
\begin{proof}
	We have for all $A \in \mathcal B(\reels^+)$, $t > 0$,
	\[
	\proba(S_t^{i,N} \in A) = \proba(S_t^{i,N} \in A \cap [0,t[) 
	+ \proba(S_t^{i,N} \in A \cap [t,+\infty[).
	\]
	First, denoting $A-t = \{ x \in \reels^+, x + t \in A \}$, one has
	\[
	\left\{S_t^{i,N} \in A \cap [t,+\infty[ \right\} \subset 	\left\{S_0^{i,N} \in A - t \right\}.
	\]
	Thus, 
	\begin{align*}
	    \mathbb{P}		\left(S_t^{i,N} \in A \cap [t,+\infty[ \right)&\leq \mathbb{P}		\left(S_0^{i,N} \in A - t \right)
        \\
        &\leq \lVert \rho_0 \rVert_{\infty}  \lambda(A - t ) = \lVert \rho_0 \rVert_{\infty} \lambda(A).
	\end{align*}
	Second, the event
	\(
	\left\{S_t^{i,N} \in A \cap [0,t[ \right\} 
	\)
	means that the last spike of the neuron \(i\) occurred at a time in \(t - A\).
    The probability of this event is less than the probability that there is at least one spike of the neuron $i$ in $t - A$. 
	So, one has,
	\begin{align*}
	\proba(S_t^{i,N} \in & A \cap [0,t[)
    \\
	&  \leq \mathbb{E}\left[1-\exp\left(
	-\int_{t-A} \mathbbm 1_{ \{ V_u^{i,N} = 0 \} } \alpha(I_u^{i,N})du
	\right)\right]
    \\
	&\leq \mathbb{E}\left[\int_{t-A} \alpha(I_u^{i,N})du\right]\\
&	\leq \|\alpha\|_\infty \lambda(t -A) =  \|\alpha\|_\infty \lambda(A)
	.
	\end{align*}
\end{proof}
\begin{widetext}
\subsubsection{\texorpdfstring{$(\mu_t^N)_{t \geq 0}$}{TEXT} is Markov}
\begin{proposition}\label{prop:mu-Markov}
The random process $(\mu_t^N)_{t \geq 0}$ is a Markov process on $\mathcal P_N(E)$.
\end{proposition}
\begin{proof}
    We propose a constructive proof by showing how to construct such a process. If we know the process at time $t_0$, we need to find its next jumping time that we denote by $\tau$. This time is obtained by drawing a random variable following an exponential distribution with parameter 
    \[
    \lambda_{tot} = \sum_{(V,S,\xi) \in \supp(\mu_{t_0}^{N})} \beta V + \alpha\big(I(\xi)\big)(1-V).
    \]
    Until this jumping time $\tau$, we have 
    \[
        \forall t_0 \leq t < \tau, \quad
        \mu_t^N = \frac{1}{N} \sum_{(V,S,\xi) \in \supp(\mu_{t_0}^{N})} \delta_{\left( V,\ S + t -t_0,\
            \frac{1}{N} \sum\limits_{ (\tilde{V}, \tilde{S}, \tilde{W}) \in \supp(\xi) }
            \delta_{\left( \tilde V,\ \tilde S + t - t_0,\ \tilde W \right)} \right)}.
    \]
    The jump at time $\tau$ is first associated to a neuron $\hat X = (\hat{V},\hat{S},\hat{\xi}) \in \supp(\mu_{t_0}^{N})$ for which $\hat{V}$  jumps to $1-\hat{V}$. 
    Any such neuron has probability $\frac{\beta \hat{V} + \alpha\big(I(\hat{\xi})\big)(1-\hat{V})}{\lambda_{tot}}$ to be chosen. Hence, in addition to the jump of
    $\hat{V}$, $\hat{S}$ jumps to \(0\) if \(\hat{V} = 0\). It stays in
    $\hat{S}+\tau-t_0$ if \(\hat{V} = 1\). Furthermore, all the $\xi \in \supp(\mu_{\tau^-}^{N})$ jump. 
    To account for these changes, recall that each neuron can be identified by its \(\hat{S}+\tau-t_0\) in all
    $\xi \in \supp(\mu_{\tau^-}^{N})$ (see Remark~\ref{rem:20210215}).
    We then have to control the potentiations and the depressions: for any \((V,S,\xi) \in \supp{(\mu_{\tau^-}^{N})}\), we 
    introduce independent random variables \(\tilde{U}^+(S)\) and \(\tilde{U}^-(S)\), uniformly distributed on \([0,1]\).
    We can now describe the jumps of \(\xi\). First, if \(\hat{V} = 1\), we have for any \(\xi\)
    \begin{align*}
    \begin{split}
    \Delta \xi  =
    \frac{1}{N}\sum_{ (\tilde{V}, \tilde{S}, \tilde{W}) \in \supp(\xi)}
    \Biggl\{ \mathbbm{1}_{\{\tilde S = \hat{S} + \tau - t_0 \}}
    \big[
    \delta_{\big(0,\tilde S, \tilde W\big)}
    -
    \delta_{\big(1,\tilde S, \tilde W\big)}    
    \big]
    \Biggr\}.
    \end{split}    
\end{align*}
Then, if \(\hat{V} = 0\), for the \textit{spiking} neuron $\hat{X}$, that is for    
\((V, S, \xi)  \in \supp{(\mu_{\tau^-}^{N})}\)  with  \(S = \hat{S} + \tau - t_0\), we have
    \begin{align*}
    \begin{split}
    \Delta \xi  =
    \frac{1}{N}\sum_{ (\tilde{V}, \tilde{S}, \tilde{W}) \in \supp(\xi)}
    \Biggl\{
    &
    \mathbbm 1_{\{\tilde S = \hat{S} + \tau - t_0\}}
    \big[
    \delta_{\big(
        1,0, \tilde W 
        + \mathbbm 1_{\{ \tilde{U}^+(\tilde{S}) \leq p^+(\tilde{S},\tilde{W})\}}
        - \mathbbm 1_{\{ \tilde{U}^-(\tilde{S}) \leq p^-(\tilde{S},\tilde{W})\}}
        \big)} 
        - \delta_{\big(
        0,\tilde S, \tilde W 
        \big)} 
    \big]
    \\
    & 
    + 
    \mathbbm 1_{\{\tilde S \neq \hat{S} + \tau - t_0\}}
    \big[
    \delta_{\big(
        \tilde{V},\tilde{S}, \tilde W 
        + \mathbbm 1_{\{ \tilde{U}^+(\tilde{S}) \leq p^+(\tilde{S},\tilde{W})\}}
        \big)} 
        -
            \delta_{\big(
        \tilde{V},\tilde S, \tilde W 
        \big)} 
    \big]
    \Biggr\},
    \end{split}
    \end{align*}
and for all $(V, S, \xi)  \in \supp{(\mu_{\tau^-}^{N})}$ with $S \neq \hat{S} + \tau - t_0$,
    \begin{align*}
    \Delta \xi  =
    \frac{1}{N}\sum_{ (\tilde{V}, \tilde{S}, \tilde{W}) \in \supp(\xi)}
    \Biggl\{
    & 
    \mathbbm 1_{\{ \tilde S = \hat{S}  + \tau - t_0\}}
    \big[
    \delta_{\big(
        1,0, \tilde W 
        - \mathbbm 1_{\{ \tilde{U}^-(\tilde{S}) \leq p^-(S,\tilde{W})\}}
        \big)} 
        -
            \delta_{\big(
        0,\tilde S, \tilde W 
        \big)} 
    \big]
    \Biggr\}.\quad \quad \quad \quad \quad \quad 
\end{align*}
This completes the construction of the dynamics of the process $(\mu_{t}^{N})_{t\geq t_0}$ for any $t_0$.
\end{proof}
\end{widetext}

\twocolumngrid
\subsection{Preliminaries for the derivation of the main result}
The purpose of this section is to introduce the necessary tools for deriving the conjectured evolution equations of the
typical neuron $(X_t^*)_{t \geq 0}$ from those of the finite-size neural network
$(X_t^N)_{t \geq 0}$. To do so, we first detail the necessary definitions and notations, before stating the assumptions needed.

\subsubsection*{Definitions and notations}
We seek to identify possible deterministic limit processes
$(\mu_t^*)_{t \geq 0}$ of $(\mu_t^N)_{t \geq 0}$ when $N$ tends to infinity, for the weak topology.
To do so, for $T>0$, we consider the
empirical distributions on $D_E[0,T]$,
the space of c\`adl\`ag functions from $[0,T]$ to $E$.
\begin{nota}\label{def:mesure occupation}
	Let $Z$ be a complete separable metric space.
	We denote by $D_Z[0,T]$ the space of \textit{c\`adl\`ag} functions (right 
	continuous with left limits)
	from $[0,T]$ to $Z$.
\end{nota}
We define 
\[
    X^{i,N} = (X_t^{i,N})_{0 \leq t \leq T} \text{ and } X^{*} = (X_t^{*})_{0 \leq t \leq T},
\]
with the distribution of $X_t^*$ being $\mu_t^*$, which we denote by $\mathcal L(X_t^*) = \mu_t^*$. Thus, we can denote by 
\[
    \mu^{N} = \frac{1}{N} \sum_{i=1}^N \delta_{X^{i,N}} \ \text{ and } \ \mu^* = \mathcal L(X^*).
\]
Thereby, we consider the convergence in distribution of the sequences
$(\mu^N)_{N\geq 1}\in \mathcal{P}(D_{E}[0,T])^\mathbb{N}$ where we recall that
\[
	E = \{0,1\} \times \reels^+ \times 
	\mathcal P\big(\{0,1\} \times \reels^+ \times \relatifs\big).
\]
The space $D_E[0,T]$ is equipped with the \(J_1\) Skorohod  topology.
We also consider a distance \(d\) on $D_E[0,T]$ with the two properties:
first	\(d\) induces the \(J_1\) Skorohod topology;
second, the space	\((D_E[0,T],d)\) is a separable and complete space (i.e a Polish 
space).
The existence of such a distance is proved in Billingsley~\cite[Sec 
12]{billingsley_convergence_1999}. We also use the property that if 
a space $G$ is Polish, then $\mathcal P(G)$ equipped with the associated weak 
convergence, is also Polish \cite[Thm 17.23]{kechris_class_1995}.
\\
Assuming that a limit $\mu^* \in \mathcal{P}(D_{E}[0,T])$ is deterministic (i.e. 
$\mathcal L(\mu^N) \rightarrow_{N_{\infty}} \delta_{\mu^*}$),
the aim of this article is to find the system satisfied by any limit point which satisfies 
\begin{align}\label{ass:hypothese-integrale}
    \begin{split}
    	\forall & \Psi  \in C_b^{1,1}(E),\\
    	&\espe{\int_0^T \langle \mu_u^N, \Psi \rangle du} \rightarrow_{N_{\infty}} 
    	\espe{\int_0^T\langle \mu_u^*, \Psi \rangle du}
        \\
        & \qquad \qquad \qquad \qquad \qquad \quad 
        = \int_0^T\langle \mu_u^*, \Psi \rangle du,
    \end{split}
\end{align}
where the derivation is in the Fréchet sense (see the $C_b^{1,1}$ Notation~\ref{not:defintion de Cb11} below).
Indeed, knowing $\langle \mu_t^*, \Psi \rangle$ for all $\Psi \in C_b^{1,1}(E)$ is sufficient 
to determine the distribution $\mu_t^*$.
Hence, we restrict ourselves to
\textbf{the study of }$\bf{\mathbb E \langle \mu_t^N, \Psi \rangle}$ for $\Psi$ in $C_b^{1,1}(E)$.
Thereby, we are looking for the limit equation of a
\textit{typical neuron} that we denote by $X_t^* = (V_t^*, S_t^*, \xi_t^*)$ with distribution $\mu_t^*$.

We recall the notion of Fréchet differentiation.
We use the definition given in \cite{hale1980ordinary} at the beginning of page $6$.
\begin{definition}\label{def:frechet-diff}
	Let \(F\) and \(G\) be two Banach spaces embedded with the norms $\lVert\cdot\rVert_F$ and $\lVert\cdot\rVert_G$.
	Let $D$ be an open set of \(F\).
	We say that $\varphi:D \rightarrow G$ is Fréchet differentiable at $y \in D$ if there exists a linear bounded operator $L(y): h \in F \mapsto L(y)h \in G$ such that for all
	\( h \in F \) satisfying \( y + h \in D \), we have
	\[
	\lVert \varphi ( y + h ) - \varphi ( y ) 
	- L(y)h \rVert_{G}
	\leq \rho(\lVert h \rVert_{F}, y),
	\]
	where 
      \( \rho(\lVert h \rVert_{F}, y)\in \smallO{\lVert h \rVert_{F}}\).
    
    We call \( L(y) \) the Fréchet derivative of $\varphi$ at $y$ and 
	$L(y)h$ the Fréchet differential of $\varphi$ at $y$ in the direction $h$. 
	
    We say that $\varphi$ is Fréchet differentiable on $D$ if for all $y \in D$,
	$\varphi$ is Fréchet differentiable at $y$.
    \end{definition}
Then, we can define a space suitable for our following computations.
\begin{nota}\label{not:defintion de Cb11}
	We denote by $C_b^{1,1}(E)$ the space of functions from $E$ 
	to $\reels$ that are bounded, continuously
	differentiable with respect to their second variable, Fréchet
	differentiable (see Definition~\ref{def:frechet-diff}) with respect
	to their third variable and finally, 
	both these derivatives are bounded.
\end{nota}
The Fréchet derivative
with respect to the third variable is defined on the larger space
of the signed distribution on $E_m$, $\mathcal M(E_m)$, equipped with the
total variation norm, $\lVert \cdot \rVert_{TV}$.\begin{definition}\label{def:Jordan decomposition et variation totale}
	Let $(X, \mathcal{X})$ be a measurable space equipped with a signed distribution $\eta$. 
    One associates to $\eta$ the (unique) pair of positive
	distributions $\eta^+$ and $\eta^-$ defined for any \(A \in\mathcal{X}\) by
	\begin{align*}
	\eta^+(A) &= \sup\{\eta(B), B\in \mathcal{X}, B\subset 
	A\}\\
	\eta^-(A) &= \sup\{-\eta(B), B\in \mathcal{X}, B\subset 
	A\}.
	\end{align*}
	The Jordan decomposition of the signed distribution \(\eta\) writes
	\[
	\eta(A)	 = \eta^+(A) - \eta^-(A).
	\]
	The variation \(|\eta|\) of the signed distribution \(\eta\) is defined, for all \(A \in 
	\mathcal{X}\), by
	\[
	\lvert \eta \rvert (A) \coloneqq  \eta^+(A) + \eta^-(A).
	\]
	We define
    the total variation of \(\eta\) as
	\[
	\lVert \eta \rVert_{TV} \coloneqq  \lvert \eta \rvert (X) = \eta^+(X) + \eta^-(X).
	\]
	The space of signed distributions embedded with the norm $\lVert\cdot\rVert_{TV}$ is a Banach space. 
\end{definition}
In particular, $\Big(\mathcal M(E_m), \lVert \cdot \rVert_{TV}\Big)$
is a Banach space, see \cite[Rk 1.7]{ambrosio_functions_2000},
enabling us to define the Fréchet derivative of
$\Psi \in C_b^{1,1}(E)$.
For all $\Psi \in \mathcal C_b^{1,1}(E)$,
for all $(v,s,\xi) \in E$ and $h \in \mathcal M(E_m)$,
we denote by 
$\partial_{\xi} \Psi(v,s,\xi) \cdot h$
the Fréchet derivative of $\Psi$ at
$(v,s,\xi)$ in the direction $h$ (see Definition~\ref{def:frechet-diff}).

Now, we define two operators \(\eta^N \) and \(\zeta^N\)
to describe the jumps of the empirical measure when 
one neuron \(V\) jumps from \(0\) to \(1\) or from \(1\) to \(0\).
They both associate to a pair
	$(s,\xi) \in \reels^+ \times \mathcal P_N(E_m)$
	the following signed distributions on $E_m$:
	\begin{align*}
	\eta^N ( s, \xi )
	& \coloneqq  \frac{1}{N} \sum_{(1, s, \tilde w) \in \supp( \xi )} 
	\Big[ \delta_{\big( 0, s, \tilde w \big)} 
	- \delta_{\big( 1, s, \tilde w \big)} \Big] \\
    \zeta^N ( s, \xi ) 
	& \coloneqq  \frac{1}{N} \sum_{(0,  s, \tilde w) \in \supp( \xi )} 
	\Big[ \delta_{\big( 1, 0, \tilde w \big)} 
	- \delta_{\big( 0, s, \tilde w \big)} \Big].
	\end{align*}

\begin{nota}\label{not:egalite-en-esp}
	The symbol $\overset{\esp}{=}$ refers to the equality between the expectations of
	two random variables: 
	\[
	\espe{X} = \espe{Y} \quad \Leftrightarrow \quad X \overset{\esp}{=} Y.
	\]
\end{nota}

\begin{remarque}
	Let $\tau$ be any spiking time of the neuron $i$, i.e. \(V^{i,N}_{\tau-} = 0\) and 
	\(V^{i,N}_{\tau} = 1\). Then, according to
	Lemma~\ref{lem:supp-s-distinct}, for every $j$, there is a unique
	$(\tilde{v},\tilde{s},\tilde{w})$ in the support of \(\xi_{\tau^-}^{j,N}\) 
	such that \(\tilde{s} = S^{i,N}_{\tau^{-}}\). So, 
	$\zeta^N(S_{\tau^-}^{i,N},\xi_{\tau^-}^{j,N})$
	is a signed distribution with two atoms.
	\[
	\zeta^N ( S_{\tau^-}^{i,N}, \xi_{\tau^-}^{j,N})  =
	\frac{1}{N}
	\delta_{\left( 1, 0, \tilde{w} \right)} 
	- \frac{1}{N}\delta_{\left( 0, S_{\tau^-}^{i,N}, \tilde{w} \right)}.
	\]
	Similarly, at a time $t$ when $V^{i,N}_{\tau^-} = 1$  and 
	$V^{i,N}_{\tau} = 0$, there is a unique 
	$(\tilde{v},\tilde{s},\tilde{w})$ in the support of \(\xi_{\tau^-}^{j,N}\) 
	such that \(\tilde{s} = S^{i,N}_{\tau^{-}}\) and $\eta^N(S_{\tau^-}^{i,N},\xi_{\tau^-}^{j,N})$ is a signed 
	distribution
	with two atoms
	\[
	\eta^N ( S_{\tau^-}^{i,N}, \xi_{\tau^-}^{j,N})   =
	\frac{1}{N}
	\delta_{\left(0,S_{\tau^-}^{i,N},  \tilde{w} \right)} 
	- \frac{1}{N}\delta_{\left( 1, S_{\tau^-}^{i,N}, \tilde{w} \right)}.
	\]
	We define here $\eta^N$ and $\zeta^N$ only on $\mathcal P_N(E_m)$, the set of empirical distributions of order
	$N$ over $E_m$. With the changes occurring when passing to the large $N$
	limit, these distributions will no longer be required. New ones are defined on the full
	spaces $\mathcal P(E_m)$ later on, see Propositions~\ref{prop:limite-terme-1} and
	\ref{prop:limite-terme-2}.
\end{remarque}
Third, we define the derivative \(\varepsilon'\) of any probability distribution $\varepsilon$ on $\reels^+$ using the space 
of continuous functions with bounded derivative $C_b^1$ by,
\begin{equation}\label{eq:def-derivee-proba}
\forall \phi\in C_b^1(\reels^+), \quad
\langle \varepsilon', \phi \rangle \coloneqq   - \langle \varepsilon, \phi' \rangle.
\end{equation}
Finally, let us define for all $\xi \in \mathcal P(E_m)$ and $t \geq 0$,
the probability distribution $\xi \oplus t$ such that
\begin{multline*}
    \forall (v,A,w) \in \mathcal B(\{0,1\},[t,+\infty[, \relatifs), \qquad \\
\xi \oplus t(v,A,w) = \xi (v,A-t,w).
\end{multline*}

\begin{nota}\label{nota:produitmesure}
	For all $x,y \in \{0,1\}$, we denote by $\delta_x \otimes \xi$ the
	distribution such that
	\begin{multline*}
	    \forall (v,A,w) \in \{0,1\} \times \mathcal B(\reels^+)
	    \times \relatifs, \\ \{\delta_x \otimes \xi^y\}(v,A,w) \coloneqq \delta_x(v)\xi(y, A, w).
	\end{multline*}
	Finally, when $\mu_t^*$ and $\xi_t^*$ admit densities in $s$, we use the abuse of notation 
	\begin{align*}
	    \mu_t^*(v, ds,d\xi) & = \mu_t^*(v, s,d\xi)ds
        \\
        {\xi_t^*}(v, ds,w) & = {\xi_t^*}(v, s,w)ds.
	\end{align*}
\end{nota}

\subsubsection*{Assumptions}

\begin{assumption}\label{ass:lim-mu-star-1}
		Assume that for all $T>0$, the sequence of the system's empirical distributions $(\mu^N)_{N\geq 1}\in \mathcal{P}(D_{E}[0,T])^\mathbb{N}$
		converges in distribution to a (deterministic) probability distribution,
		$\mu^*\in \mathcal{P}(D_{E}[0,T])$, when $N$ tends to infinity. 
		In particular, the limit \eqref{ass:hypothese-integrale} holds.
\end{assumption}
Then, we need the following assumptions for our computations to hold:
\begin{assumption}\label{ass:lim-mu-star-2}
	The functions $\xi \mapsto \alpha(I(\xi))$ and for all $w \in \relatifs$,
	$s\mapsto p^{\pm}(s,w)$, are bounded continuous functions.
\end{assumption}
Our calculations concerning the dynamics of $\esp \langle \mu_t^N, \Psi \rangle$ enable us to conjecture
the limit distribution dynamics. In order to make this conjecture a theorem, we still need to prove
the convergence as $N$ tends to infinity of some terms of this dynamics. It will be explained in further
details just after Proposition~\ref{prop:limite-terme-2} for one term and just before Conjecture~\ref{th:equation limite mu sans plasticite} for another one.
Our Main Result \ref{th:equation limite systeme de particules avec plasticite} concerns the dynamics of the typical neuron $(X_t^*)_{t \geq 0}$.
In particular, it exposes the PDE that we expect to be satisfied by the density of $\xi_t^*$ in $s$
and requires a compatibility assumption: the boundary condition of the mean field must be satisfied at time $t=0$.
\begin{assumption}\label{ass:compabilite EDP}
	Assume that for all $t \geq 0$, $\xi_t^*$ admits a density of class $C^1$ in $s$
	and in particular at time $t=0$, for all $w \in \relatifs$,
	the following densities satisfy the boundary conditions for $\mu_0^*$-almost all $(V_0^*, S_0^*, \xi_0^*)$,
	\begin{align*}
    \hspace{-1.5em}
	\begin{split}
	&{\xi_0^*}(0,0,w) = 0
	\\
	&{\xi_0^*}(1,0,w) = 
    \\ 
    & \  
    \int_{\reels^+ \times \mathcal{P}(E_m)}
    		\alpha\big(I(\xi')\big)
    		\frac{p^-(S_0^*,w+1)\xi_0^{*}(0,s',w+1)}{\xi^*_0(0,s',\relatifs)}
    		\mu_0^{*}(0,s',d\xi')ds'
    \\ 
    & \ + 
    \int_{\reels^+ \times \mathcal{P}(E_m)}
    		\alpha\big(I(\xi')\big)
    		\frac{(1-p^-(S_0^*,w))\xi_0^{*}(0,s',w)}{{\xi_0^{*}}(0,s',\relatifs)}
    		\mu_0^{*}(0,s',d\xi')ds'.
	\end{split}
	\end{align*}
\end{assumption}

\begin{widetext}
\subsection{The case without plasticity}
\label{supp:The case without plasticity}
The purpose of this section is to conjecture the evolution equations of the
typical neuron $(X_t^*)_{t \geq 0}$ 
when $p^+ \equiv p^- \equiv 0$ (without plasticity).
To do so, we first derive the infinitesimal generator of $(\mu_t^N)_{t \geq 0}$ and second, 
derive the dynamics of $\mathbb E \langle \mu_t^N, \Psi \rangle$ where $\Psi$ is a test function with properties described below.
Then, we study in detail the most complex terms of this dynamics.
Finally, we \textcolor{black}{conjecture} the dynamics of any limit point $(X_t^*)_{t \geq 0}$.

\subsubsection{Generator of \texorpdfstring{$(\mu_t^N)_{t \geq 0}$}{TEXT} }
We now look for the generator of the Markov process $(\mu_t^N)_{t \geq 0}$.
We describe the generator only on the set of 
cylindrical 
functions $\Phi \in C_b(\mathcal P(E))$ that is there exists $\Psi \in C^{1,1}_b(E)$ such that
\[
\forall \mu \in \mathcal P(E), \qquad
\Phi(\mu) = \langle\mu, \Psi\rangle.
\]
We consider the first jump of $(\mu_t^N)_{t \geq 0}$ that we denote by $\tau > 0$. We split the generator into the drift term and the jump term using the equality $1 = \mathbbm 1_{t < \tau} + \mathbbm 1_{t \geq \tau}$, hence splitting the generator (obtained when $t$ tends to $0$) of $(\mu_t^N)_{t \geq 0}$ into a drift term \(\circled{D}\) and a jump term \(\circled{J}\).
Let $\mu_0 \in \mathcal P_N(E)$ such that
$\mu_0 = \frac{1}{N} \sum_i \delta_{(v^i,s^i,\xi^i)}$ with
$(v^1,s^1,\xi^1), \dots, (v^N,s^N,\xi^N)$ are in $E$.
Let $\Psi \in C_b^{1,1}(E)$, then

\begin{align*}
\langle \mu_t^N - \mu_0, \Psi\rangle &= 
\langle\mathbbm 1_{t < \tau}\mu_t^N - \mu_0 , \Psi \rangle
+
\langle\mathbbm 1_{t \geq \tau} \mu_t^N, \Psi \rangle
= 
\frac{1}{N} \sum_i 
\left(
\Psi \big(v^i,s^i+t,\xi^i \oplus t\big)\mathbbm 1_{t < \tau}
- \Psi \big(v^i,s^i,\xi^i\big)
\right)
+
\langle\mathbbm 1_{t \geq \tau} \mu_t^N, \Psi \rangle
\\
& = \underbrace{\frac{1}{N} \sum_i 
	\left(
	\Psi \big(v^i,s^i+t,\xi^i \oplus t\big)
	- \Psi \big(v^i,s^i,\xi^i\big)
	\right)}_{\circled{D}}
+
\underbrace{\langle\mathbbm 1_{t \geq \tau} \mu_t^N, \Psi \rangle
	-\frac{1}{N} \sum_i 
	\Psi \big(v^i,s^i+t,\xi^i \oplus t\big)\mathbbm 1_{t \geq \tau}}_{\circled{J}}.
\end{align*}
Thereby, we obtain that the drift term satisfies, with $- \partial_s \xi$ the derivative defined as in \eqref{eq:def-derivee-proba},
\begin{align*}
\lim_{t \to 0} \frac{\esp\circled{D}}{t}
&=
\lim_{t \to 0} \frac{1}{N} \sum_i 
\frac{\espe{ \Psi \big(v^i,s^i+t,\xi^i \oplus t\big)} - 
	\Psi \big(v^i,s^i,\xi^i\big)}{t}
\\
&
= \frac{1}{N} \sum_i 
\partial_s\Psi \big(v^i,s^i,\xi^i\big)
+ \partial_{\xi}\Psi \big(v^i,s^i,\xi^i\big)
\cdot (-\partial_s \xi^i)
\\
&=
\int_{ E } \Big( \partial_s \Psi \left(v,s,\xi\right)
+\partial_{\xi} \Psi \left(v,s,\xi\right) \cdot 
(- \partial_s \xi) \Big) \mu_0(dv,ds,d\xi).
\end{align*}

For the jump term \circled{J}, when we take the limit as $t$ tends to $0$ in
$\frac{\esp\circled{J}}{t}$, the only term left is the one due to the
first jump at time $\tau$. Indeed, all the other terms are of order $t$ or
more. We thus obtain that
\begin{align*}
&\lim_{t \to 0} \frac{\esp\circled{J}}{t} = \frac{1}{N} \sum_i 
\mathbbm 1_{\{v^i = 1\}}  \beta \Biggl\{
\Big[ \Psi \Big( 0,s^i, \xi^i
+ \eta^N (s^i, \xi^i ) \Big) 
- \Psi \Big(1, s^i, \xi^i \Big) \Big]
\\
& \qquad \qquad \qquad \qquad \qquad \qquad
+ \sum_{j \neq i}
\Big[ \Psi \Big( v^j, s^j, \xi^j 
+ \eta^N ( s^i, \xi^j ) \Big) 
- \Psi \Big( v^j, s^j, \xi^j \Big) \Big]
\Biggr\}
\\
& \qquad \qquad \quad
+ \frac{1}{N} \sum_i \mathbbm 1_{\{v^i = 0\}}
\alpha(I(\xi^i)) \Biggl\{
\Big[
\Psi \Big( 1, 0,\xi^i 
+ \zeta^N ( s^i, \xi^i ) \Big) 
- \Psi \Big(0,s^i,\xi^i\Big) 
\Big]
\\
& \qquad \qquad \qquad \qquad \qquad \qquad\qquad\qquad
+ \sum_{j\neq i}
\Big[ \Psi \Big( v^j, s^j, \xi^j 
+ \zeta^N ( s^i, \xi^j ) \Big) 
- \Psi \Big( v^j, s^j, \xi^j \Big) \Big]
\Biggr\}.
\end{align*}
In the previous equation, by adding and removing the missing terms
in the sums inside the braces, we get
\begin{align*}
&\lim_{t \to 0} \frac{\esp\circled{J}}{t} = \frac{1}{N} \sum_i 
\mathbbm 1_{\{v^i = 1\}}  \beta \Biggl\{
\Big[ \Psi \Big( 0,s^i, \xi^i
+ \eta^N (s^i, \xi^i ) \Big) 
- \Psi \Big(1, s^i, \xi^i + \eta^N (s^i, \xi^i ) \Big) \Big]
\\
& \qquad \qquad \qquad \qquad \qquad \qquad
+ \sum_{j}
\Big[ \Psi \Big( v^j, s^j, \xi^j 
+ \eta^N ( s^i, \xi^j ) \Big) 
- \Psi \Big( v^j, s^j, \xi^j \Big) \Big]
\Biggr\}
\\
& \qquad \qquad \quad
+ \frac{1}{N} \sum_i \mathbbm 1_{\{v^i = 0\}}
\alpha(I(\xi^i)) \Biggl\{
\Big[
\Psi \Big( 1, 0,\xi^i 
+ \zeta^N ( s^i, \xi^i ) \Big) 
- \Psi \Big(0,s^i,\xi^i + \zeta^N ( s^i, \xi^i )\Big) 
\Big]
\\
& \qquad \qquad \qquad \qquad \qquad \qquad\qquad\qquad
+ \sum_{j}
\Big[ \Psi \Big( v^j, s^j, \xi^j 
+ \zeta^N ( s^i, \xi^j ) \Big) 
- \Psi \Big( v^j, s^j, \xi^j \Big) \Big]
\Biggr\}
\\
&
= \int_{\reels^+ \times \mathcal{P}(E_m)} \beta
\Big[ \Psi \Big( 0, s, \xi + \eta^N ( s, \xi ) \Big) 
- \Psi \Big( 1, s, \xi + \eta^N ( s, \xi ) \Big) \Big]
{\mu_{0}}(1,ds,d\xi)
\\
& \quad
+ \frac{1}{N} \sum_i 
\mathbbm 1_{\{v^i = 1\}}  \beta \sum_{j}
\Big[ \Psi \Big( v^j, s^j, \xi^j 
+ \eta^N ( s^i, \xi^j ) \Big) 
- \Psi \Big( v^j, s^j, \xi^j \Big) \Big]
\\
& \quad
+ \int_{\reels^+ \times \mathcal{P}(E_m)} \alpha\big(I(\xi)\big) 
\Big[ \Psi \Big( 1, 0,\xi + \zeta^N ( s,\xi ) \Big) 
- \Psi \Big( 0,s,\xi + \zeta^N ( s,\xi ) \Big) \Big]
\mu_{0}(0,ds,d\xi)
\\
& \quad
+ \frac{1}{N} \sum_i \mathbbm 1_{\{v^i = 0\}}
\alpha(I(\xi^i)) \sum_{j}
\Big[ \Psi \Big( v^j, s^j, \xi^j 
+ \zeta^N ( s^i, \xi^j ) \Big) 
- \Psi \Big( v^j, s^j, \xi^j \Big) \Big].
\end{align*}
We can now deduce the dynamics of $\esp \langle \mu_t^N, \Psi\rangle$ from \(\circled{D}\), \(\circled{J}\) and the Markov property of the process $(\mu_t^N)_{t \geq 0}$:

for all $\Psi \in C_b^{1,1}(E)$,
\begin{align}\label{eq:mu-t-N}
\begin{split}
\langle \mu_t^N, \Psi\rangle
&\overset{\esp}{=}
\langle \mu_0^N, \Psi\rangle 
+ \int_0^t \int_E 
\Big(
\partial_s \Psi(v,s, \xi) + \partial_{\xi} \Psi(v,s, \xi) \cdot (-\partial_s \xi) 
\Big)
\mu_u^N(dv,ds,d\xi) du
\\
& \quad
+ \int_0^t \int_{\reels^+ \times \mathcal{P}(E_m)} \beta
\Big[ \Psi \Big( 0, s, \xi + \eta^N ( s, \xi ) \Big) 
- \Psi \Big( 1, s, \xi + \eta^N ( s, \xi ) \Big) \Big]
{\mu_{u^-}^N}(1,ds,d\xi) du
+ \circled{1}
\\
& \quad
+ \int_0^t \int_{\reels^+ \times \mathcal{P}(E_m)} \alpha\big(I(\xi)\big) 
\Big[ \Psi \Big( 1, 0,\xi + \zeta^N ( s,\xi ) \Big) 
- \Psi \Big( 0,s,\xi + \zeta^N ( s,\xi ) \Big) \Big]
\mu_{u^-}^N(0,ds,d\xi) du
+ \circled{2}
\end{split}
\end{align}
where $\circled{1}$ and $\circled{2}$ are the most complex elements:
\begin{align}
\begin{split}
\circled{1} & \coloneqq  \int_0^t \sum_i \mathbbm 1_{\{V_{u^-}^{i,N} = 1\}}
\frac{\beta}{N} \sum_j
\Big[ \Psi \Big( V_{u^-}^{j,N}, S_{u^-}^{j,N}, \xi_{u^-}^{j,N} 
+ \eta^N ( S_{u^-}^{i,N}, \xi_{u^-}^{j,N} ) \Big)
- \Psi \Big( V_{u^-}^{j,N}, S_{u^-}^{j,N}, \xi_{u^-}^{j,N} \Big) \Big] du,
\end{split}
\\
\begin{split}
\circled{2} & \coloneqq  \int_0^t \sum_i \mathbbm 1_{\{V_{u^-}^{i,N} = 0\}} 
\frac{\alpha(I_{u^-}^{i,N})}{N} \sum_j
\Big[ \Psi \Big( V_{u^-}^{j,N}, S_{u^-}^{j,N}, \xi_{u^-}^{j,N} 
+ \zeta^N ( S_{u^-}^{i,N}, \xi_{u^-}^{j,N} ) \Big)
- \Psi \Big( V_{u^-}^{j,N}, S_{u^-}^{j,N}, \xi_{u^-}^{j,N} \Big) \Big] du.
\end{split}
\label{eq:definitionde2entoure}
\end{align}

\subsubsection{Drift term due to the returns to resting potential}
When a neuron returns to its resting potential state $0$, the total variation
of the jumps of the empirical distributions 
\((\xi_t^{1,N}),\cdots,(\xi_t^{N,N})\) 
are of order
$\frac{1}{N}$. Hence, as $N$ tends to infinity, the limit of the term
\circled{1} depends on the Fréchet derivative of $\Psi$. The direction of this
derivative describes the way some of the mass of the distribution $\xi^*$ is
transported. This transport can be described informally as follows:
between time $t=0$ and time $t=\epsilon$ small, for
all $s,\ w$, a proportion $\beta \epsilon$ of the mass of $\xi_0^*$ that was
in $(1,s,w)$ is transported to $(0,s,w)$.
\begin{proposition}\label{prop:limite-terme-1}
	Under Assumptions~\ref{hyp:cdt-init} and~\ref{ass:lim-mu-star-1},
	for all $\Psi \in C_b^{1,1}(E)$,
	\begin{equation}\label{eq:limite terme 1}
	\lim_{N \rightarrow \infty} \circled{1}
	\overset{\esp}{=}
	\int_0^t \int_E  
	\partial_{\xi} \Psi ( v,s,\xi ) \cdot (\beta \delta_0 \otimes \xi^1 - \beta \delta_1 \otimes \xi^1) 
	\mu_{u^-}^*(dv,ds,d\xi) du.
	\end{equation}
\end{proposition}
\begin{proof}
	By Lemma~\ref{lem:supp-s-distinct}, for all $j$, the support of
	$\xi_{u^-}^{j,N}$ almost surely does not possess two points with
	the same second coordinate. Hence, we have the following upper bound, for
	all $s \in \reels^+$, almost surely,
	\begin{equation}\label{eq:born-norm-eta-N}
	\lVert \eta^N ( s, \xi_{u^-}^{j,N} ) \rVert_{TV} \leq \frac 2N.
	\end{equation}
	We write the term of $\circled{1}$ which is between square brackets using the Fréchet derivative:
	\begin{align*}
	&\Psi \Big( V_{u^-}^{j,N}, S_{u^-}^{j,N}, \xi_{u^-}^{j,N} 
	+ \eta^N ( S_{u^-}^{i,N}, \xi_{u^-}^{j,N} ) \Big) 
	- \Psi \Big( V_{u^-}^{j,N}, S_{u^-}^{j,N}, \xi_{u^-}^{j,N} \Big)
	\\
	& \qquad \qquad
	= \partial_{\xi} \Psi (V_{u^-}^{j,N}, S_{u^-}^{j,N}, \xi_{u^-}^{j,N}) \cdot 
	\Big( \eta^N ( S_{u^-}^{i,N}, \xi_{u^-}^{j,N} ) \Big) 
	+ {\scriptstyle\mathcal{O}}_{N \infty} \Big(  
	\lVert \eta^N ( S_{u^-}^{i,N}, \xi_{u^-}^{j,N} ) \rVert_{TV} \Big).
	\end{align*}
	Hence, we have:
	\begin{align*}
	\circled{1} &\overset{\esp}{=} \int_0^t 
	\sum_i \mathbbm 1_{ \{V_{u^-}^{i,N} = 1\} }
	\frac{\beta}{N} \sum_j
	\Big[ \partial_{\xi} \Psi ( V_{u^-}^{j,N}, S_{u^-}^{j,N}, \xi_{u^-}^{j,N} )
	\cdot 
	\Big( \eta^N ( S_{u^-}^{i,N}, \xi_{u^-}^{j,N} ) \Big) 
	+ {\scriptstyle\mathcal{O}}_{N \infty} 
	\Big( \lVert \eta^N (S_{u^-}^{i,N}, \xi_{u^-}^{j,N} ) \rVert_{TV} \Big) \Big] du
	.
	\end{align*}
	Thus, with the upper bound \eqref{eq:born-norm-eta-N}, we obtain that
	\begin{align*}
	\circled{1} &\overset{\esp}{=} \int_0^t 
	\sum_i \mathbbm 1_{ \{V_{u^-}^{i,N} = 1\} }
	\frac{\beta}{N} \sum_j 
	\Big[ \partial_{\xi} \Psi ( V_{u^-}^{j,N}, S_{u^-}^{j,N}, \xi_{u^-}^{j,N} )
	\cdot 
	\Big( \eta^N ( S_{u^-}^{i,N}, \xi_{u^-}^{j,N} ) \Big) \Big] du
	+ {\scriptstyle\mathcal{O}}_{N \infty} ( 1  ).
	\end{align*}
	Now, by linearity of $\partial_{\xi} \Psi$ and inverting the sums, we get
	\begin{align*}
	\circled{1} &\overset{\esp}{=} \int_0^t 
	\frac{1}{N} \sum_j 
	\Big[ \partial_{\xi} \Psi ( V_{u^-}^{j,N}, S_{u^-}^{j,N}, \xi_{u^-}^{j,N} )
	\cdot \Big( \sum_i \mathbbm 1_{ \{V_{u^-}^{i,N} = 1\} }
	\beta 
	\eta^N (  S_{u^-}^{i,N}, \xi_{u^-}^{j,N} ) \Big) \Big] du 
	+ {\scriptstyle\mathcal{O}}_{N \infty} ( 1 ).
	\end{align*}
	From Remark~\ref{rem:egalite-mu-xi},
	for all $j$ and $u \geq 0$, $\xi_{u^-}^{j,N}$ and $\mu_{u^-}^N$
	have the same support in $V$ and $S$, and from
	Lemma~\ref{lem:supp-s-distinct}, these supports are $N$ almost surely 
	distinct points. Hence, almost surely, for any function
	$f:\reels^+ \times \relatifs\to \reels^+$,
	\[
	\forall j, \quad
	\sum_i\sum_{(1, S_{u^-}^{i,N}, \tilde w) \in \supp(\xi_{u^-}^{j,N})} 
	f(S_{u^-}^{i,N}, \tilde w)
	= \sum_{(1, \tilde s, \tilde w) \in \supp(\xi_{u^-}^{j,N})}
	f(\tilde s, \tilde w).
	\]	
	We deduce that almost surely,
	\begin{align*}
	\sum_i \mathbbm 1_{ \{V_{u^-}^{i,N} = 1\} }
	\beta  \eta^N (  S_{u^-}^{i,N}, \xi_{u^-}^{j,N} ) 
	&= \sum_i \frac{\beta}{N}
	\sum_{(1, S_{u^-}^{i,N}, \tilde w) \in \supp(\xi_{u^-}^{j,N})} 
	\Big[ \delta_{\big( 0, S_{u^-}^{i,N}, \tilde w \big)} 
	- \delta_{\big( 1, S_{u^-}^{i,N}, \tilde w \big)} \Big]
	\\
	&= \frac{\beta}{N}
	\sum_{(1, \tilde s, \tilde w) \in \supp(\xi_{u^-}^{j,N})} 
	\Big[ \delta_{\big( 0, \tilde s, \tilde w \big)} 
	- \delta_{\big( 1, \tilde s, \tilde w \big)} \Big]
	\\
	&= \beta \delta_0 \otimes {\xi_{u^-}^{j,N}}(1,\cdot,\cdot)
	- \beta \delta_1 \otimes {\xi_{u^-}^{j,N}}(1,\cdot,\cdot),
	\end{align*}
	where we used Notation~\ref{nota:produitmesure}.
	We deduce that
	\[
	\circled{1}
	\overset{\esp}{=}
	\int_0^t \int_E  
	\partial_{\xi} \Psi ( v,s,\xi ) \cdot (\beta \delta_0 \otimes \xi^1 - \beta \delta_1 \otimes \xi^1) 
	\ \mu_{u^-}^N(dv,ds,d\xi) du
	+ {\scriptstyle\mathcal{O}}_{N \infty} ( 1 ).
	\]
	\textcolor{black}{The application $\xi \mapsto \xi(1,\cdot,\cdot)$ is bounded continuous.
	Thus, as $\Psi \in C_b^{1,1}(E)$, the application
	\[
	(s,v,\xi) \mapsto \partial_{\xi} \Psi ( v,s,\xi ) \cdot (\beta \delta_0 \otimes \xi^1 - \beta \delta_1 \otimes \xi^1)
	\]
	is bounded continuous.
	We conclude that
	under Assumption~\ref{ass:lim-mu-star-1},
	\begin{equation*}
			\lim_{N \rightarrow \infty} \circled{1}
			\overset{\esp}{=}
	\int_0^t \int_E  
	\partial_{\xi} \Psi ( v,s,\xi ) \cdot (\beta \delta_0 \otimes \xi^1 - \beta \delta_1 \otimes \xi^1) 
	\mu_{u^-}^*(dv,ds,d\xi) du.
	\end{equation*}}
\end{proof}

\subsubsection{Drift term due to action potentials}
As $N$ tends to infinity, the limit of the term \circled{2} also gives
a drift term for $\xi^*$. This term is more difficult to deal with because
the spiking rates
$(\mathbbm 1_{\{ V_t^{i,N} = 0 \}}\alpha(I_t^{i,N}))_{1\leq i \leq N}$ depend
on the actual state of the neural network $X_t^N$.
Informally, the limit of the term \circled{2} represents
the following mass transport on the probability distributions: between time $t=0$
and time $t=\epsilon$ small, for all $s,\ w$, a proportion 
$\epsilon \int_{\mathcal P(E_m)} \frac{\alpha(I(\tilde\xi))
}{\mu_0^*(0,s,\mathcal P(E_m))} \mu_0^*(0,s,d\tilde \xi)$ of the mass of
$\xi_0^*$ that was at $(0,s,w)$ is transported to $(1,[0,\epsilon],w)$.
\begin{proposition}\label{prop:limite-terme-2}
	We assume that Assumptions~\ref{hyp:cdt-init} and~\ref{ass:lim-mu-star-2} hold.
	For all $\Psi$ Fréchet differentiable with respect to its third variable:
	\textcolor{black}{\[
	\circled{2} \overset{\esp}{=}
	\int_0^t \int_{E}
	\partial_{\xi} \Psi ( v,s,\xi )
	\cdot 
	\Big(\delta_1 \otimes \nu^1 ( {\xi},{\mu_{u^-}^N}) 
	- \delta_0 \otimes \nu^0 ( {\xi},{\mu_{u^-}^N})\Big)\mu_{u^-}^N(dv,ds,d\xi) du
	+ {\scriptstyle\mathcal{O}}_{N \infty} ( 1 ).
	\]}
	where $\nu^0$ and $\nu^1$ are defined in the proof: equations
	\eqref{eq:defnutilde0} and \eqref{eq:defnutilde1}.
\end{proposition}
\begin{proof}
	By Lemma~\ref{lem:supp-s-distinct}, for any $j$, the support of
	$\xi_{u^-}^{j,N}$ almost surely does not contain two atoms with the same
	$S$. Furthermore, we have assumed that $\alpha$ is bounded. Thus, we have
	the following upper bounds, for all
	$i,\ j \in \{1,\cdots,N\} $, almost surely,
	\begin{equation}\label{eq:bound nu N}
		\lVert\alpha(I_{u^-}^{i,N}) \zeta^N ( S_{u^-}^{i,N},\xi_{u^-}^{j,N} ) \rVert_{TV} 
		\leq \frac{2 \|\alpha\|_\infty}{N}.
	\end{equation}
	Using the Fréchet derivative as previously, from 
	\eqref{eq:definitionde2entoure}, we obtain:
	\begin{align*}
	\circled{2} \overset{\esp}{=} \int_0^t \frac{1}{N}
	\sum_j 
	\Big[ \partial_{\xi} \Psi ( V_{u^-}^{j,N}, S_{u^-}^{j,N}, \xi_{u^-}^{j,N} )
	\cdot \Big( \sum_i \mathbbm 1_{\{V_{u^-}^{i,N} = 0\}}
	\alpha(I_{u^-}^{i,N}) 
	\zeta^N ( S_{u^-}^{i,N},\xi_{u^-}^{j,N} )& \Big) \Big]  du
	+ {\scriptstyle\mathcal{O}}_{N \infty} (1).
	\end{align*}
	We now look at the term in parentheses, it gives the direction of
	the Fréchet derivative. It is a distribution on $E_m$. For all $j$, we denote by
	\begin{align*}
	\circled{2.a} & \coloneqq 
	\sum_i \mathbbm 1_{\{V_{u^-}^{i,N} = 0\}}
	\alpha(I_{u^-}^{i,N})  
	\zeta^N ( S_{u^-}^{i,N},\xi_{u^-}^{j,N} )
	\\
	&=
	\sum_i
	\mathbbm 1_{\{V_{u^-}^{i,N} = 0\}}
	\alpha(I_{u^-}^{i,N})
	\frac{1}{N}
	\sum_{(0, S_{u^-}^{i,N}, \tilde w) \in \supp( \xi_{u^-}^{j,N} )}
	\Big[ \delta_{\big( 1, 0, \tilde w \big)} 
	- \delta_{\big( 0, S_{u^-}^{i,N}, \tilde w \big)} \Big]
	\\
	&=
	\frac{1}{N}
	\sum_i
	\mathbbm 1_{\{V_{u^-}^{i,N} = 0\}}
	\alpha(I_{u^-}^{i,N})
    \sum_{\tilde w\in \relatifs} \mathbbm 1_{ \{(0, S_{u^-}^{i,N}, \tilde w) \in \supp( \xi_{u^-}^{j,N} )\} }
    \Big[ \delta_{\big( 1, 0, \tilde w \big)} 
	- \delta_{\big( 0, S_{u^-}^{i,N}, \tilde w \big)} \Big].
	\end{align*}
But 
\begin{equation}\label{eq:w_to_gamma}
    \mathbbm 1_{ \{(0, S_{u^-}^{i,N}, \tilde w) \in \supp( \xi_{u^-}^{j,N} )\} } = \frac{\frac 1N \sum_{(v, s, w)\in \supp( \xi_{u^-}^{j,N})} \mathbbm 1_{\{v=0,s=S_{u^-}^{i,N},w=\tilde w\}}}{\frac 1N \sum_{(v, s)\in \supp( \xi_{u^-}^{j,N}(\cdot,\cdot,\relatifs))} \mathbbm 1_{\{v=0,s=S_{u^-}^{i,N}\}}}
        = \frac{\xi_{u^-}^{j,N}(0, S_{u^-}^{i,N}, \tilde w)}{\xi_{u^-}^{j,N}(0, S_{u^-}^{i,N},\relatifs)},
\end{equation}
hence inciting us to define for any $w \in \relatifs$,
	the Radon-Nikodym derivative \(\gamma_\xi\)
    \begin{align}\label{eq:gamma-xi}
	\forall (s,w) \in \reels^+ \times \relatifs,\qquad
	\gamma_{\xi}(s,w) = \frac{d\xi(0,\cdot, w)}{d\xi(0,\cdot, \relatifs)}(s).
	\end{align}
	Note that for all $A \in \mathcal B(\reels^+)$,
	we have
	\[
	\xi(0, A, w) = \int_A \gamma_{\xi}(s,w)\xi(0, ds, \relatifs).
	\]
	As $\xi(0, A, w) \leq \xi(0, A, \relatifs)$, we deduce that $\gamma_{\xi} \leq 1$.
	In particular, 
    using \eqref{eq:w_to_gamma}, we obtain that  
	\[
		\mathbbm 1_{ \{(0, S_{u^-}^{i,N}, \tilde w) \in \supp( \xi_{u^-}^{j,N} )\} }
		= \gamma_{{\xi_{u^-}^{j,N}}}(S_{u^-}^{i,N}, \tilde w).
    \]
	Thereby, using the distribution
	$\mu_{u^-}^N(0,\cdot,\cdot) = \frac{1}{N}\sum_i 
	\mathbbm 1_{\{V_{u^-}^{i,N} = 0\}}
	\delta_{(S_{u^-}^{i,N},\xi_{u^-}^{i,N})}$, we obtain that
	\begin{align*}
	\circled{2.a} &= \sum_{\tilde w \in \relatifs} \frac{1}{N} \sum_i 
	\mathbbm 1_{\{V_{u^-}^{i,N} = 0\}}
	\alpha(I_{u^-}^{i,N}) \gamma_{\xi_{u^-}^{j,N}}(S_{u^-}^{i,N},\tilde w)
	\Big( \delta_{(1,0,\tilde w)}
	- \delta_{(0,S_{u^-}^{i,N},\tilde w)} \Big)
	\\
	& =
	\sum_{\tilde w \in \relatifs} \int_{\reels^+ \times \mathcal P(E_m)}
	\alpha(I(\tilde \xi)) \gamma_{\xi_{u^-}^{j,N}}(s,\tilde w)
	\Big( \delta_{(1,0,\tilde w)} - \delta_{(0,s,\tilde w)} \Big)
	\mu_{u^-}^N(0,ds,d\tilde \xi)
	.
	\end{align*}
	This last equation incites us to define the two following distributions.
	For $\xi \in \mathcal{P}(E_m)$ and
	$\mu \in \mathcal{P}(E)$, we define $\nu^0 (\xi,\mu) $ and
	$\nu^1 (\xi,\mu)$ the distributions on $\reels^+ \times \relatifs$ such that
	for all $( A,w) \in \mathcal B(\reels^+) \times \relatifs$,
	\begin{align}\label{eq:defnutilde0}
	\nu^0 (\xi,\mu) ( A,w) & \coloneqq 
	\int_{ \mathcal{P}(E_m) } \int_A \alpha\big(I(\tilde \xi)\big)
	\gamma_{\xi}(s, w)\mu(0,ds,d \tilde \xi)
	,\\
	\nu^1 (\xi,\mu) ( A,w)
	& \coloneqq  
	\mathbbm 1_{ \{0 \in A\} } 
	\nu^0 (\xi,\mu) (\reels^+,w )
	.
	\label{eq:defnutilde1}
	\end{align}
	Hence, we deduce that
	\[
	\circled{2.a} = \delta_1 \otimes \nu^1 ( {\xi_{u^-}^{j,N}},{\mu_{u^-}^N}) 
	- \delta_0 \otimes \nu^0 ( {\xi_{u^-}^{j,N}},{\mu_{u^-}^N})
	\]
	and then
	\begin{align*}
	&\circled{2} \overset{\esp}{=} \int_0^t \frac{1}{N} \sum_j
	\partial_{\xi} \Psi ( V_{u^-}^{j,N}, S_{u^-}^{j,N}, \xi_{u^-}^{j,N} )
	\cdot\  \circled{2.a} du
	+ {\scriptstyle\mathcal{O}}_{N \infty} ( 1 )
	\\ & \overset{\esp}{=}
	\int_0^t \int_{E}
	\partial_{\xi} \Psi ( v,s,\xi )
	\cdot 
	\Big(\delta_1 \otimes \nu^1 ( {\xi},{\mu_{u^-}^N}) 
	- \delta_0 \otimes \nu^0 ( {\xi},{\mu_{u^-}^N})\Big)\mu_{u^-}^N(dv,ds,d\xi) du
	+ {\scriptstyle\mathcal{O}}_{N \infty} ( 1 ).
	\end{align*}
\end{proof}
From this result and under Assumption~\ref{ass:lim-mu-star-1}, we expect $\esp$\circled{2} to converge -- as $N$ tends to infinity -- to
\begin{equation}\label{eq:limite terme 2}
\int_0^t \int_E  
\partial_{\xi} \Psi ( v,s,\xi ) \cdot 
(\delta_1 \otimes \nu^1 ( \xi, {\mu_{u^-}^*} ) - 
\delta_0 \otimes \nu^0 ( \xi, {\mu_{u^-}^*} )) 
\ \mu_{u^-}^*(dv,ds,d\xi) du.
\end{equation}
This convergence does not hold a priori because the functions
\begin{equation}\label{eq:function-convergence-not-proved-2}
( v,s,\xi ) \mapsto \partial_{\xi} \Psi ( v,s,\xi ) \cdot 
[\delta_1 \otimes \nu^{1} (\xi, {\mu_{u^-}^N} ) - 
\delta_0 \otimes \nu^0 ( \xi, {\mu_{u^-}^N} )]
\end{equation}
are not a priori continuous. We expect to show this convergence using a sequence (indexed by
$N$) of continuous functions both getting closer and closer to \eqref{eq:function-convergence-not-proved-2} and converging to
\[
( v,s,\xi ) \mapsto \partial_{\xi} \Psi ( v,s,\xi ) \cdot 
[\delta_1 \otimes \nu^{1} (\xi, {\mu_{u^-}^*} ) - 
\delta_0 \otimes \nu^0 ( \xi, {\mu_{u^-}^*} )].
\]

\subsubsection*{Equation on \texorpdfstring{$\mu_t^*$}{TEXT}}
We denote by $\mu_t^* = \mathcal{L}(X_t^*) = \mathcal{L}(V_t^*,S_t^*,\xi_t^*)$.
Thus, under Assumption~\ref{ass:lim-mu-star-1} and from Propositions~\ref{prop:limite-terme-1}-
\ref{prop:limite-terme-2}, taking the limit as $N$ tends to infinity in equation~\eqref{eq:mu-t-N} we can formulate the following conjecture.
\begin{conjecture}\label{th:equation limite mu sans plasticite}
	Assume that Assumptions~\ref{hyp:cdt-init},~\ref{ass:lim-mu-star-1} and
	\ref{ass:lim-mu-star-2} hold.
	Then, the dynamics of $\mu_t^*$ is given by:
	\begin{itemize}
		\item \( \mathcal L(X_0^*) = \lim_{N \rightarrow \infty} \mathcal L(X_0^{1,N}) = \mu_0^* \)
		and in particular $S_0^*$ follows the distribution $\rho_0$
		which is absolutely continuous with respect to the Lebesgue measure,
		\item for all $t \geq 0$, for all $\Psi \in C_b^{1,1}(E)$
		(continuously differentiable with respect to its second variable and
		Fréchet differentiable with respect to its third variable), and
		using Notation~\ref{nota:produitmesure},
		\begin{align}\label{eq:mu star sans plasticite}
		\begin{split}
		\langle \mu_t^*, \Psi\rangle
		& =
		\langle \mu_0^*, \Psi\rangle 
		+ \int_0^t \int_E 
		\Big(
		\partial_s \Psi(v,s, \xi) + \partial_{\xi} \Psi(v,s, \xi) \cdot (-\partial_s \xi) 
		\Big)
		\mu_u^*(dv,ds,d\xi) du
		\\
		& \quad
		+ \int_0^t \int_{\reels^+ \times \mathcal P(E_m)} \beta 
		\Big[ \Psi \Big( 0, s, \xi \Big) 
		- \Psi \Big( 1, s, \xi \Big) \Big]
		{\mu_{u^-}^*}(1,ds,d\xi) du
		\\
		& \quad
		+ \int_0^t \int_E 
		\partial_{\xi} \Psi ( v,s,\xi ) \cdot 
		( \beta \delta_0 \otimes \xi^1 - \beta \delta_1 \otimes \xi^1 ) 
		\mu_{u^-}^*(dv,ds,d\xi) du
		\\
		& \quad
		+ \int_0^t \int_{\reels^+ \times \mathcal P(E_m)} \alpha\big(I(\xi)\big)
		\Big[ \Psi \Big( 1, 0,\xi\Big) 
		- \Psi \Big( 0,s,\xi \Big) \Big]
		{\mu_{u^-}^*}(0,ds,d\xi) du
		\\
		& \quad
		+ \int_0^t \int_E 
		\partial_{\xi} \Psi ( v,s,\xi ) \cdot 
		(\delta_1 \otimes \nu^1 ( \xi, {\mu_{u^-}^*} ) - \delta_0 \otimes \nu^0 ( \xi, {\mu_{u^-}^*} ))
		\ \mu_{u^-}^*(dv,ds,d\xi) du
		.
		\end{split}
		\end{align}
	\end{itemize}
\end{conjecture}
We deduce from this conjecture that the joint distribution of the two first components of
$\xi_t^*$ and $\mu_t^*$ would then be equal for all $t$.
\begin{consequence}\label{cor:egalite-s-mu-xi-star}
	Assume that Conjecture~\ref{th:equation limite mu sans plasticite} holds.
	Then, for all $t \geq 0$,
	\[
	\xi_t^*(\cdot,\cdot,\relatifs) = \mu_t^*(\cdot,\cdot,\mathcal P(E_m)).
	\]
\end{consequence}
\begin{proof}
	For all $\xi_0^* \in \supp(\mu_0^*)$,
	we have $\mu_0^*(\cdot,\cdot,\mathcal P(E_m)) = \xi_0^*(\cdot,\cdot,\relatifs)$.
	\textcolor{black}{Then, we show that $(\mu_t^*(\cdot,\cdot,\mathcal
	P(E_m)))_{t \geq 0}$ and $(\xi_t^*(\cdot,\cdot,\relatifs))_{t \geq 0}$ have the same dynamics, which enables us to conclude the proof.
	First, in equation \eqref{eq:mu star sans plasticite}, taking a function $\Psi$ which depends only on $v$ and $s$, we obtain the dynamics of
	$(\mu_t^*(\cdot,\cdot,\mathcal
	P(E_m)))_{t \geq 0}$:}
	\begin{align}\label{eq:dynamique de mu P E_m}
	\begin{split}
		\langle \mu_t^*, \Psi\rangle
		= &
		\langle \mu_0^*, \Psi\rangle 
		+ \int_0^t \int_E 
		\partial_s \Psi(v,s)
		\mu_u^*(dv,ds,d\xi) du
		\\
		&
		+ \int_0^t \int \beta 
		\left[ \Psi ( 0, s) 
		- \Psi ( 1, s ) \right]
		{\mu_{u^-}^*}(1,ds,d\xi) du
		\\
		&
		+ \int_0^t \int \alpha\big(I(\xi)\big)
		\left[ \Psi ( 1, 0) 
		- \Psi ( 0,s ) \right]
		{\mu_{u^-}^*}(0,ds,d\xi) du
		.
	\end{split}
	\end{align}
	\textcolor{black}{
	Second, taking $\Psi$ depending only on $\xi$ in equation~\eqref{eq:mu star sans plasticite}, we obtain that
	\begin{align*}
	\begin{split}
	\langle \mu_t^*, \Psi\rangle -
	\langle \mu_0^*, \Psi\rangle
    =
	\int_0^t \int_E 
	\partial_{\xi} \Psi ( \xi ) \cdot 
	\Big[&\beta \delta_0 \otimes \xi^1 - \beta \delta_1 \otimes \xi^1
	\\
	& 
	+\delta_1 \otimes \nu^1 ( \xi, {\mu_{u^-}^*} ) - \delta_0 \otimes \nu^0 ( \xi, {\mu_{u^-}^*} )
	-\partial_s \xi\Big]
	\ \mu_{u^-}^*(dv,ds,d\xi) du
	.
	\end{split}
	\end{align*}
	Working on the right hand side term, we obtain
	\[
		\langle \mu_t^*, \Psi\rangle -
		\langle \mu_0^*, \Psi\rangle 
		=
		\esp\left[\Psi(\xi_t^*)\right] - \esp\left[\Psi(\xi_0^*)\right]
		=
		\int_0^t \esp\left[\partial_{\xi} \Psi (\xi_u^*) \cdot \partial\xi_u^*\right] du
	\]
	and then deduce that \(\phi \in C_b^1(E_m)\) with bounded derivative,
	\begin{align}\label{eq:xi-avec-nu-tilde}
	\begin{split}
	\langle \xi_{t}^*, \phi \rangle &= \langle \xi_{0}^*, \phi \rangle +
	\int_0^t \langle -\partial_s \xi_{u}^*, \phi \rangle du
	+ \int_0^t \beta \langle \delta_0 \otimes {\xi_{u^-}^*}^1 - \delta_1 \otimes {\xi_{u^-}^*}^1, \phi \rangle du
	\\
	& \quad + \int_0^t \langle \delta_1 \otimes \nu^1 ({\xi_{u^-}^*},{\mu_{u^-}^*}) 
	- \delta_0 \otimes \nu^0 ({\xi_{u^-}^*},{\mu_{u^-}^*}), \phi \rangle du
	.
	\end{split}
	\end{align}
	Evaluating this last equation for a function $\phi$ depending only on $v$ and $s$ gives us the dynamics 
	of $\xi_{t}^*(\cdot,\cdot,\relatifs)$.
	We obtain exactly the same equation as \eqref{eq:dynamique de mu P E_m}:
	\begin{itemize}
		\item using \eqref{eq:def-derivee-proba}, we have for all $u \in [0,t]$,
		$
			\langle -\partial_s \xi_{u}^*, \phi \rangle
			= \langle \xi_{u}^*, \partial_s \phi \rangle,
		$
		\item as for all $s$ in the support of $\xi \in \mathcal P(E_m)$, $\sum_{w\in\relatifs}\gamma_\xi(s,w) = 1$, we obtain that (see \eqref{eq:defnutilde0} and \eqref{eq:defnutilde1} for the definitions of $\nu^0$ and $\nu^1$) for all $A \in \mathcal B(\reels^+)$:
		\begin{align*}
		\nu^0 (\xi,\mu) ( A,\relatifs) & =
		\int_{ \mathcal{P}(E_m) } \int_A \alpha\big(I(\tilde \xi)\big)
		\mu(0,ds,d \tilde \xi)
		,
		\\
		\nu^1 (\xi,\mu) ( A,\relatifs)
		& =
		\mathbbm 1_{ \{0 \in A\} } 
		\int_{ \mathcal{P}(E_m) } \int_{\reels^+} \alpha\big(I(\tilde \xi)\big)
		\mu(0,ds,d \tilde \xi)
		.
		\end{align*}
	\end{itemize}
}
\end{proof}

\subsubsection*{The typical neuron dynamics}\label{sec:dyn-particule-sans-plast} 
\begin{consequence}\label{cor:equation limite systeme de particules sans plasticite}
\textcolor{black}{Assume that Conjecture~\ref{th:equation limite mu sans plasticite}
		and Assumption~\ref{ass:compabilite EDP} hold with \(p^\pm \equiv 0\).}
	Then, we have
	\begin{itemize}
	    \item $d S_t^* = dt$.
		\item $\xi_t^*$ admits a density in $s$ that
		satisfies the following equations,
		\begin{align*}
		\hspace{-3em}
		&\left\{
		\begin{array}{ll}
		\partial_t{\xi_t^*}(0,s,w) &= -\partial_s {\xi_t^*}(0,s,w) + \beta{\xi_t^*}(1,s,w)
		- \int_{ \mathcal{P}(E_m)} \alpha\big(I(\xi')\big)
		\frac{{\xi_t^*}(0,s,w)}{{\xi_t^*}(0,s,\relatifs)} {\mu_t^*}(0,s,d\xi')
		\\
		\\
		{\xi_t^*}(0,0,w) &= 0,
		\end{array}
		\right.
		\\
		\hspace{-3em}
		\\
		\hspace{-3em}
		&\left\{
		\begin{array}{ll}
		\partial_t{\xi_t^*}(1,s,w) &= -\partial_s {\xi_t^*}(1,s,w)
		- \beta {\xi_t^*}(1,s,w)
		\\
		\\
		{\xi_t^*}(1,0,w) &= \int_{\reels^+ \times \mathcal{P}(E_m)}
		\alpha\big(I(\xi')\big) 
		\frac{{\xi_t^*}(0,s',w)}{{\xi_t^*}(0,s',\relatifs)}
		{\mu_t^*}(0,s',d\xi')ds'.
		\end{array}
		\right.
		\end{align*}
		\item At rate $\beta \mathbbm 1_{ \{V_{t^-}^* = 1\} }$, 
		$(V_{t^-}^*,S_{t^-}^*,\xi_{t^-}^*)$ jumps to $(0,S_{t^-}^*,\xi_{t^-}^*)$.
		\item At rate $\alpha\big(I(\xi_{t^-}^*)\big) \mathbbm 1_{ \{V_{t^-}^* = 0\} }$, 
		$(V_{t^-}^*,S_{t^-}^*,\xi_{t^-}^*)$ jumps to $(1,0,\xi_{t^-}^*)$.
	\end{itemize}
\end{consequence}
\begin{proof}
	Both the first two and last two points are clear with Conjecture~\ref{th:equation limite mu 
	sans plasticite}. 
	We detail the third point. 
    \textcolor{black}{First, from the regularity assumption on the
	density in $s$ of $\mu_t^*$, Consequence~\ref{cor:egalite-s-mu-xi-star} tells us that $\xi_t^*$ admits a density $C^1$ in $s$
	and satisfies \eqref{eq:xi-avec-nu-tilde}.
	Using \eqref{eq:gamma-xi}, \eqref{eq:defnutilde0}, \eqref{eq:defnutilde1},
	we thus obtain by integration by parts the following equations on the density functions
	$s \mapsto {\xi_t^*}(v,s, w)$:}
	\begin{align*}
	\partial_t{\xi_t^*}(0,s,w) =& -\partial_s {\xi_t^*}(0,s,w) - \delta_0(s) {\xi_t^*}(0,0,w)
	+ \beta{\xi_t^*}(1,s,w)
	\\
	&
	- \int_{ \mathcal{P}(E_m)}
	\alpha\big(I(\xi')\big)
	\frac{{\xi_t^*}(0,s,w)}{{\xi_t^*}(0,s,\relatifs)}
	{\mu_t^*}(0,s,d\xi') ,
	\\
	\\
	\partial_t{\xi_t^*}(1,s,w) =& -\partial_s {\xi_t^*}(1,s,w) - \delta_0(s) {\xi_t^*}(1,0,w)
	- \beta {\xi_t^*}(1,s,w)
	\\
	&
	+ \delta_0(s) \int_{\reels^+ \times \mathcal{P}(E_m)}
	\alpha\big(I(\xi')\big) 
	\frac{{\xi_t^*}(0,s',w)}{{\xi_t^*}(0,s',\relatifs)}
	{\mu_t^*}(0,ds',d\xi')
	.
	\end{align*}
By integrating the previous equations on $s\in[0,\epsilon]$
and taking $\epsilon$ to $0$ we obtain the equations
with conditions at boundaries $s=0$ of the corollary.
\end{proof}

\subsection{The case with plasticity}\label{supp:The case with plasticity}
	We now deal with the synaptic weight jumps.
	As soon as a neuron spikes, several different synaptic weight jumps are
	possible. We denote by $J_i^N$ the possible increments of the weights when
	the neuron $i$ spikes
	\[	
	J_i^N = \{A=(a^{kl})_{1 \leq k,l \leq N}: a^{ii} \in \{-1,0,1\}
	\text{ and }\forall k,l \neq i, a^{il} \in \{0,1\}, a^{ki} \in \{0,-1\},
	a^{kl}=0 \}.
	\]
	For all $\Delta \in J_i^N$, we denote by $P_{t,i}^{\Delta}$ the
	probability that the weight matrix be incremented by $\Delta$ conditionally
	to $X_t^N$. In order to ease its understanding,
	we express it with the weight matrix $W_t^N$ rather than 
	using the empirical distributions $(\xi_{t}^{i,N})_{1 \leq i \leq N}$:
	\begin{align*}
	P_{t,i}^{\Delta} &= 
	\Bigg[
	\frac{\Delta^{ii}(1+\Delta^{ii})}{2}p^+(S_{t^-}^{i,N},W_{t^-}^{ii,N})\big(1 - p^-(S_{t^-}^{i,N},W_{t^-}^{ii,N})\big)
	\\
	&
	\qquad+\big(1-\lvert \Delta^{ii}\rvert\big)
	\Big(p^+(S_{t^-}^{i,N},W_{t^-}^{ii,N})p^-(S_{t^-}^{i,N},W_{t^-}^{ii,N})
	\\
	& \qquad\qquad\qquad\qquad\qquad
	+ \big(1 - p^+(S_{t^-}^{i,N},W_{t^-}^{ii,N})\big)
	\big(1 - p^-(S_{t^-}^{i,N},W_{t^-}^{ii,N})\big)\Big)
	\\
	& 
	\qquad+\frac{\Delta^{ii}(\Delta^{ii}-1)}{2}p^-(S_{t^-}^{i,N},W_{t^-}^{ii,N})\big(1 - p^+(S_{t^-}^{i,N},W_{t^-}^{ii,N})\big)
	\Bigg]
	\\
	&\quad \prod_{k,l \neq i} \Big(
	\Delta^{il} p^+(S_{t^-}^{l,N},W_{t^-}^{il,N}) + 
	(1-\Delta^{il})\big(1 - p^+(S_{t^-}^{l,N},W_{t^-}^{il,N})\big)
	\Big)
	\\
	& \qquad\quad
	\Big(
	-\Delta^{ki} p^-(S_{t^-}^{k,N},W_{t^-}^{ki,N}) + 
	(1+\Delta^{ki})\big(1 - p^-(S_{t^-}^{k,N},W_{t^-}^{ki,N})\big)
	\Big).
	\end{align*}
    
	After the spike of the neuron $i$ at time $t$ and assuming the weights are
	incremented of $\Delta \in J_i^N$ at this time, $\xi_{t^-}^{i,N}$ jumps to
	$\xi_{t,\Delta}^{i,N}$ and for all $j \neq i$, $\xi_{t^-}^{j,N}$ jumps to
	$\xi_{t,\Delta,i}^{j,N}$ such that
	\begin{align}\label{eq:def-xi-delta-i}
	\begin{split}
	\xi_{t,\Delta}^{i,N} & = \xi_{t^-}^{i,N} + 
	\frac{1}{N} \sum_{(0, S_{t^-}^{i,N}, \tilde{W}) \in
		\supp(\xi_{t^-}^{i,N})} 
	\big( \delta_{(1,0,\tilde{W} + \Delta^{ii})} 
	- \delta_{(0,S_{t^-}^{i,N},\tilde{W})}\big)
	\\
	& \quad
	+ \sum_{l \neq i}\frac{1}{N}
	\sum_{(\tilde{V}, S_{t^-}^{l,N}, \tilde{W}) \in
		\supp(\xi_{t^-}^{i,N})} 
	\big( \delta_{(\tilde{V},S_{t^-}^{l,N},\tilde{W} + \Delta^{il})}
	- \delta_{(\tilde{V},S_{t^-}^{l,N},\tilde{W})}\big)
	\end{split}
	\end{align}
	\begin{align}\label{eq:def-xi-delta-i-j}
	\xi_{t,\Delta,i}^{j,N} & = \xi_{t^-}^{j,N}
	+\frac{1}{N}
	\sum_{(0, S_{t^-}^{i,N}, \tilde{W}) \in
		\supp(\xi_{t^-}^{j,N})} 
	\big( \delta_{(1,0,\tilde{W} + \Delta^{ji})}
	- \delta_{(0,S_{t^-}^{i,N},\tilde{W})}\big).
	\end{align}
	For all $\Psi \in C_b^{1,1}(E)$, we now have
	\begin{align*}
	\begin{split}
	\langle \mu_t^N, \Psi\rangle
	&\overset{\esp}{=}
	\langle \mu_0^N, \Psi\rangle 
	+ \int_0^t \int_E 
	\Big(
	\partial_s \Psi(v,s, \xi) + \partial_{\xi} \Psi(v,s, \xi) \cdot (-\partial_s \xi)  
	\Big)
	\mu_u^N(dv,ds,d\xi) du
	\\
	&
	+ \int_0^t \sum_i \mathbbm 1_{\{V_{u^-}^{i,N} = 1\}}  \frac{\beta}{N} \Biggl\{
	\Bigl[ \Psi \Big( 0,S_{u^-}^{i,N},\xi_{u^-}^{i,N} 
	+ \eta^N ( S_{u^-}^{i,N}, \xi_{u^-}^{i,N} ) \Big) 
	- \Psi \Big( 1,S_{u^-}^{i,N},\xi_{u^-}^{i,N} \Big) \Big]
	\\
	& \qquad
	+ \sum_{j \neq i} 
	\Big[ \Psi \Big( V_{u^-}^{j,N}, S_{u^-}^{j,N}, \xi_{u^-}^{j,N} 
	+ \eta^N ( S_{u^-}^{i,N}, \xi_{u^-}^{j,N} ) \Big) 
	- \Psi \Big( V_{u^-}^{j,N}, S_{u^-}^{j,N}, \xi_{u^-}^{j,N} \Big) \Big]
	\Biggr\} du
	\\
	&
	+ \int_0^t \sum_i \mathbbm 1_{\{V_{u^-}^{i,N} = 0\}} 
	\frac{\alpha(I_{u^-}^{i,N})}{N} \sum_{\Delta \in J_i^N} P_{u,i}^{\Delta}
	\Biggl\{
	\Big[ 
	\Psi \Big( 1, 0, \xi_{u,\Delta}^{i,N} \Big) 
	- \Psi \Big( 0,S_{u^-}^{i,N},\xi_{u^-}^{i,N} \Big) \Big]
	\\
	& \qquad \qquad \qquad \qquad 
	+ \sum_{j \neq i} \Big[ 
	\Psi \Big( V_{u^-}^{j,N},S_{u^-}^{j,N}, \xi_{u,\Delta,i}^{j,N} \Big) 
	- \Psi \Big( V_{u^-}^{j,N},S_{u^-}^{j,N},\xi_{u^-}^{j,N} \Big) \Big]
	\Biggr\} du.
	\end{split}
	\end{align*}
	By adding and removing
	\[
	\esp\Bigg[\int_0^t \sum_i \mathbbm 1_{\{V_{u^-}^{i,N} = 0\}} 
	\frac{\alpha(I_{u^-}^{i,N})}{N} \sum_{\Delta \in J_i^N} P_{u,i}^{\Delta}
	\Big[\Psi \Big( 0,S_{u^-}^{i,N}, \xi_{u,\Delta,i}^{i,N} \Big) 
	- \Psi \Big( 0,S_{u^-}^{i,N},\xi_{u^-}^{i,N} \Big)\Big]du\Bigg],
	\]
	we can rewrite the last term to obtain a complete sum over $j$:
	\begin{align}\label{eq:mu-avec-plasticite}
	\begin{split}
	\langle \mu_t^N, \Psi\rangle
	&\overset{\esp}{=}
	\langle \mu_0^N, \Psi\rangle 
	+ \int_0^t \int_E 
	\Big(
	\partial_s \Psi(v,s, \xi) + \partial_{\xi} \Psi(v,s, \xi) \cdot (-\partial_s \xi)  
	\Big)
	\mu_u^N(dv,ds,d\xi) du
	\\
	&
	+ \int_0^t \sum_i \mathbbm 1_{\{V_{u^-}^{i,N} = 1\}}  \frac{\beta}{N} \Biggl\{
	\Bigl[ \Psi \Big( 0,S_{u^-}^{i,N},\xi_{u^-}^{i,N} 
	+ \eta^N ( S_{u^-}^{i,N}, \xi_{u^-}^{i,N} ) \Big) 
	- \Psi \Big( 1,S_{u^-}^{i,N},\xi_{u^-}^{i,N} \Big) \Big]
	\\
	& \qquad
	+ \sum_{j \neq i} 
	\Big[ \Psi \Big( V_{u^-}^{j,N}, S_{u^-}^{j,N}, \xi_{u^-}^{j,N} 
	+ \eta^N ( S_{u^-}^{i,N}, \xi_{u^-}^{j,N} ) \Big) 
	- \Psi \Big( V_{u^-}^{j,N}, S_{u^-}^{j,N}, \xi_{u^-}^{j,N} \Big) \Big]
	\Biggr\} du
	\\
	&
	+ \int_0^t \sum_i \mathbbm 1_{\{V_{u^-}^{i,N} = 0\}} 
	\frac{\alpha(I_{u^-}^{i,N})}{N} \sum_{\Delta \in J_i^N} P_{u,i}^{\Delta}
	\Biggl\{
	\Big[ 
	\Psi \Big( 1, 0, \xi_{u,\Delta}^{i,N} \Big) 
	- \Psi \Big( 0,S_{u^-}^{i,N},\xi_{u,\Delta,i}^{i,N} \Big) \Big]
	\\
	& \qquad \qquad \qquad \qquad \qquad \qquad
	+ \sum_{j} \Big[ 
	\Psi \Big( V_{u^-}^{j,N},S_{u^-}^{j,N}, \xi_{u,\Delta,i}^{j,N} \Big) 
	- \Psi \Big( V_{u^-}^{j,N},S_{u^-}^{j,N},\xi_{u^-}^{j,N} \Big) \Big]
	\Biggr\} du.
	\end{split}
	\end{align}
	We then denote by
	\begin{equation}\label{eq:def-term-3-potentiation}
		\circled{3} = \sum_i \mathbbm 1_{\{V_{u^-}^{i,N} = 0\}} 
		\frac{\alpha(I_{u^-}^{i,N})}{N} 
		\sum_{\Delta \in J_i^N} P_{u,i}^{\Delta}\Big[ 
		\Psi \Big( 1, 0, \xi_{u,\Delta}^{i,N} \Big) 
		- \Psi \Big(  0, S_{u^-}^{i,N}, \xi_{u,\Delta,i}^{i,N} \Big) \Big]
	\end{equation}
	and
	\[
	\circled{4} = \sum_i \mathbbm 1_{\{V_{u^-}^{i,N} = 0\}} 
	\frac{\alpha(I_{u^-}^{i,N})}{N} \sum_{\Delta \in J_i^N} P_{u,i}^{\Delta}
	\sum_j\Big[ 
	\Psi \Big( V_{u^-}^{j,N},S_{u^-}^{j,N}, \xi_{u,\Delta,i}^{j,N} \Big) 
	- \Psi \Big( V_{u^-}^{j,N},S_{u^-}^{j,N},\xi_{u^-}^{j,N} \Big) \Big].
	\]
	\subsubsection*{The depression term}
	We first deal with the term \circled{4} which describes the depression of
	the weights.
	\begin{proposition}\label{prop:limite-terme-4}
		Then, for all $u \geq 0$,
		\textcolor{black}{\begin{align*}
		\circled{4}
		\overset{\esp}{=} &
		\int  
		\partial_{\xi} \Psi ( v,s,\xi ) \cdot 
		[\delta_1 \otimes \nu^{-,1} (s, \xi, {\mu_{u}^N} ) - 
		\delta_0 \otimes \nu^0 ( \xi, {\mu_{u}^N} )]
		\ \mu_{u}^N(dv,ds,d\xi) + \smallO{1}.
		\end{align*}}
		where $\nu^{-,1}$ is defined in the proof, see \eqref{eq:def nu_moins_1}.
	\end{proposition}
	\begin{proof}
		First, by Lemma~\ref{lem:supp-s-distinct}, we have almost surely for all $i,j$,
		\[
			\lVert\frac{1}{N}
			\sum_{(0, S_{u^-}^{i,N}, \tilde{W}) \in
				\supp(\xi_{u^-}^{j,N})} 
			\big( \delta_{(1,0,\tilde{W} + \Delta^{ji})}
			- \delta_{(0,S_{u^-}^{i,N},\tilde{W})}\big)\rVert_{TV}
			\leq \frac{2}{N}.
		\]
		Then, we use Fréchet derivative to obtain that
		\begin{align*}
		\Psi \Big( V_{u^-}^{j,N},S_{u^-}^{j,N},&\ \xi_{u,\Delta,i}^{j,N} \Big) 
        - \Psi \Big( V_{u^-}^{j,N},S_{u^-}^{j,N},\xi_{u^-}^{j,N} \Big) =
		\\&
		\partial_{\xi}
		\Psi \Big( V_{u^-}^{j,N},S_{u^-}^{j,N},\xi_{u^-}^{j,N} \Big)
		\cdot
		\Big(\frac{1}{N}
		\sum_{(0, S_{u^-}^{i,N}, \tilde{W}) \in
			\supp(\xi_{u^-}^{j,N})} 
		\big( \delta_{(1,0,\tilde{W} + \Delta^{ji})}
		- \delta_{(0,S_{u^-}^{i,N},\tilde{W})}\big)\Big)
		+ \smallO{\frac{1}{N}}.
		\end{align*}
		We deduce by linearity that
		\begin{align*}
		\circled{4}= \frac{1}{N}
		\sum_j \partial_{\xi}
		\Psi &\Big( V_{u^-}^{j,N},S_{u^-}^{j,N},\xi_{u^-}^{j,N} \Big)
		\cdot
		\\&
		\Big(\underbrace{\sum_i \mathbbm 1_{\{V_{u^-}^{i,N} = 0\}}
			\frac{\alpha(I_{u^-}^{i,N})}{N}
			\sum_{(0, S_{u^-}^{i,N}, \tilde{W}) \in
				\supp(\xi_{u^-}^{j,N})} 
			\sum_{\Delta \in J_i^N} P_{u,i}^{\Delta}
			\big( \delta_{(1,0,\tilde{W} + \Delta^{ji})}
			- \delta_{(0,S_{u^-}^{i,N},\tilde{W})}\big)}_{\circled{4a}}\Big)
		+ \smallO{1}.
		\end{align*}
		But for any function $f$ on $\{1,0,-1\}$ we have for all $i\neq j$,
		\begin{align*}
		\sum_{\Delta \in J_i^N} P_{u,i}^{\Delta}f(\Delta^{ji})
		&= \sum_{\Delta \in J_i^N,\Delta^{ji}=-1} P_{u,i}^{\Delta} f(-1) 
		+ \sum_{\Delta \in J_i^N,\Delta^{ji}=0} P_{u,i}^{\Delta} f(0)
		\\
		&= p^-(S_{u^-}^{j,N}, W_{u^-}^{ji})f(-1)
		+ (1-p^-(S_{u^-}^{j,N}, W_{u^-}^{ji}))f(0),
		\end{align*}
		and for $i=j$,
		\begin{align*}
		\sum_{\Delta \in J_j^N} P_{u,j}^{\Delta}f(\Delta^{jj}) &=
        p^-(S_{u^-}^{j,N}, W_{u^-}^{jj})
		(1-p^+(S_{u^-}^{j,N}, W_{u^-}^{jj}))f(-1)
		+
		p^+(S_{u^-}^{j,N}, W_{u^-}^{jj})(1-p^-(S_{u^-}^{j,N}, W_{u^-}^{jj}))f(1)
		\\
		& \quad + \Big(
		(1-p^-(S_{u^-}^{j,N}, W_{u^-}^{jj}))
		(1-p^+(S_{u^-}^{j,N}, W_{u^-}^{jj}))
		+
		p^+(S_{u^-}^{j,N}, W_{u^-}^{jj})
		p^-(S_{u^-}^{j,N}, W_{u^-}^{jj})\Big)f(0).
		\end{align*}
		Denoting by 
        \begin{align*}
            C^{j,N} &= \mathbbm 1_{\{V_{u^-}^{j,N} = 0\}}
		\frac{\alpha(I_{u^-}^{j,N})}{N}
		\sum_{(0, S_{u^-}^{j,N}, \tilde{W}) \in
			\supp(\xi_{u^-}^{j,N})}
		\Big[
		-p^-(S_{u^-}^{j,N},\tilde{W})p^+(S_{u^-}^{j,N},\tilde{W})
		\big( \delta_{(1,0,\tilde{W} - 1)}
		- \delta_{(0,S_{u^-}^{j,N},\tilde{W})}\big)
        \\
		& \qquad\qquad\qquad\qquad\qquad\qquad\qquad\quad\qquad\qquad
		+ p^+(S_{u^-}^{j,N},\tilde{W})\big(2p^-(S_{u^-}^{j,N},\tilde{W})-1\big)
		\big( \delta_{(1,0,\tilde{W})}
		- \delta_{(0,S_{u^-}^{j,N},\tilde{W})}\big)
        \\
		& \qquad\qquad\qquad\qquad\qquad\qquad\qquad\quad\qquad\qquad
		+ p^+(S_{u^-}^{j,N},\tilde{W})\big(1 - p^-(S_{u^-}^{j,N},\tilde{W})\big)
		\big( \delta_{(1,0,\tilde{W}+1)}
		- \delta_{(0,S_{u^-}^{j,N},\tilde{W})}\big)\Big]
        ,
        \end{align*}
        a term of order $\frac{1}{N}$, we obtain that
		\begin{align*}
		\circled{4a} - C^{j,N}
		&=
		\sum_i \mathbbm 1_{\{V_{u^-}^{i,N} = 0\}}
		\frac{\alpha(I_{u^-}^{i,N})}{N}
		\sum_{(0, S_{u^-}^{i,N}, \tilde{W}) \in
			\supp(\xi_{u^-}^{j,N})} 
		\Big[
		p^-(S_{u^-}^{j,N},\tilde{W})
		\big( \delta_{(1,0,\tilde{W} - 1)}
		- \delta_{(0,S_{u^-}^{i,N},\tilde{W})}\big)
		\\
		& \qquad\qquad\qquad\qquad\qquad\qquad\qquad\quad\qquad\qquad\quad
		+ \big(1 - p^-(S_{u^-}^{j,N},\tilde{W})\big)
		\big( \delta_{(1,0,\tilde{W})}
		- \delta_{(0,S_{u^-}^{i,N},\tilde{W})}\big)\Big]
		\\
		&= 
		\sum_i \mathbbm 1_{\{V_{u^-}^{i,N} = 0\}}
		\frac{\alpha(I_{u^-}^{i,N})}{N}
		\sum_{\tilde W \in \relatifs}
		\mathbbm 1_{ \{ (0, S_{u^-}^{i,N}, \tilde{W}) 
			\in \supp(\xi_{u^-}^{j,N}) \} }
		\Big[
		p^-(S_{u^-}^{j,N},\tilde{W})
		\big( \delta_{(1,0,\tilde{W} - 1)}
		- \delta_{(0,S_{u^-}^{i,N},\tilde{W})}\big)
		\\
		& \qquad\qquad\qquad\qquad\qquad\qquad\qquad\quad\qquad\qquad\qquad\qquad\quad
		+ \big(1 - p^-(S_{u^-}^{j,N},\tilde{W})\big)
		\big( \delta_{(1,0,\tilde{W})}
		- \delta_{(0,S_{u^-}^{i,N},\tilde{W})}\big)\Big]
		\\
		& =
		\sum_i \mathbbm 1_{\{V_{u^-}^{i,N} = 0\}}
		\frac{\alpha(I_{u^-}^{i,N})}{N}
		\sum_{\tilde W \in \relatifs}
		\gamma_{{\xi_{u^-}^{j,N}}}(S_{u^-}^{i,N},\tilde W)
		\Big[
		p^-(S_{u^-}^{j,N},\tilde{W})
		\big( \delta_{(1,0,\tilde{W} - 1)}
		- \delta_{(0,S_{u^-}^{i,N},\tilde{W})}\big)
		\\
		& \qquad\qquad\qquad\qquad\qquad\qquad \qquad \qquad \qquad \qquad
		+ \big(1 - p^-(S_{u^-}^{j,N},\tilde{W})\big)
		\big( \delta_{(1,0,\tilde{W})}
		- \delta_{(0,S_{u^-}^{i,N},\tilde{W})}\big)\Big],
		\end{align*}
        where we used the definition of $\gamma_{\xi}$ given by \eqref{eq:gamma-xi}.
		Hence, we deduce that 
		\begin{align*}
		\circled{4a}
		&= \sum_{\tilde W \in \relatifs}
		\int_{ \mathcal{P}(E_m) }\int_{\reels^+}
		\alpha(I(\tilde \xi))
		\gamma_{{\xi_{u^-}^{j,N}}}(\tilde s,\tilde W)
		\Big[
		p^-(S_{u^-}^{j,N},\tilde{W})
		\big( \delta_{(1,0,\tilde{W} - 1)}
		- \delta_{(0,\tilde s,\tilde{W})}\big)
		\\
		& \qquad\qquad\qquad\qquad\qquad\qquad\qquad\qquad\quad    
		+ \big(1 - p^-(S_{u^-}^{j,N},\tilde{W})\big)
		\big( \delta_{(1,0,\tilde{W})}
		- \delta_{(0,\tilde s,\tilde{W})}\big)\Big]
		{\mu_{u^-}^N}(0,d \tilde s,d \tilde \xi)
		+\mathcal O \Big(\frac{1}{N}\Big)
		.
		\end{align*}
		This incites us to define the following distribution.
		For all triple $(s,\xi,\mu) \in \reels^+ \times \mathcal P(E_m) \times \mathcal P(E)$,
		we denote by $\nu^{-,1}(s,\xi,\mu)$ the distribution
		on $\reels^+ \times \relatifs$ such that for all 
		$(A,w) \in \mathcal B(\reels^+) \times \relatifs$,
		\begin{align}\label{eq:def nu_moins_1}
		\begin{split}
		\nu^{-,1} (s,\xi,\mu)(A,w)
		&\coloneqq  
		\mathbbm 1_{ \{ 0 \in A \} }
		\int_{ \reels^+ \times \mathcal{P}(E_m) } \alpha\big(I(\tilde \xi)\big) p^-(s,w+1)
		\ \gamma_\xi(\tilde s,w+1)\mu(0,d\tilde s,d\tilde \xi)
		\\
		& \qquad 
		+ \mathbbm 1_{ \{ 0 \in A \} }
		\int_{ \reels^+ \times \mathcal{P}(E_m) } \alpha\big(I(\tilde \xi)\big) (1 - p^-(s,w))
		\ \gamma_\xi(\tilde s,w)\mu(0,d \tilde s,d\tilde \xi).
		\end{split}
		\end{align}
        Thereby, using the function $\nu^0$
		defined in \eqref{eq:defnutilde0}, we get
		\begin{align*}
		&\circled{4a}
		= \delta_1 \otimes \nu^{-,1} (S_{u^-}^{j,N},{\xi_{u^-}^{j,N}},{\mu_{u^-}^N})
		-\delta_0 \otimes \nu^{0} ({\xi_{u^-}^{j,N}},{\mu_{u^-}^N})
		+\smallO{1}
		.
		\end{align*}
		We finally obtain that
		\begin{align*}
		\circled{4}
		\overset{\esp}{=} &
		\int  
		\partial_{\xi} \Psi ( v,s,\xi ) \cdot 
		[\delta_1 \otimes \nu^{-,1} (s, \xi, {\mu_{u}^N} ) - 
		\delta_0 \otimes \nu^0 ( \xi, {\mu_{u}^N} )]
		\ \mu_{u}^N(dv,ds,d\xi) + \smallO{1}.
		\end{align*}
		It completes the proof.
	\end{proof}
	\textcolor{black}{Under Assumption~\ref{ass:lim-mu-star-1}, it is reasonable to expect that, as $N$ tends to infinity, $\esp$\circled{4} converges to
	\begin{align*}
	\int  
	\partial_{\xi} \Psi ( v,s,\xi ) \cdot 
	[\delta_1 \otimes \nu^{-,1} (s, \xi, {\mu_{u}^*} ) - 
	\delta_0 \otimes \nu^0 ( {\xi}, {\mu_{u}^*} )]
	\ \mu_{u}^*(dv,ds,d\xi).
	\end{align*}
	As noted in the remark given just after Proposition~\ref{prop:limite-terme-2},
	this convergence does not hold a priori because the functions
	\begin{equation}\label{eq:function-convergence-not-proved-4}
		( v,s,\xi ) \mapsto \partial_{\xi} \Psi ( v,s,\xi ) \cdot 
		[\delta_1 \otimes \nu^{-,1} (s, \xi, {\mu_{u}^N} ) - 
		\delta_0 \otimes \nu^0 ( \xi, {\mu_{u}^N} )]
	\end{equation}
	are not a priori continuous. We expect to show this convergence using a sequence (indexed by
	$N$) of continuous functions both getting closer and closer to \eqref{eq:function-convergence-not-proved-4} and converging to
	\[
		( v,s,\xi ) \mapsto \partial_{\xi} \Psi ( v,s,\xi ) \cdot 
		[\delta_1 \otimes \nu^{-,1} (s, \xi, {\mu_{u}^*} ) - 
		\delta_0 \otimes \nu^0 ( \xi, {\mu_{u}^*} )].
	\]
}
	
	\subsubsection*{The potentiation term}
	Now, for the term \circled{3} describing potentiation, see
	\eqref{eq:def-term-3-potentiation}, we need to assume that $\Psi$ is linear
	in its third variable $\xi$.
	\begin{proposition}\label{prop:limite-terme-3}
		Under Assumptions~\ref{hyp:cdt-init},~\ref{ass:lim-mu-star-1} and
		\ref{ass:lim-mu-star-2}
		we find that for all $u \geq 0$, for all $\Psi \in C_b^{1,1}(E)$ 
		\textbf{linear in its third variable} $\xi$,
		\[
		\lim_{N \rightarrow \infty}	\circled{3} \overset{\esp}{=}
		\int_E \alpha\big(I(\xi)\big)
		\Big[ \Psi \Big( 1, 0, \nu^+ ( \xi ) \Big) 
		- \Psi \Big( 0,s, \xi \Big) \Big]
		{\mu_u^*}(0,ds,d\xi),
		\]
		where $\nu^+$ is defined in the proof, see \eqref{eq:def-nuplus}.
	\end{proposition}
	\begin{proof}
	 By linearity of $\Psi$ in its third variable, we obtain that
	\[
		\circled{3} = \sum_i \mathbbm 1_{\{V_{u^-}^{i,N} = 0\}} 
		\frac{\alpha(I_{u^-}^{i,N})}{N} 
		\sum_{\Delta \in J_i^N} P_{u,i}^{\Delta}\Big[ 
		\Psi \Big( 1, 0, \sum_{\Delta \in J_i^N} P_{u,i}^{\Delta}
		\xi_{u,\Delta}^{i,N} \Big) 
		- \Psi \Big(  0, S_{u^-}^{i,N}, \sum_{\Delta \in J_i^N} P_{u,i}^{\Delta}\xi_{u,\Delta,i}^{i,N} \Big) \Big].
	\]
	and then by definition of $\xi_{u,\Delta}^{i,N}$, 
	see \eqref{eq:def-xi-delta-i}, we have
	\begin{align*}
		&\sum_{\Delta \in J_i^N} P_{u,i}^{\Delta} \xi_{u,\Delta}^{i,N}
		= \xi_{u^-}^{i,N} + 
		\frac{1}{N} \sum_{(0, S_{u^-}^{i,N}, \tilde{W}) \in
			\supp(\xi_{u^-}^{i,N})} 
		\sum_{\Delta \in J_i^N} P_{u,i}^{\Delta}
		\big( \delta_{(1,0,\tilde{W} + \Delta^{ii})} 
		- \delta_{(0,S_{u^-}^{i,N},\tilde{W})}\big)
		\\
		& \qquad\qquad\qquad\qquad
		+ \sum_{l \neq i}\frac{1}{N}
		\sum_{(\tilde{V}, S_{u^-}^{l,N}, \tilde{W}) \in
			\supp(\xi_{u^-}^{i,N})} 
		\sum_{\Delta \in J_i^N} P_{u,i}^{\Delta}
		\big( \delta_{(\tilde{V},S_{u^-}^{l,N},\tilde{W} + \Delta^{il})}
		- \delta_{(\tilde{V},S_{u^-}^{l,N},\tilde{W})}\big).
	\end{align*}
	As for the term \circled{4}, noting the second term of the right hand side of the last equation is
	of order $\frac{1}{N}$, we can find a term $\tilde{C}_i^N$ of $\frac{1}{N}$ such that
	\begin{align*}
		\sum_{\Delta \in J_i^N} P_{u,i}^{\Delta} \xi_{u,\Delta}^{i,N}
		+ \tilde{C}^{i,N}
		&
		= \xi_{u^-}^{i,N} 
		+ \sum_{l \neq i}\frac{1}{N}
		\sum_{(\tilde{V}, S_{u^-}^{l,N}, \tilde{W}) \in
			\supp(\xi_{u^-}^{i,N})}
		p^+(S_{u^-}^{l,N},\tilde W)
		\big( \delta_{(\tilde{V},S_{u^-}^{l,N},\tilde{W} + 1)}
		- \delta_{(\tilde{V},S_{u^-}^{l,N},\tilde{W})}\big)
		\\
		&
		= \frac{1}{N} \sum_{(\tilde{V}, \tilde S, \tilde{W}) \in
			\supp(\xi_{u^-}^{i,N})} 
		p^+(\tilde S,\tilde W)
		\delta_{(\tilde{V},\tilde S,\tilde{W} + 1)}
		+ \big(1-p^+(\tilde S,\tilde W)\big)
		\delta_{(\tilde{V},\tilde S,\tilde{W})}
		\\
		&
		= \nu^+(\xi_{u^-}^{i,N}) 
	\end{align*}
	where
	\begin{equation}\label{eq:def-nuplus}
	\nu^+(\xi) (v,A,w)
	= \int_A p^+(s,w-1) \xi(v,ds,\{w-1\}) 
	+ \int_A (1-p^+(s,w)) \xi(v,ds,w).
	\end{equation}
	\textcolor{black}{But $p^+$ takes values in $[0,1]$ so using the triangular inequality,
		$\xi \mapsto \nu^+(\xi)$ is continuous. We can then conclude using
		Assumption~\ref{ass:lim-mu-star-1}.}
	\end{proof}
 
\subsubsection{Justifying Main Result~\ref{th:equation limite systeme de particules avec plasticite}}
We are now in position to detail why Main Result~\ref{th:equation limite systeme de particules avec plasticite} can be conjectured.
	Under Assumptions~\ref{hyp:cdt-init},~\ref{ass:lim-mu-star-1} and~\ref{ass:lim-mu-star-2},
	by taking $N$ to infinity in equation
	\eqref{eq:mu-avec-plasticite} and using the results of
	Propositions~\ref{prop:limite-terme-4} and~\ref{prop:limite-terme-3}, 
	we should get the following dynamics of $\mu_t^*$:
	\begin{itemize}
		\item \(\mathcal L(X_0^*) = \lim_{N \rightarrow \infty} \mathcal L(X_0^{1,N}) = \mu_0^*\)
		and in particular $S_0^*$ is distributed by the $\rho_0$ distribution
		which is absolutely continuous with respect to the Lebesgue measure,
		\item using Notation~\ref{nota:produitmesure} we have,
		for all $t \geq 0$, for all $\Psi \in C_b^{1,1}(E)$,
		\begin{align}\label{eq:dynamique mu plasticite}
		\hspace{-2em}
		\begin{split}
		\langle \mu_t^*, \Psi\rangle
		= &
		\ \langle \mu_0^*, \Psi\rangle 
		+ \int_0^t \int_E 
		\Big(
		\partial_s \Psi(v,s, \xi) + \partial_{\xi} \Psi(v,s, \xi) \cdot (-\partial_s \xi) 
		\Big)
		\mu_u^*(dv,ds,d\xi) du
		\\
		&
		+ \int_0^t \int_{\reels^+ \times \mathcal P(E_m)} \beta 
		\Big[ \Psi \Big( 0, s, \xi \Big) 
		- \Psi \Big( 1, s, \xi \Big) \Big]
		{\mu_u^*}(1,ds,d\xi) du
		\\
		&
		+ \int_0^t \int_E  
		\partial_{\xi} \Psi ( v,s,\xi ) \cdot (\beta \delta_0 \otimes \xi^1 - \beta \delta_1 \otimes \xi^1) 
		\mu_{u}^*(dv,ds,d\xi) du
		\\
		&
		+ \int_0^t \int_{\reels^+ \times \mathcal P(E_m)} \alpha\big(I(\xi)\big)
		\Big[ \Psi \Big( 1, 0, \nu^+(\xi) \Big) 
		- \Psi \Big( 0,s, \xi \Big) \Big]
		{\mu_u^*}(0,ds,d\xi) du
		\\
		&
		+ \int_0^t \int_E  
		\partial_{\xi} \Psi ( v,s,\xi ) \cdot 
		[\delta_1 \otimes \nu^{-,1} (s, \xi, \mu_{u}^* ) - 
		\delta_0 \otimes \nu^0 ( \xi, \mu_{u}^* )]
		\ \mu_{u}^*(dv,ds,d\xi) du,
		\end{split}
		\end{align}
		where $\nu^+$, $\nu^{-,1}$ and $\nu^0$ are respectively defined in~\eqref{eq:def-nuplus}, \eqref{eq:def nu_moins_1} and \eqref{eq:defnutilde0}.
	\end{itemize}
	With the same arguments as the one used in Consequence~\ref{cor:egalite-s-mu-xi-star},
	one has that for all $t \geq 0$,
	\[
	\xi_t^*(\cdot,\cdot,\relatifs) = \mu_t^*(\cdot,\cdot,\mathcal P(E_m)).
	\]
	Then, from Main Result~\ref{th:equation limite systeme de particules avec plasticite},
	$\mu_t^*$ is assumed to admit a density $C^1$ in $s$.
	Therefore, the last equation tells us that $\xi_t^*$ admits a density
	$C^1$ in $s$ and by integration by parts in equation \eqref{eq:dynamique mu plasticite}, we have:
	\begin{align*}
	\begin{split}
	\partial_t \xi_{t}^* &= 
	-\partial_s \xi_{t}^*
	+ \beta \left(\delta_0 \otimes {\xi_{t}^*}^1 - \delta_1 \otimes {\xi_{t}^*}^1\right)
	+ \delta_1 \otimes \nu^{-,1} (S_{t}^*,{\xi_{t}^*},{\mu_t^*}) 
	- \delta_0 \otimes \nu^0 ({\xi_{t}^*},{\mu_t^*})
	.
	\end{split}
	\end{align*}
	From this equation and \eqref{eq:gamma-xi}, \eqref{eq:defnutilde0}, \eqref{eq:defnutilde1},
	we thus obtain the following equations on the density functions $s \mapsto {\xi_t^*}(v,s, w)$:
	\begin{align*}
	&\partial_t{\xi_t^*}(0,s,w) = -\partial_s {\xi_t^*}(0,s,w) - \delta_0(s) {\xi_t^*}(0,0,w)
	+ \beta{\xi_t^*}(1,s,w)
	- \int_{ \mathcal{P}(E_m)}
	\alpha\big(I(\xi')\big)
	\frac{{\xi_t^*}(0,s,w)}{{\xi_t^*}(0,s,\relatifs)}
	{\mu_t^*}(0,s,d\xi') ,
	\\
	\\
	&\partial_t{\xi_t^*}(1,s,w) = -\partial_s {\xi_t^*}(1,s,w) - \delta_0(s) {\xi_t^*}(1,0,w)
	- \beta {\xi_t^*}(1,s,w)
	\\
	&
	\qquad\qquad\qquad \ \
	+ \delta_0(s) \int_{\reels^+ \times \mathcal{P}(E_m)}
	\frac{\alpha\big(I(\xi')\big)}{{\xi_{t}^*}(0,s',\relatifs)}
	\Big[
	{p^-(S_t^*,w+1)\xi_{t}^*}(0,s',w+1)
	+(1-p^-(S_t^*,w)){\xi_{t}^*}(0,s',w)
	\Big]
	{\mu_t^*}(0,ds',d\xi')
	.
	\end{align*}
	By integrating the previous equations on $s\in[0,\epsilon]$
	and taking $\epsilon$ to $0$, 
	we obtain the equations with conditions at boundaries $s=0$ of the Main Result~\ref{th:equation limite systeme de particules avec plasticite}.

\section{Code of the Mean-Field equations}\label{ann:ch-4}
We explain in this section our simulations both for the microscopic model, defined through $(V_t^{i,N}, S_t^{i,N}, (W_t^{ij,N})_j)_{i,t \geq 0}$, and the mean field system, $(V_t^{i,*}, S_t^{i,*}, \xi_t^{i,*})_{i,t \geq 0}$. Indeed, we cannot simulate directly the dynamics of $X_t^*$ because we do not have access to its law $\mu_t^*$. Hence, instead we simulate $N$ \textit{typical} neurons $X_t^{1,*},\cdots,X_t^{N,*}$ having the same dynamics as $X_t^*$, except that instead of $\mu_t^*$ we use an approximation of it, $\mu_t^{*,N} = \frac 1N \sum_k \delta_{X_t^{k,*}}$. We start this section with an overview giving the main parameters and describing the way $\xi$ is approximated in the mean field system. Then, we give the initial conditions and finally, we detail how the dynamics are simulated in both cases.

\subsection{Overview}
\subsubsection{Parameters}
First, they both have the same parameters that we list here:
\begin{itemize}
    \item N: number of neurons,
    \item $\alpha$: rate function with the total synaptic current as input (jump of $V$ from $0$ to $1$),
    \item $\beta$: rate of the return to the resting potential (jump of $V$ from $1$ to $0$),
    \item $p^{\pm}$: STDP functions with time delay between spikes as input and probability of synaptic weight jump as output,
    \item $dt$: time step of the simulation,
    \item $t_f$: final time of the simulation,
    \item $p_V$: initial probability of $V$ being in state $1$,
    \item $\rho_{0/1}$: respectively the distribution of the initial $S$ depending on the corresponding $V$ being $0$ or $1$,
    \item $p_W$: initial distribution of the weights.
\end{itemize}

\subsubsection{Approximating \texorpdfstring{$\xi$}{xi}}
The neuron $i$ satisfies the following equations
\begin{itemize}
	\item $X_0^{i,*}$ is determined from the initial state of the microscopic model.
	\item \( d S_t^{i,*} = dt \).
	\item $\xi_t^{i,*}$ admits a density in $s$ and satisfies the following equation: 
	\begin{align}\label{eq:derive-xi-code}
	\hspace{-3em}
	\begin{split}
	\partial_t {\xi_{t}^{i,*}}(0,s,w) &= -\partial_s {\xi_{t}^{i,*}}(0,s,w)
	+ \beta{\xi_{t}^{i,*}}(1,s,w)
	- {\xi_{t}^{i,*}}(0,s,w) \underbrace{ \frac{\frac{1}{N} \sum_k
			\alpha( I(\xi_t^{k,*}) ) \mathbbm 1_{ \{V_t^{k,*} = 0, S_t^{k,*} = s\} } }{
			\frac{1}{N} \sum_l \mathbbm 1_{ \{V_t^{l,*} = 0, S_t^{l,*} = s\} } } }_{a_t^0(s)}
	\\
    \partial_t {\xi_{t}^{i,*}}(1,s,w) &= -\partial_s {\xi_{t}^{i,*}}(1,s,w)
	- \beta {\xi_{t}^{i,*}}(1,s,w)
    \\
	{\xi_{t}^{i,*}}(0,0,w) &= 0,
	\\
	{\xi_{t}^{i,*}}(1,0,w) &= \frac{\frac{dt}{N} \sum_k
		\alpha( I(\xi_t^{k,*}) ) 
		\big[{p^-(S_t^{i,*},w+1)\xi_{t}^{i,*}}(0,S_t^{k,*},w+1) 
		+ (1-p^-(S_t^{i,*},w)){\xi_{t}^{i,*}}(0,S_t^{k,*},w)\big]}{
		\frac 1N \sum_l \mathbbm 1_{ \{V_t^{l,*} = 0, S_t^{l,*} = S_t^{k,*} \} } }
	.
	\end{split}
	\end{align}
    Note that we used $\mu_t^{*,N}$ instead of $\mu_t^*$ present in the limit system of the Main Result \ref{th:equation limite systeme de particules avec plasticite}. 
	\item At rate $\alpha(I(\xi_{t^-}^{i,*})) \mathbbm 1_{ \{V_{t^-}^{i,*} = 0\} } + \beta \mathbbm 1_{ \{V_{t^-}^{i,*} = 1\} }$,
    \begin{itemize}
        \item[$\circ$] $X_{t^-}^{i,*}=(1, S_{t^-}^{i,*}, \xi_{t^-}^{i,*}) \to X_{t}^{i,*}=(0, S_{t^-}^{i,*}, \xi_{t^-}^{i,*})$,
        \item[$\circ$] $X_{t^-}^{i,*}=(0, S_{t^-}^{i,*}, \xi_{t^-}^{i,*}) \to X_{t}^{i,*}=(1, 0, \xi_{t}^{i,*})$ with 
	\[
	\hspace{-4em}
	\xi_{t}^{i,*}(v,A,w) = \int_A p^+(s,w-1) \xi_{t^-}^{i,*}(v,ds,\{w-1\}) 
	+ \int_A (1-p^+(s,w)) \xi_{t^-}^{i,*}(v,ds,w).
	\]
    \end{itemize}
\end{itemize}

We discretise the equations given in \eqref{eq:derive-xi-code} with an explicit Euler scheme to estimate the derivatives. This means that for a function
$u: (t,s) \mapsto u(t,s)$,
\[
\partial_t u(t,s) \to \frac{u(t+dt,s) - u(t,s)}{dt}
\quad \text{and} \quad
-\partial_s u(t,s) \to \frac{u(t,s - dt) - u(t,s)}{dt}.
\]
This transformation from continuous to discrete space is done in the subsection \ref{sec:from_t_to_t_plus_dt}.

We simulated with a fixed $dt$ time step instead of following exactly the jumping times one by one. Then, as $\xi$ variables are distributions (in $s$ and $w$), we use 2D histograms to represent them. For the continuous variable $s$, we used a grid of size $dt$ starting from $0$ and ending in $m_s = M_s dt$ ($m_s = 15$ms in the simulations of Section \ref{ann:ch-4}) with $M_s \in \mathbb{N}$. It is also necessary to choose bounds for synaptic weights, we  denote them by $w_{min},\ w_{max} \in \relatifs$. Therefore, $\xi$ are probability density functions evaluated on a 2D-grid, 
\[
\xi: (m \times dt,w) \in \llbracket 0,  M_s \rrbracket dt \times \llbracket w_{min},  w_{max} \rrbracket \mapsto \xi(m \times dt,w) \in \reels^+.
\]
Finally, we estimate the function $s \mapsto a_t^0(s)$ which is the average firing rate of neurons with time since last spike equals to $s$. 
We used the \texttt{Polynomials.fit} function from the \texttt{Polynomials} package in Julia to approximate $a_t^0$ with a polynomial of order $5$.

\subsection{Initial conditions}
The initial conditions are the same for $V$ and $S$ in both models. We draw them in an independent and identically distributed (iid) way. For $S$ we have to be more careful as the densities $\rho_{0}$ and $\rho_{1}$ have to satisfy the boundary conditions given by the two last equations in \eqref{eq:derive-xi-code}:
\begin{align}\label{eq:init_cond}
    \rho_{0}(0) = 0 \quad \text{and} \quad 
    \rho_{1}(0) = dt\sum\limits_{\tilde{m}=0}^{M_s} \xi_0^{i_0,*}(0,mdt,\relatifs)\frac{\frac{1}{N} \sum_k
			\alpha( I(\xi_t^{k,*}) ) \mathbbm 1_{ \{V_t^{k,*} = 0, S_t^{k,*} = s\} } }{
			\frac{1}{N} \sum_l \mathbbm 1_{ \{V_t^{l,*} = 0, S_t^{l,*} = s\} } }, 
\end{align}
where $i_0$ is a given neuron; we chose $i_0=1$.
Regarding the synaptic weights, the two systems differ. For the microscopic system, we draw the synaptic weights $(W_0^{ij,N})_{i,j}$ from the distribution $p_W$. 
For the mean field system, we truncate $\rho_{0}$ and $\rho_{1}$ at $m_s$ and normalise them.

Here are the steps we followed.
\begin{enumerate}
	\item We draw $(V_0^{i,N})_i=(V_0^{i,*})_i$ as $N$ iid random variables with binomial distribution of parameter $p_V$.
	
	\item For $i$ such that $V_0^{i,*}=0$, we draw $(S_0^{i,N})_i=(S_0^{i,*})_i$ iid with distributions $\rho_{0}$.

    \item We draw the synaptic weights $(W_0^{ij,N})_{i,j}$ iid with distribution $p_W$.
	
	\item We define for all $(m,w)\in \llbracket 0,  M_s \rrbracket \times \llbracket w_{min},  w_{max} \rrbracket$, 
    \[
    \xi_0^{i,*}(0,mdt,w) = \frac{(1-p_V) \rho_0(mdt) p_W(w)}{dt\sum\limits_{\tilde{m}=0}^{M_s}\rho_0(\tilde{m}dt)},
    \]
    where \(\rho_0\) is distributed as a LogNormal distribution with parameters \(\mu=0.8\) and \(\sigma=1\).
    \item  We first compute from \eqref{eq:init_cond} and then we define for all $(m,w)\in \llbracket 0,  M_s \rrbracket \times \llbracket w_{min},  w_{max} \rrbracket$,
    \[    
        \xi_0^{i,*}(1,mdt,w) = \frac{p_V \rho_1(mdt) p_W(w)}{dt\sum\limits_{\tilde{m}=0}^{M_s}\rho_1(\tilde{m}dt)},
    \]
    where $\rho_1$ is an exponential distribution with parameter \(\rho_1(0)\).
\end{enumerate}
We illustrate in Fig. \ref{fig:initialisation} these initial densities that we used in our simulations.
\begin{figure}[!ht]
    \centering
	\subfloat[]{
		\includegraphics[width=0.45\textwidth]{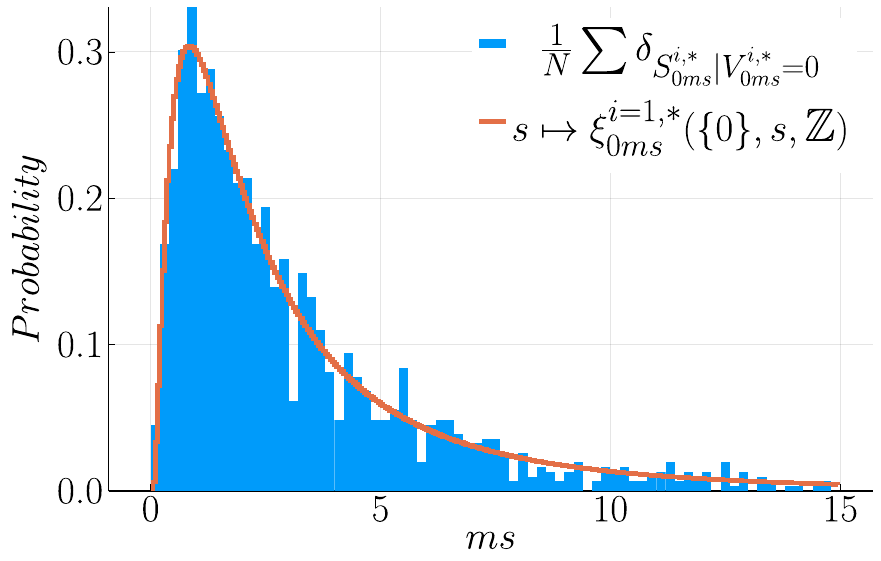}
		\label{sub:distribution-S-V0-initial}}
    \hspace{2em}
	\subfloat[]{\includegraphics[width=0.45\textwidth]{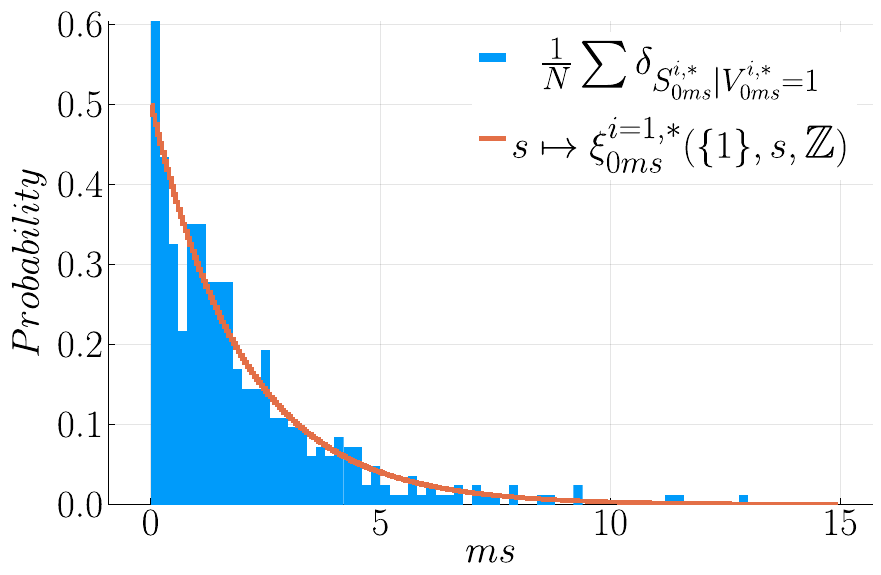}
	\label{sub:distribution-S-V1-initial}}
	\caption{Distributions of the time from last spikes for neurons in state $V=0$ in \eqref{sub:distribution-S-V0-initial} and for neurons in state $V=1$ in \eqref{sub:distribution-S-V1-initial}, at the initial time $t=0\, ms$. Parameters' values of Section \ref{sec:Simulations} are used.}
	\label{fig:initialisation}
\end{figure}

\subsection{Dynamics: from t to t+dt}
\subsubsection{The microscopic system}
We compute the synaptic currents $I_t^{i,N} = \frac{1}{N}\sum_j W_t^{ij,N} V_t^{j,N}$. We draw the jumping times $\tau_t^{i,N}$ of the neurons from exponential distributions with parameter
\[
\alpha(I_t^{i,N})\delta_{\{V_t^{i,N}=0\}} + \beta \delta_{\{V_t^{i,N}=1\}}.
\]
If $\tau_t^{i,N} < dt$ and 
\begin{itemize}
    \item $V_t^{i,N}=1$, then we set $V_{t+dt}^{i,N}=0$ and add $dt$ to the corresponding $S$,
    \item $V_t^{i,N}=0$, then we set $V_{t+dt}^{i,N}=1$, $S_{t+dt}^{i,N}=dt-\tau_t^{i,N}$ and we potentiate (+1) the incoming weights $(W_t^{ij,N})_j$ with probability $p^+(S_t^{j,N},W_t^{ij,N})$ and depress (-1) the outgoing weights $(W_t^{ji,N})_j$ with probability $p^-(S_t^{j,N},W_t^{ji,N})$.
\end{itemize}
If $\tau_t^{i,N} > dt$, we just add $dt$ to the corresponding $S$.

\subsubsection{The mean field system}\label{sec:from_t_to_t_plus_dt}
We describe how to pass from $t$ to $t+dt$ in the mean field system, in particular we have to describe how to discretise equation \eqref{eq:derive-xi-code}.
\begin{enumerate}
	\item We compute the $I_t^{i,*}$ (initially all the same and deterministic).
	\item We draw the jumping times $\tau_t^{i,*}$ from the exponential distributions of parameters
	\[
	\alpha(I_t^{i,*}) \mathbbm 1_{ \{V_t^{i,*} = 0\} } + \beta \mathbbm 1_{ \{V_t^{i,*} = 1\} }.
	\]
	\item We estimate the function $a_t^0 = 
	(a_t^0(0),\cdots, a_t^0(M_sdt))$. To do so, we can approximate this term looking back at the original term which is
    \begin{align}\label{eq:code_approx_current_m}
        a_t^0(mdt)=\frac{\mu_t^*(0,m dt,\mathcal{P}(E_m))}{{\xi_t^{*}}(0,m dt,\mathbb Z)}
        \underbrace{\int_{ \mathcal{P}(E_m)}
        \alpha\big(I(\xi')\big) \mu_t^*(0,m dt,d\xi')}_{\circled{c}}.
    \end{align}
    Regarding the quotient term in~\ref{eq:code_approx_current_m}, it should be equal to one in theory (same marginal law in $S$ of $\mu$ and $\xi$) so we keep it equal to one. Hence, we are left with the term $\circled{c}$ which is nothing but the expectation of the current knowing that the membrane potential is $V=0$ and its time since last spike $s=mdt$, $\espe{\alpha\big(I(\xi_t^*)\big) | V_t^*=0, S_t^*=mdt}$. We illustrate the approximation of $\circled{c}$ in Fig.~\ref{fig:approx_current} where we can see the randomness of it -- even though being an expectation -- as soon as the initial condition is erased: transition at $5\, ms$ in Fig.~\ref{sub:approx_current_5ms} and in Fig.~\ref{sub:approx_current_30ms} nothing deterministic is left when $t_f=30\,ms>m_s$.
    \begin{figure}[ht!]
        \centering
        \subfloat[]{
            \includegraphics[width=0.45\textwidth]{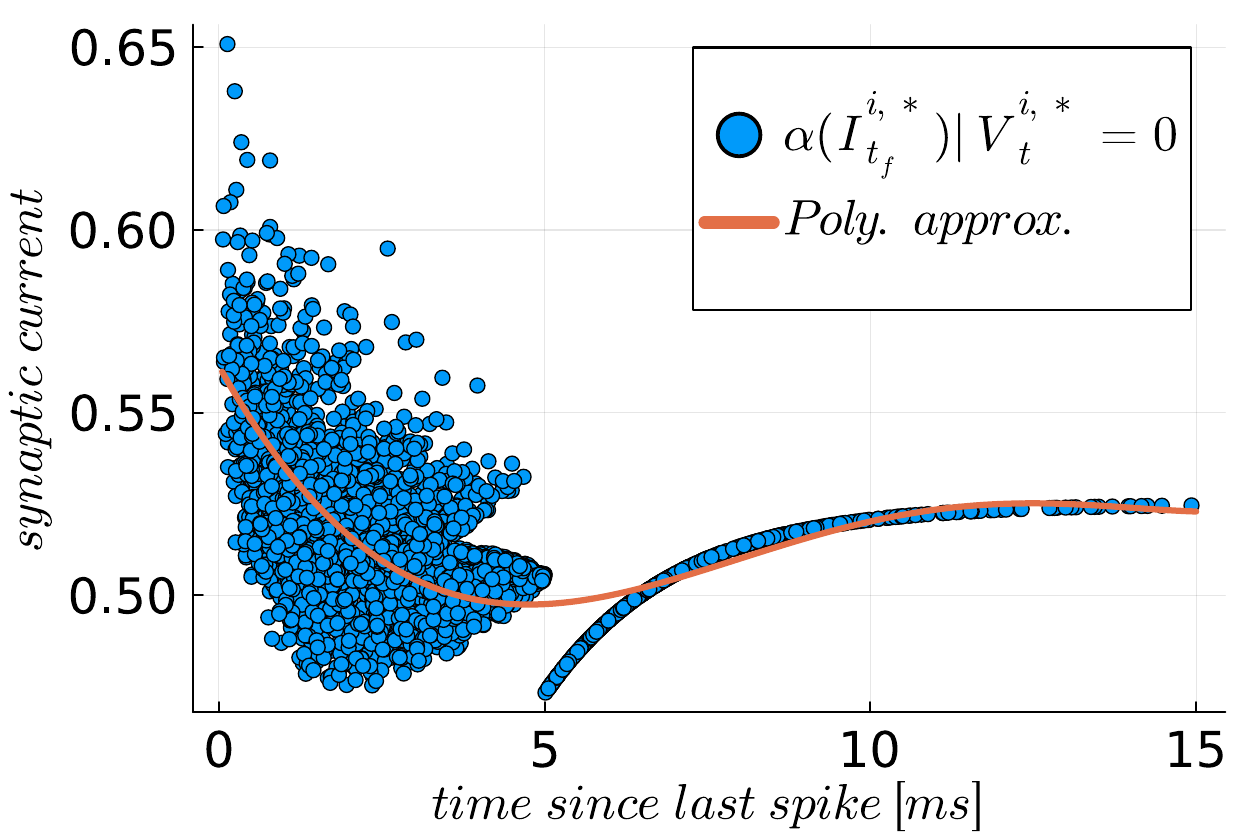}
            \label{sub:approx_current_5ms}}
        \hspace{2em}
        \subfloat[]{\includegraphics[width=0.45\textwidth]{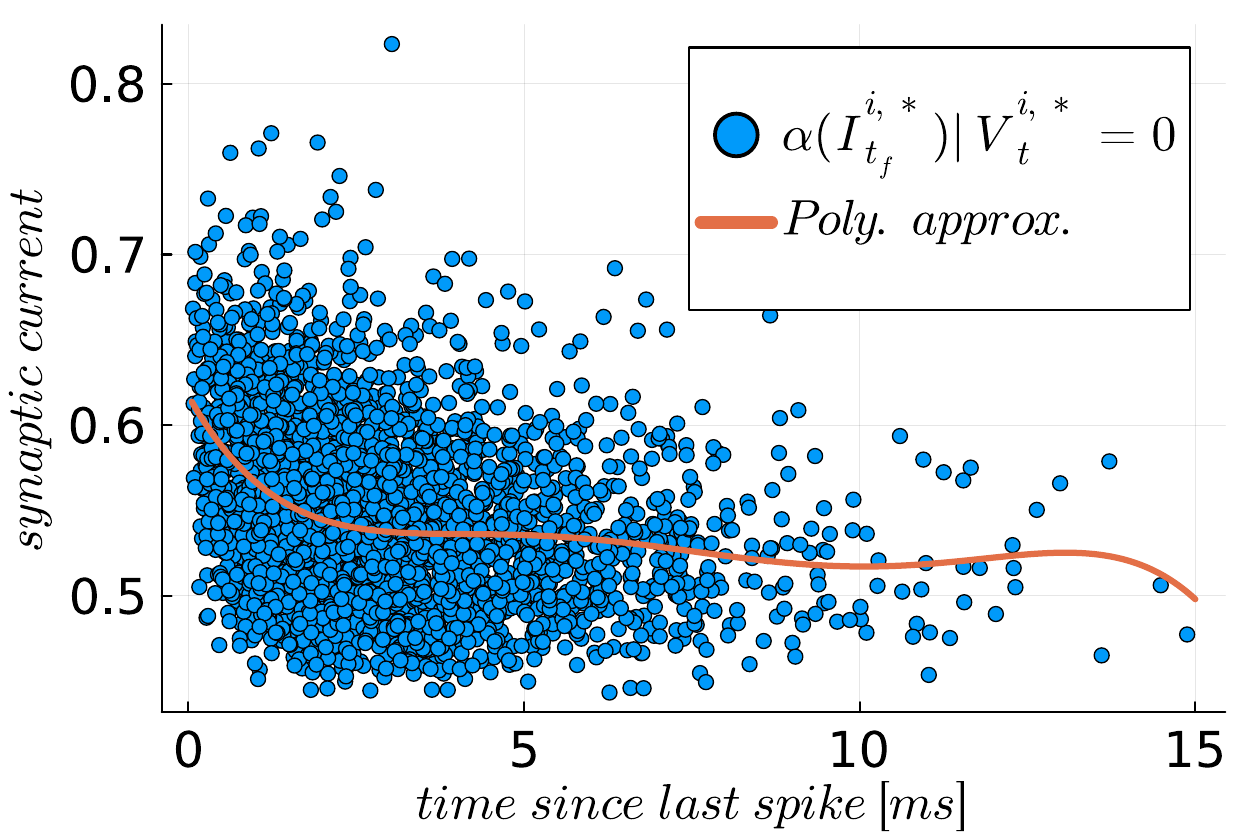}
        \label{sub:approx_current_30ms}}
        \caption{Scatter plot of the synaptic currents from all neurons in state V=0 (blue dots) and the approximation of its expectation (term $\circled{c}$ in \eqref{eq:code_approx_current_m}) at final time $t_f=5ms$ in \ref{sub:approx_current_5ms} and $t_f=30ms$ in \ref{sub:approx_current_30ms}. Parameters' values of Section \ref{sec:Simulations} are used.}
        \label{fig:approx_current}
    \end{figure}
	\item We compute the terms due to the drift of the $\xi$.
	For all $m \in \llbracket 0,M_s \rrbracket$:
	\begin{align*}
	&{\xi_{t+dt}^{i,*}}(0,m dt,w) = 
	{\xi_t^{i,*}}(0,(m - 1)dt,w) \left(1- a_t^0\big((m-1)dt\big)dt \right)
	+ {\xi_t^{i,*}}(1,(m - 1)dt,w) \beta dt
	\\
	&{\xi_{t+dt}^{i,*}}(0,0,w) = 0,
	\\
	\\
	&{\xi_{t+dt}^{i,*}}(1,mdt,w) = 
	{\xi_t^{i,*}}(1,(m - 1)dt,w) \left(1 - \beta dt \right)
	\\
	&{\xi_{t+dt}^{i,*}}(1,0,w) = dt \sum_m  
	a_t^0(mdt) \big[p^-(S_t^{i,*},w+1)\xi_t^{i,*}(0,mdt,w+1) 
    + \big(1 - p^-(S_t^{i,*},w)\big) \xi_t^{i,*}(0,mdt,w)\big]
	\end{align*}
    \item If $\tau_t^{i,*} < dt$ and $V_t^{i,*} = 0$,
    \[
	\xi_{t+dt}^{i,*}(v,mdt,w) = p^+(mdt,w-1)\xi_{t}^{i,*}(v,mdt,w-1) + (1-p^+(mdt,w))\xi_{t}^{i,*}(v,mdt,w).
	\]
	\item We update
	\[
	S_{t+dt}^{i,*} = (S_t^{i,*} + dt) \mathbbm 1_{ \{\tau_t^{i,*} \geq dt\} } + 
	(dt - \tau_t^{i,*}) \mathbbm 1_{ \{\tau_t^{i,*} < dt\} },
	\]
	\item For $i$ such that $\tau_t^{i,*} < dt$, we make the membrane potentials $V^{i,*}$ jump accordingly.
	\item We compute the elements that we keep in memory at each time step : 
	expectations of the potentials, of the time from last the spike, of the pre-synaptic weights, distribution of the synaptic current input.
\end{enumerate}
\end{widetext}

\bibliography{biblio}

\end{document}